\documentclass[prodmode]{acmsmall}
\usepackage[T1]{fontenc}
\usepackage[utf8]{inputenc}
\usepackage{microtype}
\usepackage{srcltx}
\usepackage{booktabs}
\usepackage[pdfborder={0 0 0}]{hyperref}

\renewcommand\tabcolsep{4pt}
\let\accentvec\mathbf

\let\mathbf\accentvec
\usepackage{tabularx}

\usepackage{amsmath}
\usepackage{amssymb}
\usepackage{amsfonts}
\usepackage{mathtools}

\usepackage{algorithm}
\usepackage[noend]{algpseudocode}
 \newcommand{\lv}[1]{\ensuremath{\textnormal{\textrm{level}}(#1)}}

\let\originalleft\left
\let\originalright\right
\renewcommand{\left}{\mathopen{}\mathclose\bgroup\originalleft}
\renewcommand{\right}{\aftergroup\egroup\originalright}

\newcommand{\HH}{\mathcal{H}}
\newcommand{\RR}{\mathcal{R}}

\newcommand{\E}{\textnormal{\textrm{E}}}

\newtheorem{claim}[theorem]{Claim}

\newcommand{\ppv}{\ensuremath{\textnormal{P}^{\textnormal{se}}}}
\newcommand{\pv}{\ensuremath{\textnormal{C}^{\textnormal{se}}}}
\newcommand{\pwa}{\ensuremath{\textnormal{P}^{\textnormal{wa}}}}

\newcommand{\pas}{\ensuremath{\textnormal{P}^{\textnormal{as}}}}

\usepackage{pgfplots}
\newcommand{\Cpp}{C{}\texttt{++}{}}
\usepackage{listings}

\newcommand{\webpage}{\url{http://eiche.theoinf.tu-ilmenau.de/quicksort-experiments/}}

\title{How Good is Multi-Pivot Quicksort?}

\author{Martin Aum\"{u}ller \affil{IT University of Copenhagen} 
Martin Dietzfelbinger \affil{Technische Universität Ilmenau
} Pascal Klaue \affil{3DInteractive GmbH}}

\begin{abstract}
    \emph{Multi-Pivot Quicksort} refers to variants of classical quicksort where
    in the partitioning step $k$ pivots are used to split the input into $k + 1$
    segments. For many years, multi-pivot quicksort was regarded as impractical,
    but in 2009 a 2-pivot
    approach by Yaroslavskiy, Bentley, and Bloch was chosen as the standard sorting algorithm in
    Sun's Java 7. In 2014 at ALENEX, Kushagra et al. introduced an even faster algorithm
    that uses
    three pivots. This paper studies what possible advantages 
    multi-pivot quicksort might offer in general. The contributions 
    are as follows: 
    Natural comparison-optimal algorithms for multi-pivot quicksort are devised and
    analyzed. The analysis shows that the benefits of using multiple pivots with
    respect to 
    the average comparison count are marginal and these strategies
    are inferior to 
    simpler strategies such as the well known median-of-$k$ approach.
    A substantial part of the partitioning cost is caused by 
    rearranging elements. A rigorous analysis of an algorithm 
    for rearranging elements in the partitioning step is carried out, observing mainly 
    how often array cells are accessed during partitioning. The algorithm behaves
    best if 3 to 5 pivots are used. Experiments show that this translates into
    good cache behavior and is closest to predicting observed running times of
    multi-pivot quicksort algorithms. Finally, it is studied how choosing 
    pivots from a sample affects sorting cost.
    The study is theoretical in the sense that although 
    the findings motivate design recommendations for multipivot quicksort algorithms that lead to
    running time improvements over known algorithms in an experimental setting, these improvements are small.
\end{abstract}
\acmformat{Martin Aum\"{u}ller, Martin Dietzfelbinger, and Pascal Klaue. 2016. How
Good is Multi-Pivot Quicksort?}
\begin{document}
\begin{bottomstuff}
    Authors' addresses: M. Aum\"{u}ller, IT University of Copenhagen, Rued Langgaards Vej 7,
    2300 København S, Denmark; e-mail: maau@itu.dk;
    M. Dietzfelbinger, Fakultät für Informatik und Automatisierung, Techni\-sche
    Universität Ilmenau, 98683 Ilmenau, Germany; e-mail: martin.dietzfelbinger@tu-ilmenau.de;
    P. Klaue, 3DInteractive GmbH, 98693 Ilmenau, Germany; email: pklaue@3dinteractive.de. 
		Work carried out while the first author was affiliated with 
		Technische Universität Ilmenau. Part of the work done while 
    the third author was a Master's student at Technische Universität Ilmenau.
\end{bottomstuff}
\maketitle

\section{Introduction}\label{sec:introduction}

Quicksort~\cite{Hoare62} is an efficient standard sorting algorithm with
implementations in practically all algorithm libraries.
Following the divide-and-conquer paradigm, on an input consisting of $n$ elements
quicksort uses a pivot element to
partition its input elements into two parts,
the elements in one part being smaller than or equal to the pivot, the elements in the other
part being larger than or equal to the pivot, and then uses recursion to sort these parts.

In $k$-pivot quicksort, $k$ elements of the
input are picked and sorted to get the pivots $p_1 \leq \cdots \leq p_k$.
Then the task is to partition 
the remaining input according to the $k+1$ \emph{segments} or \emph{groups}
defined by the pivots. Segment $i$, denoted by A$_i$,  consists of elements that are at least 
as large as $p_i$ and at most as large as $p_{i + 1}$, for $1 \leq i \leq k -
1$,
and groups A$_0$ and A$_k$ consist of elements at most as large as $p_1$ and
at least as large as $p_k$, respectively, see Figure~\ref{fig:partition}. 
These segments are then
sorted recursively. As we will explore in this paper, using more than one pivot
allows us to choose from a variety of different partitioning strategies. This
paper will provide the theoretical foundations to analyze these methods.

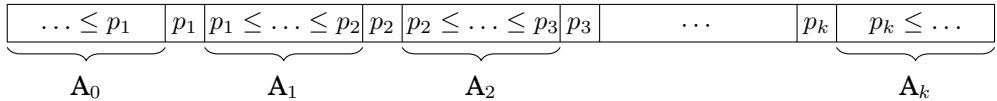
\begin{figure}
\centering
\begin{tikzpicture}[xscale=1.05,x=0.5cm,y=0.5cm]
\draw (0,0)rectangle(25,1);
\draw (4,0)--(4,1);
\draw (5,0)--(5,1);
\draw (9,0)--(9,1);
\draw (10,0)--(10,1);
\draw (14,0)--(14,1);
\draw (15,0)--(15,1);
\draw (21,0)--(21,1);
\draw (20,0)--(20,1);
\node at (2,0.5) {$\ldots \leq p_1$};
\node at (4.5,0.42){$p_1$};
\node at (7.05,0.5) {$p_1 \leq \ldots \leq p_2$};
\node at (9.5,0.42) {$p_2$};
\node at (12.05, 0.5) {$p_2 \leq \ldots \leq p_3$};
\node at (17.5,0.42) {$\cdots$};
\node at (14.5,0.42) {$p_{3}$};
\node at (23,0.5) {$ p_{k}\leq \ldots $};
\node at (20.5,0.42) {$p_{k}$};
\draw [decorate, decoration={brace,amplitude=5pt},yshift=-2pt](4,0)--(0,0) node [midway,yshift=-15pt]{$\text{A}_0$};
\draw [decorate, decoration={brace,amplitude=5pt},yshift=-2pt](9,0)--(5,0) node [midway,yshift=-15pt]{$\text{A}_1$};
\draw [decorate, decoration={brace,amplitude=5pt},yshift=-2pt](14,0)--(10,0) node [midway,yshift=-15pt]{$\text{A}_2$};
\draw [decorate, decoration={brace,amplitude=5pt},yshift=-2pt](25,0)--(21,0) node [midway,yshift=-15pt]{$\text{A}_k$};
\end{tikzpicture}
\caption{Result of the partition step in $k$-pivot quicksort using pivots $p_1, \ldots, p_k$. }
\label{fig:partition}
\end{figure}

\subsection{History and Related Work}

Variants of classical quicksort were the topic of extensive research, such as 
sampling variants \cite{sedgewick,MartinezR01}, variants for equal keys \cite{SedgewickEqual}, or 
variants for sorting strings \cite{BentleyS97}. On the other hand, up to 2009 very little
work had been done on quicksort variants
that use more than one but still a small number of pivots. This is because such
variants of quicksort were judged 
impractical in two independent PhD theses: In~\cite{sedgewick} 
Sedgewick had proposed and analyzed a dual-pivot
approach that was inferior to classical quicksort in terms of the average swap
count.  Later, Hennequin~\cite{hennequin} studied the general
approach of using $k \geq 1$ pivot elements. According to \cite{nebel12}, he
found only slight improvements with respect to the average comparison count that
would not compensate for the more involved partitioning procedure. A particularly 
popular variant of multi-pivot quicksort using a large number of
pivots is samplesort from \cite{FrazerK70}. Already in that paper it was shown that a
multi-pivot approach can be used to achieve an average comparison count very 
close to the lower bound for comparison-based sorting algorithms. Samplesort has
found applications in parallel systems and on GPU's \cite{LeischnerOS10}, and there also
exist well-engineered implementations for standard CPU's \cite{SandersW04}. 

In this paper, we focus on variants with a small constant number of pivots.
While these had been judged impractical for a long time, everything changed in 2009 when a $2$-pivot quicksort
algorithm was introduced as the standard sorting algorithm in Sun's \emph{Java 7}. We will refer
to this algorithm as the ``Yaroslavskiy, Bentley, and Bloch algorithm (YBB algorithm)''. In several previous papers, e.g., \cite{nebel12,Kushagra14,AumullerD15,MartinezNW15}, this algorithm was 
called ``Yaroslavskiy's algorithm'', which was motivated by the 
discussion found at \cite{Mailingliste}. The authors were informed \cite{BlochPersonalCom} 
that the algorithm should be considered joint work by Yaroslavskiy, Bentley, and Bloch. Wild and Nebel (joined
by Neininger in the full version)
\citeyear{nebel12,NebelWM15} analyzed a variant of the YBB algorithm and showed that it uses $1.9 n \ln
n + O(n)$ comparisons and $0.6 n \ln n + O(n)$ swaps on average to sort a random
input if two arbitrary elements are chosen as the pivots. Thus, this $2$-pivot
approach turned out to improve on classical quicksort---which makes $2n \ln n + O(n)$ comparisons and
$0.33..n \ln n + O(n)$ swaps on average---w.r.t. the average comparison count.
However, the swap count was negatively affected by using two pivots, which had
also been observed for another dual-pivot quicksort algorithm in \cite{sedgewick}.   Aumüller and
Dietzfelbinger \citeyear{AumullerD13,AumullerD15} showed
a lower bound of $1.8 n \ln n + O(n)$ comparisons on average for $2$-pivot quicksort algorithms and devised
natural $2$-pivot algorithms that achieved this lower bound. The key to understanding
what is going on here is to note that one can improve the comparison count
by deciding in a clever way with which one of the two pivots a new element should
be compared first. While optimal algorithms with respect to the average comparison count
are simple to implement, they must either count
frequencies or need to sample a small part of the input, which renders them not
competitive with the YBB algorithm with respect to running time when key
comparisons are cheap. Moreover, \citeN{AumullerD15} proposed a $2$-pivot
algorithm which makes $2 n  \ln n + O(n)$ comparisons and $0.6 n \ln n + O(n)$
swaps on average---no improvement over classical quicksort in  both cost measures---,
but behaves very well in practice.  Hence, the running
time improvement of a $2$-pivot quicksort approach could not be explained
conclusively in these works.

Very recently, \citeN{Kushagra14} proposed a novel
$3$-pivot quicksort approach. Their algorithm compares a new element with the
middle pivot first, and then with one of the two others. While the general idea
of this algorithm had been known before (see, e.g., \cite{hennequin,Tan93}),
they provided a smart way of exchanging elements. Building on the work of
\citeN{LaMarcaL99}, they showed theoretically that their
algorithm is more cache efficient than classical quicksort and the YBB
algorithm.  They reported on experiments that gave reason to believe  that the
improvements of multi-pivot quicksort algorithms with respect to running times
are due to their better cache behavior.  They also reported from experiments
with a seven-pivot algorithm, which ran more slowly than their three-pivot
algorithm.  We will describe how their (theoretical) arguments generalize to quicksort
algorithms that use more than three pivots. In connection with the running time
experiments from Section~\ref{sec:experiments}, this allows us to
make more accurate predictions than \cite{Kushagra14} about the influence of
cache behavior to running time. One result of the present study will be that it is not
surprising that their seven-pivot approach is slower, because it has worse
cache behavior than three- or five-pivot quicksort algorithms using a specific
partitioning strategy.

In actual implementations of quicksort and dual-pivot quicksort, pivots are usually
taken from a small sample of elements. For example, the median in a sample
of size $2k + 1$ is the standard way to choose the pivot in classical quicksort.
Often this sample contains only a few elements, say $3$ or $5$. The first theoretical
analysis of this strategy is due to van Emden \citeyear{vanEmden}. \citeN{MartinezR01} settled the exact analysis of the leading term of 
this strategy in 2001.
In practice, other pivot sampling strategies were applied successfully as
well, such as the ``ninther'' variant from \cite{BentleyM93}. In the
implementation of the YBB algorithm in Sun's Java 7, the second-
and fourth-largest element in a sample of size five are chosen as pivots.
The exact analysis of (optimal) sampling strategies for the YBB algorithm
is due to \citeN{NebelWM15}. Interestingly, for the comparison count of
the YBB algorithm it is not optimal to choose as pivots the tertiles
of the sample; indeed, asymmetric choices are superior from a theoretical
point of view. Moreover, it is shown there that---in contrast to classical quicksort
with the median of $2k + 1$ strategy---it is impossible to achieve the lower bound
for comparison-based sorting algorithms using the YBB algorithm. Aumüller
and Dietzfelbinger \citeyear{AumullerD15} later showed that this is not an inherent drawback
of dual-pivot quicksort. Other strategies, such as always comparing with the
larger pivot first, make it possible to achieve this lower bound. For 
more than two pivots, \citeN{hennequin}
was again the first to study how pivot sampling affects the average
comparison count when a ``most-balanced'' comparison tree is used in each classification, 
see \cite[Tableau D.3]{hennequin}.  

\subsection{Contributions}

The main contributions of the present paper are as follows: 

(i) In the style of
\cite{AumullerD15}, we study how the average comparison count of an arbitrary
$k$-pivot quicksort algorithm can be calculated. Moreover, we show a lower bound
for $k$-pivot quicksort and devise natural algorithms that achieve this lower
bound.  It will turn out that the partitioning procedures become complicated and
the benefits obtained by minimizing the average comparison count are only
minor.  In brief, optimal $k$-pivot quicksort cannot improve on simple and
well-studied strategies such as classical quicksort using the median-of-$k$
strategy. Compared with the study of $2$-pivot algorithms in \cite{AumullerD15},
the results generally carry over to the case of using $k \geq 3$ pivots.  However,
the analysis becomes more involved, and we were not able to prove tight asymptotic bounds
as in the $2$-pivot case. The interested reader is invited to peruse
\cite{AumullerD15} to get acquainted with the ideas underlying the general analysis.  

(ii) Leaving key
comparisons aside, we study the problem of rearranging the elements to
actually partition the input. We devise a natural generalization of the
partitioning algorithms used in classical quicksort, the YBB algorithm, and
the three-pivot algorithm of \cite{Kushagra14} to solve this
problem. The basic idea is that as in classical quicksort there exist
two pointers which scan the array from left to right and right to left, respectively,
and the partitioning process stops when the two pointers meet. Misplaced elements
are moved with the help of $k-1$
additional pointers that store starting points of
special array segments.  We study this algorithm with regard to three 
cost measures: (a) the average
number of scanned elements (which basically counts how often array 
cells are accessed during partitioning, see Section~\ref{sec:assignments} or \cite{NebelWM15} for a precise definition), 
(b) the average number of writes into array cells, and (c) the average number of assignments necessary 
to rearrange the elements. 
Interestingly, while moving elements around becomes more complicated during
partitioning, this algorithm scans fewer array cells than classical quicksort
for certain (small) pivot numbers.  We will see that $3$- and $5$-pivot quicksort algorithms visit the
fewest array cells, and that this translates directly into good cache behavior
and corresponds to differences in running time in practice. In brief, we provide strong evidence that the
running time improvements of multi-pivot quicksort are largely due to its better
cache behavior (as conjectured by \citeN{Kushagra14}), and
that no benefits are to be expected from using more than five pivots. 

(iii) We analyze
sampling strategies for multi-pivot quicksort algorithms with respect to comparisons and
scanned elements. We will show that for each fixed order in which elements are compared to pivots
there exist pivot choices which yield a comparison-optimal multi-pivot quicksort algorithm. When
considering scanned elements there is one optimal pivot choice.  Combining comparisons and scanned elements, the analysis provides a
candidate for the order in which elements should be compared to pivots that
has not been studied in previous attempts like \cite{hennequin,Iliopoulos14}.

\paragraph{Advice to the Practitioner} We stress that this paper is a
theoretical study regarding the impact of using more than one pivot in the quicksort
framework, which sheds light also on the limitations of this approach. 
Still, we believe that our theoretical findings lead to some interesting design
recommendations for multi-pivot quicksort algorithms. These findings are summarized 
in Section~\ref{sec:experiments} on experimental evaluation, where they are contrasted with
tests regarding empirical running times. We demonstrate that for random permutations
one obtains sorting algorithms faster than other approaches such as the
YBB algorithm \cite{WildNN15} or the three-pivot algorithm of
\citeN{Kushagra14}. While the differences in running time are statistically significant, 
the relative difference in running time is usually not more than 5\%. 
A lot of care has to be taken in making ``library-ready'' implementations of
these algorithms in terms of their performance on different input
types, especially with a focus on handling equal keys in the input. 
This must remain the objective of future work.

\subsection{Outline}
In the analysis of quicksort, the analysis of one particular partitioning
step with respect to a specific cost measure, e.g., the
number of comparisons (or assignments, or array accesses),
makes it possible to precisely analyze the cost over the whole recursion.
In Hennequin's thesis~\citeyear{hennequin} the connection
between partitioning cost and overall cost for quicksort variants with more than
one pivot has been analyzed in detail.  The result relevant for us is that if $k$
pivots are used and the (average) partitioning cost for $n$ elements is $a\cdot n+O(1)$, for a constant $a$, then the average cost for
sorting $n$ elements is \begin{equation}\label{eq:1} \frac{1}{H_{k+1}
- 1} \cdot a\cdot n\ln n + O(n), \end{equation} where $H_{k+1}$
denotes the $(k+1)$st harmonic number. In Section~\ref{sec:setup}, we will
use the continuous Master theorem from \cite{Roura01} to prove a more general
result for partitioning cost $a \cdot n + O(n^{1-\varepsilon})$. Throughout the present
paper all that interests us is the constant factor with the leading term.
(Of course, for a realistic input size $n$ the lower order term can have a big influence
on the cost measure.)

For the purpose of the analysis, we will consider the input to be a random
permutation of the integers $\{1,\ldots,n\}$. An element $x$ belongs to
group $\textrm{A}_i, 0 \leq i \leq k,$ if $p_i < x < p_{i + 1}$ (see
Fig.~\ref{fig:partition}), where we set $p_0 = 0$ and $p_{k
+ 1} = n + 1$.
When focusing on a specific cost measure we can often leave aside certain
aspects of the partitioning process. For example, in the study of the average
comparison count of an arbitrary $k$-pivot algorithm, we will only
focus on \emph{classifiying} the elements into their respective groups
$\text{A}_0,\ldots,\text{A}_k$, and omit rearranging these elements to produce
the actual partition. On the other hand, when focusing on the average swap count (or the average
number of assignments necessary to move elements around), we might assume
that the input is already classified and that the problem is to rearrange
the elements to obtain the actual partition.

\begin{figure}
    \centering
    \begin{tikzpicture}[scale=0.8, level/.style = {sibling distance = 5.5cm/#1, level distance=0.85cm}]
	\tikzset {
	    treenode/.style = {align = center, inner sep = 3pt, text centered},
	    n/.style = {treenode, circle, white, draw = black, text = black},
	    leaf/.style = {treenode, rectangle, draw = black}
	}
	\node [n] {$\text{p}_2$}
	    child {
		node [n] {$\text{p}_1$}
		    child {
			node [leaf] {$\text{A}_0$}
		    }
		    child {
			node [leaf] {$\text{A}_1$}
		    }
		}
	    child {
		node [n] {$\text{p}_4$}
		    child {
			node [n] {$\text{p}_3$}
			    child {
				node [leaf] {$\text{A}_2$}
			    }
			    child {
				node [leaf] {$\text{A}_3$}
			    }
			}
		    child {
			node [n] {$\text{p}_5$}
			    child {
				node [leaf] {$\text{A}_4$}
			    }
			    child {
				node [leaf] {$\text{A}_5$}
			    }
			}
		    };

    \end{tikzpicture}
    \caption{A comparison tree for five pivots.}
    \label{fig:comparison:tree}
\end{figure}
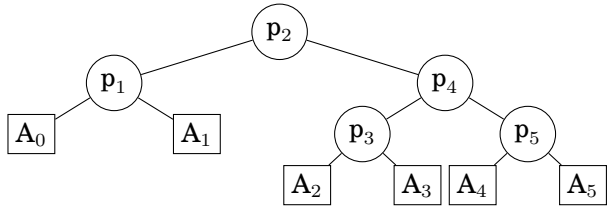

In terms of classifying the elements into groups $\text{A}_0,\ldots,
\text{A}_k$, the most basic operation is the classification of a single element.
This is done by comparing the element against the pivots in some order. This
order is best visualized using a comparison tree, which is a binary search tree with
$k+1$ leaves labeled $\text{A}_0, \ldots, \text{A}_k$ from left to right and $k$
inner nodes labeled $\text{p}_1,\ldots,\text{p}_k$ according to inorder
traversal. (Such a tree for $k=5$ is depicted in
Figure~\ref{fig:comparison:tree}.) Assume a comparison tree $\lambda$ is given,
and pivots $p_1,\ldots, p_k$ have been chosen. Then a given non-pivot element
determines a search path in $\lambda$ in the usual way; its
classification can be
read off from the leaf at the end of the path. If the input contains
$a_h$ elements of group $\text{A}_h$, for $0 \leq h \leq k$, the \emph{cost} $\text{cost}^\lambda(a_0,\ldots,a_k)$  of
a comparison tree $\lambda$  is the sum over all
$j$ of the depth of the leaf labeled $\text{A}_j$ multiplied with $a_j$, i.e.,
the total number of comparisons made when classifying the whole input using $\lambda$.
A classification algorithm then defines which comparison tree is to be
used for
the classification of an element based on the outcome of the previous
classifications. The first main result of our paper---presented in
Section~\ref{sec:act}---is that in order to (approximately) determine
the average comparison count for partitioning given $p_1, \ldots, p_k$ we only have to find out how many times on average each comparison tree is used by the
algorithm. The average comparison count up to lower order terms 
is then the sum over all trees of the product of the average number of times a tree is used and its average cost. 
Averaging over all pivot choices then gives the average comparison count
for the classification. Section~\ref{sec:3:pivot:quicksort} applies this result 
by discussing different classification strategies for 3-pivot quicksort.

In Section~\ref{sec:optimal:strategies}, we will show that there exist two very natural
comparison-optimal strategies. The first strategy counts the number of
elements $a'_0,\ldots,a'_k$ classified to groups $\text{A}_0, \ldots,
\text{A}_k$, respectively, after the first $i$ classifications. The
comparison tree used in the $(i+1)$st classification is then one
with minimum cost w.r.t. $(a'_0,\ldots,a'_k)$. The second strategy
uses an arbitrary comparison tree for the first $n^{3/4}$ classifications,
then computes a cost-minimal comparison tree for the group sizes seen in that sample, and uses
this tree in each of the remaining classifications.

A full analysis of optimal versions of $k$-pivot quicksort for $k \geq 4$
remains open. In Section~\ref{sec:comparison:median:of:k},
we resort to estimates for the cost of partitioning based on experiments to
estimate coefficients for average comparison counts for larger $k$. The results
show that the improvements given by
comparison-optimal $k$-pivot quicksort can be achieved in much simpler ways,
e.g., by combining classical quicksort with the median-of-$k$ pivot sampling
technique. Moreover, while choosing an optimal comparison tree for fixed
segment sizes is a simple application of dynamic programming, for large $k$ the
time needed for the computation renders optimal $k$-pivot quicksort useless with respect
to running time.

Beginning with Section~\ref{sec:assignments}, we will follow a different approach,
which we hope helps in understanding factors different from comparison counts that
determine the running time of multi-pivot quicksort algorithms. We
restrict ourselves to use some fixed comparison tree for each
classification,
and think only about moving elements around in a swap-efficient or
cache-efficient way. At the first glance, it is not clear why more pivots should help.
Intuitively, the more segments we have, the more work
we have to do to move elements to the segments, because we are much
more restrictive on where an element should be placed. However, we save work in the
recursive calls, since the denominator in \eqref{eq:1} gets larger as we use more pivots.
(Intuitively, the depth of the recursion tree
decreases.) So, while the partitioning cost
increases, the total sorting cost could actually decrease.
We will devise an algorithm that builds upon the
``crossing pointer technique'' of classical quicksort.
In brief, two pointers move towards each other as in classical quicksort.
Misplaced elements are directly moved to some temporary array segment which holds elements
of that group. This is done with the help of additional pointers.
Moving misplaced elements is in general not done using swaps,
but rather by moving elements in a cyclic fashion. Our cost measure for this algorithm is
the number of \emph{scanned elements}, i.e., the sum over all array cells of how
many pointers accessed this array cell during partitioning and sorting.
For the average number of scanned elements,
it turns out that there is an interesting balance between
partitioning cost and total sorting cost. In fact, in this cost measure, the average
cost drastically decreases from using one pivot to using three pivots, there
is almost no difference between $3$- to $5$-pivot quicksort, and for larger pivot numbers the
average cost increases again. Interestingly, with respect to two other cost measures that look quite similar we get higher cost as the number of pivots increases.

In Section~\ref{sec:pivot:sampling} we turn our attention to the effect of choosing pivots from a (small) sample of elements.
Building on the theoretical results regarding comparisons and scanned elements from before, it is rather easy
to develop formulae to  calculate
the average number of comparisons and the average number of scanned elements
when pivots are chosen from a small sample. Example calculations demonstrate
that the cost in both measures can be decreased by choosing
pivots from a small (fixed-sized) sample. Interestingly, the best pivot choices do not
balance subproblem sizes but tend to make the middle groups
larger. To get an idea what optimal sampling strategies should look like, 
we consider the setting that we can choose pivots of a given rank for free, and we are interested
in the ranks that minimize the specific cost measure. Our first result in this
setting shows that for every fixed comparison tree it is possible to choose the
pivots in such a way that on average we need at most $1.4426.. n\ln n + O(n)$
comparisons to sort the input, which is optimal. As a second result, 
we identify a particular pivot choice that minimizes
the average number of scanned elements. 

At the end of this paper, we report on experiments carried out to find if the theoretical
cost measures are correlated to observed running times in practice. To this end, we 
implemented $k$-pivot quicksort variants for many different pivot numbers and
compared
them with respect to their running times. In brief, these experiments will
confirm what has been conjectured in \cite{Kushagra14}: running times of quicksort algorithms are
best predicted using a cost measure related to cache misses in the CPU. 

\section{Setup and Groundwork}\label{sec:setup}
We assume that the input is a random permutation $(e_1,\ldots,e_n)$
of $\{1,\ldots,n\}$. If $n \leq k$, sort the input directly. For $n > k$,
sort the first $k$ elements such that $e_1 < e_2 <
\ldots < e_k$ and set $p_1 = e_1,\ldots,p_k = e_k$. In the \emph{partition
step}, the remaining $n-k$ elements are split into $k+1$ \emph{groups}
$\text{A}_0,\ldots,\text{A}_k$, where an element $x$ belongs to group $\text{A}_h$ if $p_{h} < x <
p_{h+1}$. (For the ease of discussion, we set  $p_0 = 0$ and
$p_{k+1} = n + 1$.) The groups $\text{A}_0,\ldots, \text{A}_k$ are then
sorted recursively. We never compare two non-pivot elements against each
other. This preserves the randomness in the groups $\text{A}_0,\ldots,
\text{A}_k$. In the remainder of this paper, we identify group sizes by $a_i := |A_i| = p_{i+1} - p_{i} -  1$ for $ i \in \{0,\ldots,k\}$.
In the first sections, we focus on analyzing the average comparison count.
Let $k \geq 1$ be fixed. Let $C_n$ denote the random variable which counts
the comparisons being made when sorting an input of length $n$, and let $P_n$ be
the random variable which counts the comparisons made in the first
partitioning step. The average comparison count of $k$-pivot quicksort
clearly obeys the following recurrence, for $n \geq k$:
\begin{align*}
\E(C_n) &= \E(P_n) + \frac{1}{\binom{n}{k}} \sum_{a_0 + \cdots + a_k = n - k}
\left(\E(C_{a_0}) + \cdots + \E(C_{a_k})\right).
\end{align*}
For $n < k$ we assume cost 0. We now collect terms with a common factor $\E(C_\ell)$, for $0 \leq \ell \leq n - k$.
To this end, fix $j \in \{0,\ldots,k\}$ and $ \ell \in \{0, \ldots, n-k\}$ and assume that
$a_j = \ell$. There are exactly $\binom{n - \ell - 1}{k - 1}$ ways to choose the
other segment sizes $a_i, i \neq j,$ such that $a_0 + \cdots + a_k = n - k$. (Note the equivalence between
segment sizes and binary strings of length $n - \ell - 1$ with exactly $k - 1$ ones.) Thus, we conclude that
\begin{align}
\E(C_n) &= \E(P_n) + \frac{k + 1}{\binom{n}{k}} \sum_{\ell = 0}^{n-k} \binom{n-\ell-1}{k-1} \E(C_\ell), \label{eq:k:pivot:recurrence}
\end{align}
which was also observed in \cite{Iliopoulos14}. (This generalizes the well known
formula $\E(C_n) = n - 1 + 2/n \cdot \sum_{0 \leq \ell \leq n - 1} \E(C_\ell)$ for
classical quicksort and the formulas for $k = 2$ from, e.g.,
\cite{AumullerD15,nebel12} and $k=3$ from \cite{Kushagra14}.) For partitioning cost 
of $a\cdot n + O(n^{1-\varepsilon})$, for constants $a$ and $\varepsilon > 0$, this recurrence has the following solution. 
\begin{theorem}
    Let $\mathcal{A}$ be a $k$-pivot quicksort algorithm that for each
    subarray of length $n$ has partitioning
    cost $\E(P_n) = a \cdot n + O(n^{1-\varepsilon})$ for a constant $\varepsilon > 0$. Then
    \begin{align}\label{eq:k:pivot:recurrence:solution}
	\E(C_n) = \frac{1}{H_{k+1} - 1} \cdot a n \ln n + O(n),
    \end{align}
    where $H_{k+1} = \sum_{i = 1}^{k+1} (1/i)$ is the $(k + 1)$st harmonic
    number.
    \label{thm:k:pivot:recurrence:solution}
\end{theorem}
\begin{proof}
    By linearity of expectation we may solve the recurrence for 
    partitioning cost $
    \E(P_{1,n}) = a \cdot n$ and $\E(P_{2,n}) = O(n^{1 - \varepsilon})$ separately. To solve for cost $\E(P_{1,n})$  we may 
    apply \eqref{eq:1}. The recurrence for partitioning cost $\E(P_{2,n})$ 
    is a standard application of the continuous Master theorem from
    \cite{Roura01}, see
    Appendix~\ref{app:sec:recurrence:solution} for details. \qed
\end{proof}
When focusing only on the average comparison count, it suffices to study the
\emph{classification problem}: Given a
random permutation $(e_1,\ldots,e_n)$ of $\{1,\ldots,n\}$, choose the pivots
$p_1,\ldots,p_k$ and classify each of the
remaining $n-k$ elements as belonging to one of the groups $\text{A}_0,\ldots, \text{A}_k$. 

Algorithmically, the classification of a single element $x$ with respect to the
pivots $p_1,\ldots,p_k$ is done by using a \emph{comparison tree} $\lambda$.
A comparison tree is a binary search tree, where the leaf nodes are labeled $\text{A}_0, \ldots, \text{A}_k$
from left to right and the inner nodes are labeled $\text{p}_1, \ldots,
\text{p}_k$ in inorder.  Figure~\ref{fig:comparison:tree} depicts a comparison
tree for five pivots. 
We denote the depth of the A$_h$ leaf in comparison tree $\lambda$ by $\text{depth}_{\lambda}(\text{A}_h)$. 
Classifying an element then means searching for this element in the search tree.
The classification of the element is the label of the leaf reached in that way; the number of 
comparisons required is $\text{depth}_{\lambda}(\text{A}_h)$ if $x$ belongs to group $\text{A}_h$.

A \emph{classification strategy} is formally described as a \emph{classification tree} as follows.
A classification tree is a $(k+1)$-way tree with a root and $n-k$
levels of inner nodes as well as one leaf level. Each inner node $v$ has two
labels: an index $i(v) \in \{k+1,\ldots,n\}$, and a comparison tree
$\lambda(v)$. The
element $e_{i(v)}$ is classified using the comparison tree $\lambda(v)$. The $k+1$ edges out of a node are
labeled $0,\ldots,k$, resp., representing the outcome of the
classification as belonging to group $\text{A}_0,\ldots,\text{A}_k$,
respectively. On each of the $(k+1)^{n-k}$ paths each index from $\{k + 1, \ldots, n\}$
occurs exactly once.
An input $(e_1,\ldots,e_n)$ determines a path in the classification tree in the
obvious way: sort the pivots,
then use the classification tree to classify $e_{k+1},\ldots,e_n$.
The classification of the input can then be read off from the nodes and edges
along the path from the root to a leaf in the classification tree.

To fix some more notation, for each node $v$, and for $h \in \{0,\ldots,k\}$, we
let $a^v_h$ be the number of edges labeled ``$h$'' on the path from the root to
$v$.  Furthermore, let $C_{h,i}$ denote the random variable which counts the
number of elements classified as belonging to group $\text{A}_h$, for $h
\in \{0,\ldots k\}$, in the first $i$ levels, for $i \in
\{0,\ldots,n-k\}$, i.e., $C_{h,i} = a^v_h$ when $v$ is the node on level $i$ of
the classification tree reached for an input. In many proofs, we will need
that $C_{h,i}$ is not far away from its expectation $a_h / (n - k - i)$. This would
be a trivial consequence of the Chernoff bound if the classification of elements
were independent. However, the
probabilities of classifying elements to a specific group change according to classifications made before.
We will use the \emph{method of averaged bounded differences} to show concentration despite dependencies between
tests.

\begin{lemma}
    Let the pivots $p_1,\ldots,p_k$ be fixed. Let $C_{h,i}$ be defined as above.
    Then for each $h$ with $h \in \{0,\ldots,k\}$ and for each $i$ with $1 \leq i \leq n - k$ we have that
    \begin{align*}
        \Pr\left( \left\vert C_{h,i} - \E(C_{h,i})\right\vert > n^{2/3} \right) \leq 2\textnormal{\text{exp}}\left(-n^{1/3}/2\right).
    \end{align*}
    \label{lem:sample:concentration:k:pivots}
\end{lemma}
\begin{proof}
    Fix an arbitrary $h \in \{0,\ldots,k\}$.
    Define the indicator random variable
    \begin{align*}
        X_j = [\text{the element classified in level $j$ belongs to group $\text{A}_h$}].
    \end{align*}
    Of course, $C_{h,i} = \sum_{1 \leq j \leq i}X_j.$
    We let
        $c_j := \left\vert \E\left( C_{h,i} \ \mid X_1, \ldots, X_j \right) - \E\left( C_{h,i} \mid X_1, \ldots, X_{j-1} \right)\right\vert.$
%
    Using linearity of expectation we may calculate
    \begin{align*}
	c_j &= \bigg|\sum_{k = 1}^j X_k + \sum_{k = j+1}^{i} \frac{a_h - C_{h,j}}{n - j - 2} - \sum_{k = 1}^{j-1} X_k - \sum_{k = j}^{i}
	\frac{a_h - C_{h,j-1}}{n - j  - 1}\bigg|\\
	    &= \bigg|X_j + \sum_{k = j+1}^{i} \frac{a_h - C_{h,j - 1} - X_j}{n - j - 2}  - \sum_{k = j}^{i}
	\frac{a_h - C_{h,j-1}}{n - j  - 1}\bigg|\\
	&= \bigg| X_j - X_j \cdot \frac{i - j}{n-j - 2} + (a_h - C_{h,j - 1})
	\left(\frac{i-j}{n - j -2} - \frac{i - j + 1}{n - j - 1}\right) \bigg|\\
	&= \bigg| X_j \left(1 - \frac{i - j}{n-j - 2}\right) - (a_h - C_{h,j - 1})
	\left(\frac{n - i - 2}{(n - j - 1)(n-j-2)}\right) \bigg|\\
	&\leq \max\left\{\bigg| X_j \left(1 - \frac{i - j}{n-j - 2}\right)\bigg |, \bigg| \frac{a_h - C_{h, j- 1}}
{n - j - 1} \bigg|\right\} \leq 1.
    \end{align*}
    Now we may apply the following bound known as the method of averaged bounded differences (see
    \cite[Theorem 5.3]{dp09}):
    \begin{align*}
        \Pr(|C_{h,i} - \E(C_{h,i})| > t) \leq 2\exp\left(-\frac{t^2}{2\sum_{j \leq
	i}c_j^2}\right).
    \end{align*}
    This yields
    \begin{align*}
            \Pr\left(|C_{h,i} - \E(C_{h,i})| > n^{2/3}\right) \leq 2\exp\left(\frac{-n^{4/3}}{2i}\right),
        \end{align*}
        which is not larger than $2\exp(-n^{1/3}/2)$.\qed
    \end{proof}

    \section{The Average Comparison Count For Partitioning}\label{sec:act}
    In this section, we will obtain a formula for the average comparison count
    of an arbitrary classification strategy.
    We make the following observations for all classification strategies:
    We need $k \log k = O(1)$ comparisons to sort $e_1,\ldots,e_k$, i.e., to
    determine the $k$ pivots $p_1,\ldots,p_k$ in order.
    If an element $x$ belongs to group $A_i$, it must be compared to
    $p_{i}$ and $p_{i+1}$. (Of course, no real comparison takes place against $p_0$ and
    $p_{k+1}$.)
    On average, this leads to $2(1 - 1/(k+1))(n-k) + O(1)$
    comparisons---regardless of the actual classification strategy.

    For the following paragraphs, we fix a classification strategy, i.e., a
    classification tree $T$. Furthermore, we let $v$ be an
    arbitrary inner node of $T$. 

    If $e_{i(v)}$ belongs to group $\text{A}_h$ then exactly
    $\text{depth}_{\lambda(v)}(\text{A}_h)$ comparisons are made to classify this element.
    %
    We let $X^T_v$ denote the number of
    comparisons that take place in node $v$ during classification.    %
    Let $P^T_n$ be the random
    variable that counts the number of comparisons being made when classifying
    an input sequence $(e_1,\ldots,e_n)$ using $T$, i.e., $P^T_n  = \sum_{v \in T} X^T_v$. For
    the average comparison count $\E(P^T_n)$ for partitioning we get:

    \begin{align*}
        \E(P^T_n) &=
        \frac{1}{\binom{n}{k}} \sum_{1 \leq p_1 < p_2 < \cdots < p_k \leq n} \E(P^T_n
        \mid p_1,\ldots,p_k).
    \end{align*}
    We define $p^v_{p_1,\ldots,p_k}$ as the probability that node $v$ is reached
    if the pivots are $p_1,\ldots,p_k$. We may write:
    \begin{align}
        \E(P^T_n \mid p_1,\ldots,p_k) &= \sum_{v \in T} \E(X_v^T \mid
        p_1,\ldots,p_k)\notag\\
        &= \sum_{v \in T} p^v_{p_1,\ldots,p_k} \cdot \E(X_v^T \mid p_1, \ldots, p_k, \text{ $v$ reached}).
        \label{eq:51000}
    \end{align}
    For a comparison tree $\lambda$ and group sizes $a'_0, \ldots, a'_k$, we define the \emph{cost} of
    $\lambda$ on these group sizes as the number of comparisons it makes for classifying an input with
    these group sizes, i.e.,
    \begin{align*}
        \text{cost}^\lambda(a'_0, \ldots, a'_k) = \sum_{0 \leq i \leq k}
        \text{depth}_{\lambda}(\text{A}_i) \cdot a'_i.
    \end{align*}
    Furthermore, we define its average cost
    $c_{\text{avg}}^{\lambda}(a'_0,\ldots,a'_k)$ as follows:
    \begin{align}\label{eq:30014}
        c_{\text{avg}}^{\lambda}(a'_0,\ldots,a'_k) := \frac{\text{cost}^\lambda(a'_0,\ldots,a'_k)}{\sum_{0 \leq i \leq k} a'_i}.
    \end{align}
    Under the assumption that node $v$ is reached and that the pivots are $p_1,\ldots,p_k$, the
    probability that the element $e_{i(v)}$ belongs to group $\text{A}_h$ is exactly
    $(a_h - a^v_h) / (n - k - \text{level}(v))$, for each $h \in \{0, \ldots,
    k\}$. (Note that this means that the order in which elements are classified is arbitrary, so that we could actually use some fixed ordering.) Summing over all groups, we get
    \begin{align*}
        \E(X_v^T \mid p_1, \ldots, p_k, \text{ $v$ reached}) = c_{\text{avg}}^{\lambda(v)}(a_0 - a^v_0, \ldots, a_k - a^v_k).
    \end{align*}
    Plugging this into \eqref{eq:51000} gives
    \begin{align}\label{eq:50000}
        \E(P^T_n \mid p_1,\ldots,p_k) &= \sum_{v \in T} p^v_{p_1,\ldots,p_k} \cdot c_{\text{avg}}^{\lambda(v)}(a_0 -
        a^v_0,\ldots,a_k - a^v_k).
    \end{align}
    Let $\Lambda_k$ be the set of all possible comparison trees. For each $\lambda
    \in \Lambda_k$, we define the random variable
    $F^\lambda$ that counts the number of times $\lambda$ is used during
    classification. For given $p_1,\ldots,p_k$, and for each $\lambda \in \Lambda_k$, we let
    \begin{align*}
        f^{\lambda}_{p_1,\ldots,p_k} := \E(F^\lambda \mid p_1,\ldots,p_k) =
        \sum_{\substack{v \in T\\\lambda(v) = \lambda}} p^v_{p_1,\ldots,p_k}
    \end{align*}
    denote the average number of times comparison tree $\lambda$ is used in $T$ under the condition
    that the pivots are $p_1, \ldots, p_k$.

    Now, if it was decided in each step by independent random experiments with
    the correct expectation $a_h/(n-k)$, for $0 \leq h \leq k$, whether an
    element belongs to group $\text{A}_h$, it would be clear that for each $\lambda
    \in \Lambda_k$ the contribution of $\lambda$ to the average classification
    cost is $f^\lambda_{p_1,\ldots,p_k} \cdot c_{\text{avg}}^\lambda(a_0,\ldots,a_k)$. This intuition can
    be proven to hold for all classification trees, except that one gets an
    additional $O(n^{1-\varepsilon})$ term due to dependencies between classifications.

    \begin{lemma} Let the pivots $p_1,\ldots,p_k$ be fixed.
        Let $T$ be a classification tree. Then there exists a constant $\varepsilon > 0$ such that
        \begin{align*}
            \E(P^T_n) =  \frac{1}{\binom{n}{k}} \sum_{1 \leq p_1 < p_2 < \cdots < p_k \leq n}
            \quad\sum_{\lambda \in \Lambda_k}
            f^{\lambda}_{p_1,\ldots,p_k} \cdot c_{\textnormal{avg}}^\lambda(a_0,\ldots,a_k) +
            O(n^{1-\varepsilon}).
        \end{align*}
        \label{lem:k:pivot:average:partition:cost}
    \end{lemma}

    \begin{proof}
        Fix the set of pivots $p_1, \ldots, p_k$. The calculations start from
        re-writing \eqref{eq:50000} in the following form:
        \begin{align}
            \label{eq:50002}
            \E(P^T_n \mid p_1,\ldots,p_k) &=
                \sum_{v \in T} p^v_{p_1,\ldots,p_k} \cdot c_{\text{avg}}^{\lambda(v)}(a_0 - a^v_0,\ldots,a_k - a^v_k)\notag\\
                &= \sum_{v \in T} p^v_{p_1,\ldots,p_k} \cdot c_{\text{avg}}^{\lambda(v)}(a_0,\ldots,a_k){} -{} \notag\\
            &\quad\quad \sum_{v \in T} p^v_{p_1,\ldots,p_k} \left(c_{\text{avg}}^{\lambda(v)}(a_0, \ldots, a_k) - c_{\text{avg}}^{\lambda(v)}(a_0 {-} a^v_0,\ldots,a_k {-} a^v_k)\right)\notag\\
            &=\sum_{\lambda \in \Lambda_k}
            f^{\lambda}_{p_1,\ldots,p_k} \cdot c_\textnormal{\text{avg}}^\lambda(a_0,\ldots,a_k) {} -{} \notag\\
            &\quad\quad\sum_{v \in T} p^v_{p_1,\ldots,p_k} \left(c_{\text{avg}}^{\lambda(v)}(a_0, \ldots, a_k) - c_{\text{avg}}^{\lambda(v)}(a_0 {-} a^v_0,\ldots,a_k {-} a^v_k)\right).\notag\\
    \end{align}
        For
        each node $v$ in the  classification tree, we say that $v$ is \emph{on track} (to the expected values) if
        \begin{align*}
            \big | c_{\text{avg}}^{\lambda(v)}(a_0,\ldots,a_k) - c_{\text{avg}}^{\lambda(v)}(a_0 -
            a^v_0, \ldots, a_k - a^v_k) \big | \leq \frac{k^2}{n^{1/12}}.
        \end{align*}
        Otherwise, $v$ is called \emph{off track}.

        By considering on-track and off-track nodes in \eqref{eq:50002} separately, we may
        calculate
        \begin{align}\label{eq:50001}
            \E(P^T_n \mid p_1,\ldots,p_k) &\leq \sum_{\lambda \in \Lambda_k}
            f^{\lambda}_{p_1,\ldots,p_k} \cdot c_{\text{avg}}^\lambda(a_0,\ldots,a_k)
            + \sum_{\substack{v \in T\\v \text{ is on track}}} p^v_{p_1,\ldots,p_k} \frac{k^2}{n^{1/12}} {} + {} \notag\\
                &\quad\sum_{\substack{v \in T\\ \text{$v$ is off track}}} p^v_{p_1,\ldots,p_k}
                \left(c_{\text{avg}}^{\lambda(v)}(a_0, \ldots, a_k) -
                c_{\text{avg}}^{\lambda(v)}(a_0 {-} a^v_0,\ldots,a_k {-} a^v_k)\right)\notag\\
        &\leq \sum_{\lambda \in \Lambda_k}
            f^{\lambda}_{p_1,\ldots,p_k} \cdot c_\textnormal{\text{avg}}^\lambda(a_0,\ldots,a_k)
            + k \cdot \sum_{\substack{v \in T\\ \text{$v$ is off track}}} p^v_{p_1,\ldots,p_k} + O(n^{11/12})\notag\\
        &= \sum_{\lambda \in \Lambda_k}
            f^{\lambda}_{p_1,\ldots,p_k} \cdot c_\textnormal{\text{avg}}^\lambda(a_0,\ldots,a_k)
            {}+ {} \notag\\
            &\quad\quad k \cdot \sum_{i = 1}^{n-k} \Pr(\text{an off-track node is reached on level $i$}) + O(n^{11/12}).
    \end{align}
        It remains to bound the second summand of \eqref{eq:50001}.
        First, we obtain the general bound:
        \begin{align*}
            \big | c_{\text{avg}}^{\lambda(v)}(a_0,\ldots,a_k) &- c_{\text{avg}}^{\lambda(v)}(a_0 -
            a^v_0, \ldots, a_k - a^v_k) \big |\\
            &\leq (k - 1) \cdot \sum_{j = 0}^{k} \bigg \vert \frac{a_j}{n-k} -
            \frac{a_j - a^v_j}{n - k - \lv{v}} \bigg \vert\\
            &\leq (k - 1) \cdot (k + 1) \cdot \max_{0 \leq j \leq k}\left\{\left\vert\frac{a_j}{n-k} -
            \frac{a_j - a^v_j}{n - k - \lv{v}} \right\vert \right\}.
        \end{align*}
        Thus, by definition, whenever $v$ is an off-track node, there exists $j \in \{0,\ldots,k\}$ such that
        \begin{align*}
    \left\vert\frac{a_j}{n-k} -
    \frac{a_j - a^v_j}{n - k - \lv{v}} \right\vert > \frac{1}{n^{1/12}}.
        \end{align*}
        Now consider the case that the random variables $C_{h,i}$ that counts the number of
        $\text{A}_h$-elements in the first $i$
        classifications are concentrated around
        their expectation, as in the statement of
        Lemma~\ref{lem:sample:concentration:k:pivots}. This happens with
        very high probability, so the contributions of the other case 
        to the average comparison count can be neglected.  For each $h \in \{0,\ldots,k\}$, and each level $i \in \{1,\ldots, n-k\}$ we calculate
    \begin{align*}
        \left|\frac{a_h}{n-k} - \frac{a_h - C_{h,i}}{n - k - i}\right| &\leq
    \left|\frac{a_h}{n-k} - \frac{a_h(1 - i/(n-k))}{n-k-i}\right| + \left| \frac{n^{2/3}}{n - k -
    i}\right|
    = \frac{n^{2/3}}{n - k - i}.
    \end{align*}
    So, for the first $i \leq n - n^{3/4}$ levels, we are
		in an \emph{on-track node} on level $i$ with very high
    probability, because the deviation of
    the ideal probability $a_h/(n-k)$ of seeing an element which belongs to group $\text{A}_h$
    and the actual probability in the node
    reached on level $i$ of seeing such an element is at most $1/n^{1/12}$.
    Thus, for the first $n - n^{3/4}$
    levels the contribution of the sums
    of the probabilities of off-track nodes is not more than $O(n^{11/12})$ to the
    first summand in \eqref{eq:50001}. For the last
    $n^{3/4}$ levels of the tree, we use that the contribution of the  probabilities
    that we reach an off-track node on level $i$ is at most $1$ for a fixed level.

    This shows that the second summand in \eqref{eq:50001} is $O(n^{11/12})$. The lemma now follows
    from averaging over all possible pivot choices. \qed
    \end{proof}

    \section{Example: 3-pivot Quicksort}\label{sec:3:pivot:quicksort}
    Here we study variants of $3$-pivot quicksort algorithms in the light of
    Lemma~\ref{lem:k:pivot:average:partition:cost}.
    This paradigm got recent attention by the work of
    \citeN{Kushagra14}, who provided evidence that---in
    practice---a $3$-pivot quicksort algorithm might be faster than
    the YBB dual-pivot algorithm.

    In $3$-pivot quicksort, we might choose from five different comparison
    trees.  These trees, together with their comparison cost, are depicted in
    Figure~\ref{fig:3:pivot:comparison:trees}. We will study the average
    comparison count of three different strategies in an artificial setting: We
    assume, as in the analysis, that our input is a permutation of
    $\{1,\ldots,n\}$. So, after choosing the pivots the algorithm knows the
    exact group sizes in advance. Transforming this strategy into a realistic one is a topic of the next section.

    All considered strategies will follow the same idea: After choosing the
    pivots, it is checked which comparison tree has the smallest average cost
    for the group sizes found in the input. Then this tree is used for all
    classifications. Our strategies differ with respect to the set of comparison trees
    they can use. In the next section we will explain
    why deviating from such a strategy, i.e., using different trees during the
    classification for fixed group sizes, does not help for minimizing the
    average comparison count.

    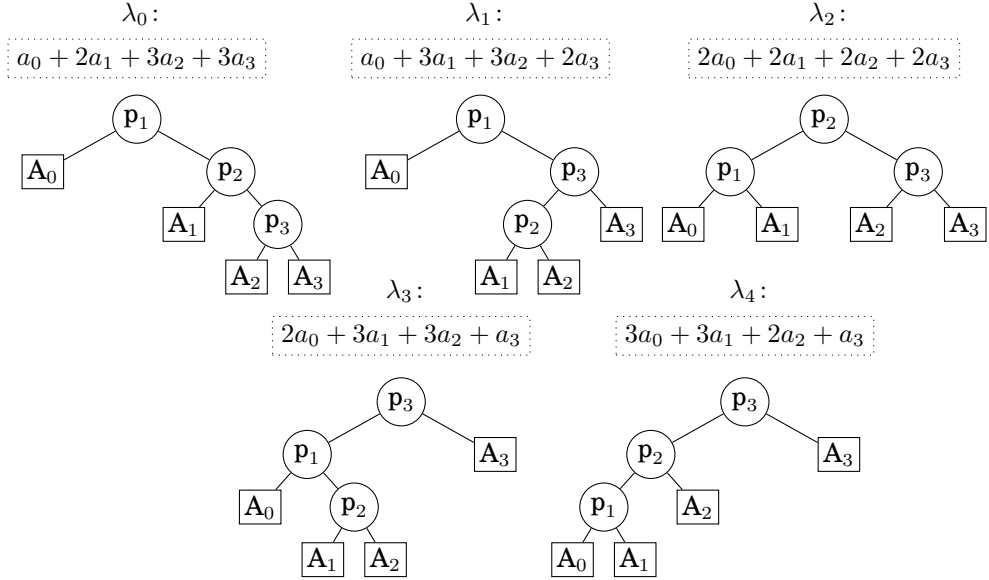
\begin{figure}
        \centering
    \begin{tikzpicture}[level/.style = {sibling distance = 2.5cm/#1, level distance=0.7cm}]
        \tikzset {
            treenode/.style = {align = center, inner sep = 2pt, text centered},
            n/.style = {treenode, circle, white, draw = black, text = black},
            leaf/.style = {treenode, rectangle, draw = black}
        }
            \node [n] {$\text{p}_1$}
            child { node [leaf] {$\text{A}_0$}}
            child { node [n] {$\text{p}_2$}
                child { node [leaf] {$\text{A}_1$}}
                child { node [n] {$\text{p}_3$}
                    child { node [leaf] {$\text{A}_2$}}
            child { node [leaf] {$\text{A}_3$}}}};

            \node[xshift=130pt] [n] {$\text{p}_1$}
    child { node [leaf] {$\text{A}_0$}}
    child { node [n] {$\text{p}_3$}
        child { node [n] {$\text{p}_2$}
                child { node [leaf] {$\text{A}_1$}}
                child { node [leaf] {$\text{A}_2$}}
            }
        child { node [leaf] {$\text{A}_3$}}};

        \node[xshift=260pt] [n] {$\text{p}_2$}
            child { node [n] {$\text{p}_1$}
                child { node [leaf] {$\text{A}_0$}}
                child { node [leaf] {$\text{A}_1$}}
            }
            child { node [n] {$\text{p}_3$}
                child { node [leaf] {$\text{A}_2$}}
                child { node [leaf] {$\text{A}_3$}}
            };

        \node [yshift = -107pt, xshift = 100pt] [n] {$\text{p}_3$}
            child { node [n] {$\text{p}_1$}
                child { node [leaf] {$\text{A}_0$}}
                child { node [n] {$\text{p}_2$}
                    child { node [leaf] {$\text{A}_1$}}
                child { node [leaf] {$\text{A}_2$}}
        }}
            child { node [leaf] {$\text{A}_3$}};

        \node [yshift = -107pt, xshift = 230pt] [n] {$\text{p}_3$}
        child { node [n] {$\text{p}_2$}
            child { node [n] {$\text{p}_1$}
                child { node [leaf] {$\text{A}_0$}}
                child { node [leaf] {$\text{A}_1$}}
            }
            child { node [leaf] {$\text{A}_2$}}}
        child { node [leaf] {$\text{A}_3$}};

        \node (l0) at (3pt, 40pt) {$\lambda_0\colon$};
        \node (l1) at (133pt, 40pt) {$\lambda_1\colon$};
        \node (l2) at (263pt, 40pt) {$\lambda_2\colon$};
        \node (l3) at (103pt, -65pt) {$\lambda_3\colon$};
        \node (l4) at (233pt, -65pt) {$\lambda_4\colon$};

        \node[draw,rectangle,dotted] (c0) at (0pt, 23pt) {$a_0 + 2 a_1 + 3 a_2 + 3 a_3$};
        \node[draw,rectangle,dotted] (c1) at (130pt, 23pt) {$a_0 + 3 a_1 + 3 a_2 + 2 a_3$};
        \node[draw,rectangle,dotted] (c2) at (260pt, 23pt) {$2 a_0 + 2 a_1 + 2 a_2 + 2a_3$};
        \node[draw,rectangle,dotted] (c3) at (100pt, -82pt) {$2a_0 + 3a_1 + 3 a_2 + a_3$};
        \node[draw,rectangle,dotted] (c4) at (230pt, -82pt) {$3a_0 + 3a_1+2 a_2 + a_3$};
        \end{tikzpicture}
        \caption{The different comparison trees for $3$-pivot quicksort with their
        comparison cost (dotted boxes, only displaying the numerator).}
        \label{fig:3:pivot:comparison:trees}
    \end{figure}

    \paragraph{The symmetric strategy} In the algorithm of \cite{Kushagra14}, the balanced
    comparison tree $\lambda_2$ is used for each classification. Using
    Lemma~\ref{lem:k:pivot:average:partition:cost}, we get\footnote{Of course,
        $\E(P_n) = 2 (n - 3)$, since each classification makes exactly two
    comparisons.}

    \begin{align*}
        \E(P_n) &=  \frac{1}{\binom{n}{3}}
        \sum_{a_0 + a_1 + a_2 + a_3 = n - 3} (2a_0 + 2a_1 + 2a_2 +  2a_3) + O(n^{1-\varepsilon})\\
        &= 2n + O(n^{1-\varepsilon}).
    \end{align*}
    Using Theorem~\ref{thm:k:pivot:recurrence:solution}, we conclude that

    \begin{align*}
        \E(C_n) = 24/13n \ln n + O(n) \approx 1.846 n \ln n + O(n),
    \end{align*}
    as known from \cite{Kushagra14}. This improves on classical quicksort ($2n
    \ln n + O(n)$ comparisons on average), but is worse than optimal dual-pivot
    quicksort ($1.8 n\ln n + O(n)$ comparisons on average \cite{AumullerD15}) or
    median-of-$3$ quicksort ($1.714 n \ln n + O(n)$ comparisons on average \cite{vanEmden}).

    \paragraph{Using three trees} Here we restrict our algorithm to
    choose only among the comparison trees $\{\lambda_1,\lambda_2,\lambda_3\}$. The computation
    of a cost-minimal comparison tree is then simple:
    Suppose that the segment sizes are $a_0,\ldots,a_3$. If
    $a_0 > a_3$ and $a_0 > a_1 + a_2$ then comparison tree $\lambda_1$ has minimum cost.
    If $a_3 \geq a_0$ and $a_3 > a_1 + a_2$ then comparison tree $\lambda_3$ has minimum cost.
    Otherwise $\lambda_2$ has minimum cost.

    Using Lemma~\ref{lem:k:pivot:average:partition:cost}, the average
    partition cost
    with respect to this set of
    comparison trees  can be calculated (using Maple\textsuperscript{\textregistered}) as follows:
    \begin{align*}
        \E(P_n) &= \frac{1}{\binom{n}{3}} \sum_{a_0 + a_1 + a_2 + a_3 =
    n-3} \min\left\{\substack{a_0 + 3a_1 + 3a_2 + 2a_3, 2a_0 + 2a_1 + 2 a_2 + 2a_3, \\2a_0 + 3a_1 +
    3a_2 + 1a_3}\right\}
    + O(n^{1-\varepsilon})\\
    &= \frac{17}{9} n +  O(n^{1-\varepsilon}).
    \end{align*}
    This yields the following average comparison cost:

    \begin{align*}
        \E(C_n) = \frac{68}{39} n \ln n + O(n) \approx 1.744 n \ln n + O(n).
    \end{align*}

    \paragraph{Using all trees} Now we let our strategies choose among all five trees.
    Using
    Lemma~\ref{lem:k:pivot:average:partition:cost} and the
    average cost for all trees from Figure~\ref{fig:3:pivot:comparison:trees}, we
    calculate (using Maple\textsuperscript{\textregistered})

    \begin{align}
        \E(P_n) &= \frac{1}{\binom{n}{3}} \sum_{a_0 + a_1 + a_2 + a_3 =
    n-3} \min\left\{\substack{a_0 + 2a_1 + 3a_2 + 3a_3, a_0 + 3a_1 + 3a_2 + 2a_3,\\
    2a_0 + 2a_1 + 2a_2 + 2a_3, 2a_0 + 3a_1 +
    3a_2+ a_3\\ 3a_0 + 3a_1 + 2a_2 + a_3}\right\}
    + O(n^{1-\varepsilon})\label{eq:cost:all:trees}\\
    &= \frac{133}{72} n +  O(n^{1-\varepsilon}).\notag
    \end{align}
    This yields the following average comparison cost:

    \begin{align*}
        \E(C_n) = \frac{133}{78} n \ln n + O(n) \approx 1.705 n \ln n + O(n),
    \end{align*}
    which is---as will be explained in the next section---the lowest possible
    average comparison count one can achieve by picking three pivots directly
    from the input. So, using three pivots gives a slightly lower average
    comparison count than quicksort using the median of three elements as the
    pivot.

    \section{(Asymptotically) Optimal Classification Strategies}\label{sec:optimal:strategies}
    In this section we will discuss four different strategies, which will all achieve the
    minimal average comparison count (up to lower order terms). Two of these four strategies
    will be optimal but unrealistic, since they assume that after the pivots are fixed
    the algorithm knows the sizes of the $k+1$ different groups. The strategies work as follows:
    One strategy maintains the group sizes of the unclassified part of the input and chooses 
    the comparison tree with minimum cost with respect to these group sizes. This will
    turn out to be the optimal classification strategy for $k$-pivot quicksort.
    To turn this strategy into an actual algorithm, we will use the group sizes
    of the already classified part of the input as a basis for choosing the comparison
    tree for the next classification. The second unrealistic strategy works like the algorithms
    for $3$-pivot quicksort. It will use the comparison tree with minimum cost with
    respect to the group sizes of the input in each classification.
    To get an actual algorithm, we estimate these group sizes in a small sampling step.
    Note that these strategies are the obvious
    generalization of the optimal strategies for dual-pivot quicksort from
    \cite{AumullerD15}.

    Since all these strategies need to compute cost-minimal comparison trees, this section
    starts with a short discussion of algorithms for this problem. Then we discuss
    the four different strategies.

    \subsection{Choosing an Optimal Comparison Tree}\label{sec:optimal:comparison:tree}
    For optimal $k$-pivot quicksort algorithms it is of
    course necessary to devise an algorithm that can compute an optimal
    comparison tree for partition sizes $a_0,\ldots,a_k$, i.e., a comparison tree
    that minimizes \eqref{eq:30014}. It is well known  that the number of binary search trees
    with $k$ inner nodes equals the $k$-th
    Catalan number, which is approximately  $4^k/ ((k+1)\sqrt{\pi
    k}).$ Choosing an optimal comparison tree is a standard application
    of dynamic programming, and is known from textbooks as ``choosing an optimum
    binary search tree'', see, e.g., \cite{Knuth}. The algorithm runs in
    time and space $O(k^2)$. In our case where the cost is only associated with the leaves, we wish to find an optimal \emph{alphabetic} binary tree. This can be solved in time $O(k \log k)$ using either the algorithm of
    \citeN{HuTu71} or that of \citeN{GarsiaW77}.

    \subsection{The Optimal Classification Strategy and its Algorithmic Variant}
    Here, we consider the following strategy\footnote{For all strategies we 
    say which comparison tree is used in a given node of the classification tree. Recall that
    the classification order is arbitrary.} $\mathcal{O}_k$: \emph{Given $a_0,\ldots,a_k$, the
    comparison tree $\lambda(v)$ is one that minimizes
    $\text{cost}^{\lambda}(a_0-a_0^v,\ldots,a_k-a_k^v)$ over all comparison trees
    $\lambda$.}

    While being unrealistic, since the exact partition sizes $a_0,\ldots,a_k$ are
    in general unknown to the algorithm, strategy $\mathcal{O}_k$ is the optimal
    classification strategy, i.e., it minimizes the average comparison
    count.

    \begin{theorem}
        Strategy $\mathcal{O}_k$ is optimal for each $k$.
    \label{thm:o:k:optimal}
    \end{theorem}

    \begin{proof}
        Strategy $\mathcal{O}_k$ chooses for each node
        $v$ in the classification tree the comparison tree that minimizes the average cost
        in \eqref{eq:50000}. So, it minimizes each term of the sum and thus minimizes
        the whole sum in \eqref{eq:50000}.\qed
    \end{proof}
    There exist other strategies whose average
    comparison count differs by at most $O(n^{1-\varepsilon})$ from the average comparison count of
    $\mathcal{O}_k$. We call such strategies \emph{asymptotically optimal}.

    Strategy $\mathcal{C}_k$ is an algorithmic variant of $\mathcal{O}_k$. It works as follows:
    \emph{The
    comparison tree $\lambda(v)$ is one that minimizes
    $\text{cost}^{\lambda}(a_0^v,\ldots,a_k^v)$ over all comparison trees $\lambda$.}

    \begin{theorem}
        Strategy $\mathcal{C}_k$ is asymptotically optimal for each $k$.
    \label{thm:l:k:optimal}
    \end{theorem}

    \begin{proof}
        Since the average comparison count is independent of the actual order in which elements are classified, assume that strategy $\mathcal{O}_k$ classifies elements in the order $e_{k+1}, \ldots,
        e_n$, while strategy $\mathcal{C}_k$ classifies them in reversed order, i.e., $e_n, \ldots,
        e_{k+1}$. Then the comparison tree that
        is used by $\mathcal{C}_k$  for element $e_i$ is the one that $\mathcal{O}_k$
        is using for element $e_{i+1}$ because both strategies use the group sizes in 
        $(e_{i + 1}, \ldots, e_n)$. Let $P_i$ and $P'_i$ denote the number of
        comparisons for the classification of the element $e_{k + i}$ using strategy $\mathcal{O}_k$ 
        and $\mathcal{C}_k$, respectively. 

        Fix some integer $i \in \{1,\ldots,n-k\}$. 
        Suppose that the input has group sizes $a_0, \ldots, a_k$. 
        Assume that the sequence $(e_{k + 1}, \ldots, e_i)$ contains $a'_h$ elements of group
        $\text{A}_h$ for $h \in \{0, \ldots, k\}$, where  
        $\vert a'_h - i \cdot
        a_h/(n-k)\vert \leq
        n^{2/3}$ for each $h \in \{0,\ldots,k\}$. Let
        $\lambda$ be a comparison tree with minimal cost w.r.t. $(a_0 - a'_0, \ldots, a_k
        - a'_k)$. For a random input having group sizes from above we calculate:

        \begin{align*}
            \Bigl \vert \E\left(P_{i+1}\right) -
            \E\left(P'_i \right) \Bigr \vert
            &= \Bigl \vert c_{\text{avg}}^\lambda(a_0 - a'_0, \ldots, a_k - a'_k) - c_{\text{avg}}^\lambda(a'_0, \ldots, a'_k)\Bigr \vert\\
            &=\Biggl \vert \frac{\sum_{h = 0}^{k} \text{depth}_\lambda(\text{A}_h) \cdot (a_h - a'_h)}{n - k- i} -
            \frac{\sum_{h = 0}^{k} \text{depth}_\lambda(\text{A}_h)\cdot a'_h}{i}\Biggr \vert\\
            &\leq\sum_{h = 0}^{k} \text{depth}_\lambda(\text{A}_h) \Biggl \vert\frac{a_h - a'_h}{n - k- i} -
        \frac{a'_h}{i}\Biggr \vert\\
        &\leq\sum_{h = 0}^{k} \text{depth}_\lambda(\text{A}_h) \left(\Biggl \vert\frac{a_h - \frac{i \cdot a_h}{n-k}}{n - k- i} -
    \frac{i \cdot a_h}{i}\Biggr \vert + \frac{n^{2/3}}{n - k - i} + \frac{n^{2/3}}{i}\right)\\
        &=\sum_{h = 0}^{k} \text{depth}_\lambda(\text{A}_h) \left(\frac{n^{2/3}}{n - k - i} + \frac{n^{2/3}}{i}\right) \leq\frac{k^2 \cdot n^{2/3}}{ n  - k - i} + \frac{k^2 \cdot n^{2/3}}{i}.
        \end{align*}
        Assume that the concentration argument of
        Lemma~\ref{lem:sample:concentration:k:pivots} holds. Then the difference between the average comparison count
        for element $e_{i+1}$ (for $\mathcal{O}_k$) and $e_{i}$ (for $\mathcal{C}_k$)
        is at most 
        \begin{align*}
        \frac{k^2 \cdot n^{2/3}}{ n  - k - i} + \frac{k^2 \cdot n^{2/3}}{i}.
        \end{align*}
        The difference of the average comparison count over all elements
        $e_{i},\ldots,e_j$, $i \geq n^{3/4}, j \leq n - n^{3/4}$, is then at most $O(n^{11/12})$.
        For elements that reside outside of this range, the difference in the
        average comparison count is at most $2n^{3/4} \cdot k$.  Furthermore, 
        error terms for cases where the concentration argument does not hold
        can be neglected because they occur with exponentially low probability. So, the total
        difference of the average comparison count between strategy
        $\mathcal{O}_k$ and strategy $\mathcal{C}_k$ is at most $O(n^{11/12})$. \qed
    \end{proof}
    This shows that the optimal strategy $\mathcal{O}_k$ can be approximated by
    an actual algorithm that makes an error of up to $O(n^{11/12})$, which sums up to an
    error term of $O(n)$ over the whole recursion by Theorem~\ref{thm:k:pivot:recurrence:solution}. In the case of
    dual-pivot quicksort, the difference between $\mathcal{O}_2$ and
    $\mathcal{C}_2$ is $O(\log n)$, which also sums up to a difference of $O(n)$ over
    the whole recursion \cite{AumullerD15}. It remains an open question to prove tighter bounds
    than $O(n^{11/12})$ in the general case.

    \subsection{A Fixed Strategy and its Algorithmic Variant}

    Now we turn to
    strategy $\mathcal{N}_k$: \emph{Given $a_0,\ldots,a_k$, the comparison tree $\lambda(v)$
        used at node $v$ is one that minimizes $\text{cost}^{\lambda}(a_0,\ldots,a_k)$ over
    all comparison trees $\lambda$.}

    Strategy $\mathcal{N}_k$ uses a fixed comparison tree for all
    classifications for given partition sizes, but it has to know these sizes in
    advance.

    \begin{theorem}
        Strategy $\mathcal{N}_k$ is asymptotically optimal for each $k$.
    \label{thm:n:k:optimal}
    \end{theorem}

    \begin{proof}
        According to Lemma~\ref{lem:k:pivot:average:partition:cost} the average comparison count is determined up to
        lower order terms by the parameters $f^{\lambda}_{p_1,\ldots,p_k}$, for each
        $\lambda \in \Lambda_k$. For each $p_1,\ldots,p_k$, strategy $\mathcal{N}_k$ chooses
        the comparison tree which minimizes the average cost. By
        Lemma~\ref{lem:k:pivot:average:partition:cost}, this is
        optimal up to an $O(n^{1-\varepsilon})$ term. \qed
    \end{proof}

    We will now describe how to implement strategy $\mathcal{N}_k$ by using
    sampling.  Strategy $\mathcal{SP}_k$ works as follows: Let $\lambda_0 \in
    \Lambda_k$ be an arbitrary comparison tree. After the pivots are chosen,
    inspect the first $n^{3/4}$ elements and classify them using $\lambda_0$. Let
    $a'_0,\ldots,a'_k$ denote the number of elements that belonged to
    $\text{A}_0,\ldots,\text{A}_k$, respectively. Let $\lambda$ be a 
    comparison tree with minimal cost for $a'_0,\ldots, a'_k$.  Then classify each of the
    remaining elements by using $\lambda$.

    \begin{theorem}
    Strategy $\mathcal{SP}_k$ is asymptotically optimal for each $k$.
    \label{thm:s:k:optimal}
    \end{theorem}

    \begin{proof}
        Fix the $k$ pivots $p_1,\ldots,p_k$ and thus $a_0,\ldots,a_k$. 
        According to Lemma~\ref{lem:k:pivot:average:partition:cost}, the average
        comparison count $\E(P^{\mathcal{SP}_k}_{n} \mid p_1,\ldots,p_k)$ can
        be calculated as follows:
        \begin{align*}
            \E(P^{\mathcal{SP}_k}_{n} \mid p_1,\ldots,p_k) = \sum_{\lambda \in \Lambda_k}
            f^\lambda_{p_1,\ldots,p_k} \cdot c_{\text{avg}}^\lambda(a_0,\ldots,a_k) +
            O(n^{1-\varepsilon}).
        \end{align*}
        Let $\lambda^\ast$ be a comparison tree with minimal cost w.r.t.
        $a_0,\ldots,a_k$. Let $a'_0,\ldots,a'_k$ be the partition sizes after
        inspecting $n^{3/4}$ elements. Let $\lambda$ be a comparison tree with minimal
        cost w.r.t.  $a'_0,\ldots,a'_k$. We call $\lambda$ \emph{good} if
        \begin{align}
            &c_{\text{avg}}^{\lambda}(a_0,\ldots,a_k) -
            c_{\text{avg}}^{\lambda^\ast}(a_0,\ldots,a_k) \leq
            \frac{2k}{n^{1/12}}, \text{ or equivalently}\notag\\
            &\text{cost}^{\lambda}(a_0,\ldots,a_k) -
            \text{cost}^{\lambda^\ast}(a_0,\ldots,a_k) \leq
            2k n^{11/12},
            \label{eq:sp:k:2}
        \end{align}
        otherwise we call $\lambda$ \emph{bad}. We define $\text{good}_\lambda$ and $\text{bad}_\lambda$ as the
        events that the sample yields a good and bad comparison tree, respectively.

        We calculate:
        \begin{align}
            \E(P^{\mathcal{SP}_k}_{n} \mid p_1,\ldots,p_k) &= \sum_{\substack{\lambda \in \Lambda_k\\\lambda \text{ good}}}
            f^{\lambda}_{p_1,\ldots,p_k} \cdot c_{\text{avg}}^{\lambda}(a_0,\ldots,a_k)
            {}+{}\notag\\
            &\quad\quad\sum_{\substack{\lambda \in \Lambda_k\\\lambda \text{ bad}}}
        f^{\lambda}_{p_1,\ldots,p_k} \cdot c_{\text{avg}}^{\lambda}(a_0,\ldots,a_k) +
        O(n^{1-\varepsilon})\notag\\
        &\leq n  \cdot c_{\text{avg}}^{\lambda^\ast}(a_0,\ldots,a_k) +
    \sum_{\substack{\lambda \in \Lambda_k\\\lambda \text{ bad}}}
        f^{\lambda}_{p_1,\ldots,p_k} \cdot c_{\text{avg}}^{\lambda}(a_0,\ldots,a_k) +
        O(n^{1-\varepsilon})\notag
        \\&\leq n \cdot c_{\text{avg}}^{\lambda^\ast}(a_0,\ldots,a_k) +
    k \cdot \sum_{\substack{\lambda \in \Lambda_k\\\lambda \text{ bad}}}
        f^{\lambda}_{p_1,\ldots,p_k} +
        O(n^{1-\varepsilon}).
        \label{eq:sp:k:1}
    \end{align}
    Now we derive an upper bound for the second summand of \eqref{eq:sp:k:1}.
    After the first $n^{3/4}$ classifications the algorithm will either use
    a good comparison tree or a bad comparison tree for the remaining
    classifications.  The probability $\Pr({\text{bad}_\lambda} \mid p_1, \ldots, p_k)$ is
    the ratio of nodes on each level from $n^{3/4}$ to $n - k$ of
    the classification tree of nodes labeled
    with bad trees (in the sense of \eqref{eq:sp:k:2}). Summing over
    all levels, the second summand of \eqref{eq:sp:k:1} is thus at most $k \cdot
    n \cdot \Pr\left({\text{bad}_\lambda} \mid p_1,\ldots,p_k\right) + O(n^{3/4})$, where the latter summand collects error terms for the first $n^{3/4}$ steps.

        \begin{lemma}
            Conditioned on $p_1,\ldots,p_k$, $\text{good}_\lambda$ occurs with very
            high probability.
        \end{lemma}

        \begin{proof}
            For each $i \in \{0,\ldots,k\}$, let $a'_i$ be the random variable that
            counts the number of elements from the sample
            that belong to group $\text{A}_i$. According to
            Lemma~\ref{lem:sample:concentration:k:pivots}, with very high probability we have that $|a'_i - \E(a'_i)| \leq
            n^{2/3}$, for each $i$ with $0 \leq i \leq k$. By the union bound, with
            very high probability there is no $a'_i$ that deviates by more than
            $n^{2/3}$ from its expectation $n^{-1/4} \cdot a_i$. We will now show
            that if this happens then the event $\text{good}_\lambda$ occurs.
            We obtain the following upper bound for an arbitrary comparison tree
            $\lambda' \in \Lambda_k$:
            \begin{align*}
                \text{cost}^{\lambda'}(a'_0,\ldots,a'_k) &= \sum_{0 \leq i \leq k} \text{depth}_{\lambda'}(\text{A}_i) \cdot a'_i \\
                &\leq \sum_{0 \leq i \leq k}\text{depth}_{\lambda'}(\text{A}_i) \cdot n^{2/3} + n^{-1/4}\cdot
                \text{cost}^{\lambda'}(a_0,\ldots,a_k)\\
                &\leq k^2 n^{2/3} + n^{-1/4}\cdot
                \text{cost}^{\lambda'}(a_0,\ldots,a_k).
            \end{align*}
            Similarly, we get a corresponding lower
            bound. Thus, for each comparison tree $\lambda' \in \Lambda_k$ it holds that
            \begin{align*}
                \frac{\text{cost}^{\lambda'}(a_0,\ldots,a_k)}{n^{1/4}}  - k^2 n^{2/3}
                \leq  \text{cost}^{\lambda'}(a'_0,\ldots,a'_k) \leq
                \frac{\text{cost}^{\lambda'}(a_0,\ldots,a_k)}{n^{1/4}}  + k^2 n^{2/3},
            \end{align*}
            and we get the following bound:
            \begin{align*}
                \text{cost}^{\lambda}(a_0,\ldots,a_k) &-
                \text{cost}^{\lambda^\ast}(a_0,\ldots,a_k) \\
                &\leq n^{1/4}
                \bigl(\text{cost}^{\lambda}(a'_0,\ldots,a'_k) -
                \text{cost}^{\lambda^\ast}(a'_0,\ldots,a'_k)\bigr) + 2 n^{1/4} \cdot k^2
                \cdot n^{2/3}\\
                &\leq 2 k^2 \cdot n^{11/12}.
            \end{align*}
            (The last inequality follows because $\lambda$ has minimal cost w.r.t.
            $a'_0,\ldots,a'_k$.) Hence, $\lambda$ is good. \qed

        \end{proof}
        Thus, the average comparison count of $\mathcal{SP}_k$ is at most a summand of
        $O(n^{1-\varepsilon})$ larger than the average comparison count of $\mathcal{N}_k$. This implies that
        $\mathcal{SP}_k$ is asymptotically optimal as well. \qed
    \end{proof}
    Since the number of comparison trees in $\Lambda_k$ is exponentially
    large in $k$, one might want to restrict the set of used comparison trees to
    some subset $\Lambda'_k \subseteq \Lambda_k$.  We remark here that our
    strategies are optimal w.r.t. any chosen subset of comparison trees as well.

    \section[The Optimal Average Comparison Count of k-Pivot Quicksort]{The Optimal Average Comparison Count of {\textnormal\MakeLowercase{$k$}}-Pivot Quicksort}\label{sec:comparison:median:of:k} 
    In this section we use the theory developed so far to discuss the
    optimal average comparison count of $k$-pivot quicksort. We compare
    the result to the well known median-of-$k$ strategy of classical quicksort \cite{vanEmden}.

    By Lemma~\ref{lem:k:pivot:average:partition:cost} and Theorem~\ref{thm:n:k:optimal},
    the minimal partitioning cost for $k$-pivot quicksort (up to lower order terms) is
    \begin{align}
        \frac{1}{\binom{n}{k}}\sum_{a_0 + \cdots + a_k = n - k} \min\left\{\text{cost}^\lambda(a_0, \ldots, a_k) \mid \lambda \in \Lambda_k\right\} + O(n^{1-\varepsilon}).
        \label{eq:optimal:cost:formula}
    \end{align}
    Then applying Theorem~\ref{thm:k:pivot:recurrence:solution} gives the minimal
    average comparison count for $k$-pivot quicksort.

    Unfortunately, we were not able to solve \eqref{eq:optimal:cost:formula} for
    $k \geq 4$.  (Already the solution for $k = 3$ as stated in
    Section~\ref{sec:3:pivot:quicksort} required a lot of manual tweaking before
    using Maple\textsuperscript{\textregistered}.) This remains an open
    question. We resorted to experiments. As noticed in \cite{AumullerD15},
    estimating the total average comparison count by sorting inputs does not
    allow us to estimate the leading term of the average comparison count
    correctly, because the $O(n)$ term in \eqref{eq:1} has a big influence on
    the average comparison count for real-world input lengths. We used the
    following approach instead: For $n = 50 \cdot 10^6$, we generated $10\,000$
    random permutations of $\{1, \ldots, n\}$ and ran strategy $\mathcal{O}_k$
    for each input, i.e., only classified the input.\footnote{Experiments with 
    other input sizes gave exactly the same results.}  For the average
    partitioning cost measured in these experiments, we then applied
    \eqref{eq:1} to derive the leading factor of the total average comparison
    count.  Table~\ref{tab:optimal:cost:k:pivot:quicksort} shows the results
    from these experiments for $k \in \{2,\ldots,9\}$. Note that the results for
    $k \in \{2,3\}$ are almost identical to the exact theoretical results. 
    Additionally, the table shows the theoretical results known for classical
    quicksort using the median-of-$k$ strategy \cite{vanEmden,hennequin}.
    \begin{table}[thb]
    \tbl{Optimal average comparison count for $k$-pivot quicksort for $k \in
    \{2,\ldots,9\}$. Note that the values for $k \geq 4$ are based on experiments. For odd $k$, we also include the average comparison count
of quicksort with the median-of-$k$ strategy. (The numbers for the
median-of-$k$ variant can be found in \protect\cite{vanEmden} or \protect\cite{hennequin}.)
    \label{tab:optimal:cost:k:pivot:quicksort}}
    {
    \begin{tabular}{crr}
        \toprule
        \textbf{Pivot Number} $k$ & \textbf{opt. $k$-pivot} & \textbf{median-of-$k$}\\ \hline
        $2$ & $1.800 n \ln n$ & ---\\
        $3$ & $1.705 n \ln n$ & $1.714 n \ln n$ \\
        $4$ & $1.650 n \ln n$ & ---\\
        $5$ & $1.610 n \ln n$ & $1.622 n \ln n$ \\
        $6$ & $1.590 n \ln n$ & --- \\
        $7$ & $1.577 n \ln n$ & $1.576 n \ln n$ \\
        $8$ & $1.564 n \ln n$ & --- \\
        $9$ & $1.555 n \ln n$ & $1.549 n \ln n$ \\
        \bottomrule
    \end{tabular}
    }
    \end{table}
    Interestingly, from Table~\ref{tab:optimal:cost:k:pivot:quicksort} we see
    that---based on our experimental data for $k$-pivot quicksort---starting from $k = 7$
		the median-of-$k$ strategy has a slightly lower
    average comparison count than the (rather complicated) optimal partitioning
    methods for $k$-pivot quicksort.

    \section{Rearranging Elements}\label{sec:assignments}
        With this section, we change our viewpoint on multi-pivot quicksort 
        in two respects: 
        we consider cost measures different than comparisons and focus on one 
        particularly interesting algorithm for the ``rearrangement problem''. The goal now
        is to find other cost measures which show differences in multi-pivot quicksort algorithms with 
        respect to running time in practice. 

    \subsection{Which Factors are Relevant for Running Time?}

        Let us first reflect on the influence of key comparisons to running time.
        From a running time perspective it
        seems unintuitive that comparisons are the crucial factor with regard
        to running time, especially when key comparisons are cheap, e.g.,
        when comparing $32$-bit integers. However, while a comparison is often
        cheap, mispredicting the destination that a branch takes, i.e., the
        outcome of the comparison, may incur a significant penalty in running
        time, because the CPU wasted work on executing instructions on the
        wrongly predicted branch. One famous example for the effect of
        branch prediction is \cite{KaligosiS06} in which quicksort is made
        faster by choosing a skewed pivot due to pipelining effects on a
        certain CPU. In very recent work, \citeN{MartinezNW15} considered
        differences in branch misses between classical quicksort and
        the YBB algorithm, but found no crucial differences. They concluded
        that the advantages in running time of the dual-pivot approach are not
        due to differences in branch prediction. 

        Traditionally, the \emph{cost of moving elements around} is also
considered as a cost measure of sorting algorithms. This cost is usually expressed as
the number of \emph{swap} operations or the number of \emph{assignments} needed
to sort the input.  \cite{Kushagra14} take a different approach and concentrate
on the I/O performance of quicksort variants with respect to their \emph{cache
behavior}. The I/O performance is often a bottleneck of an algorithm because an
access to main memory in modern computers can be slower than executing a few
hundred simple CPU instructions. Caches speed these accesses up, 
but their influence seems difficult to analyze. Let us exemplify the influence of caches on
running time. First, the cache structure of modern CPU's is usually
hierarchical. For example, the Intel i$7$ that we used in our experiments has
three data caches: There is a very small L$1$ cache (32KB of data) and a
slightly larger L$2$ cache (256KB of data) very close to the processor.  Each
CPU core has its own L$1$ and L$2$ cache. They are both $8$-way associative,
i.e., a memory segment can be stored at eight different cache lines.  Shared
among cores is a rather big L$3$ cache that can hold $8$MB of data and is
$16$-way associative.         Caches greatly influence running time.  While a
lookup in main memory costs many CPU cycles ($\approx 140$ cycles on the Intel
i$7$ used in our experiments), a cache access is very cheap and costs about
$4$, $11$, and $25$ cycles for a hit in L$1$, L$2$, and L$3$ cache,
respectively \cite{Levinthal09}. Also, modern CPU's use \emph{prefetching} to
load memory segments into cache before they are accessed. Usually, there exist
different prefetchers for different caches, and there exist different
strategies to prefetch data, e.g., ``load two adjacent cache lines'', or
``load memory segments based on predictions by monitoring data flow''.

        From a theoretical point of view, much research has been conducted to study
        algorithms with respect to their cache behavior, see, e.g., the survey
        paper of \citeN{Rahman02}. (We
        recommend this paper as an excellent introduction to the topic of caches.)
        
        In \cite{Kushagra14}
        a fast three-pivot algorithm was described. They analyzed
        its cache behavior and compared it to classical quicksort and 
        Yaroslavsiky's dual-pivot quicksort algorithm using the approach
        of \cite{LaMarcaL99}. 
        Their results gave reason to believe that the improvements of multi-pivot quicksort
        algorithms with respect to running times result from their better cache behavior.
        They also reported from experiments with a seven-pivot algorithm, which
        ran more slowly than their three-pivot algorithm. Very recently, \cite{NebelWM15} 
        gave a more detailed analysis of the cache behavior of the YBB algorithm
        also with respect to different sampling strategies. An important contribution of 
        \cite{NebelWM15} is the distinction of a theoretical measure \emph{scanned elements} (basically
        the number of times a memory cell is inspected during sorting) and 
        the usage of this cost measures to predict cache behavior.

        In this section we discuss how the considerations of \cite{Kushagra14,NebelWM15} generalize to the 
        case of using more than three pivots. 
        In connection
        with the running time experiments from Section~\ref{sec:experiments},
        this allows us to make more accurate predictions than \cite{Kushagra14}
        about the influence of cache behavior on running time. One result of
        this study will be that it is not surprising that their seven-pivot
        approach is slower, because it has worse cache behavior than three- or
        five-pivot quicksort algorithms using a specific partitioning strategy.

        We will start by specifying the problem setting, and subsequently
        introduce a generalized partitioning algorithm for $k$ pivots. This
        algorithm is the generalization of the partitioning methods used 
        in classical quicksort, the YBB algorithm, and the three-pivot
        quicksort algorithm of \cite{Kushagra14}.
        This strategy will be evaluated for
        different values of $k$ with respect to different memory-related cost 
        measures which will be introduced later. It will turn out that these
        theoretical cost measures allow us to give detailed recommendations 
        under which circumstances a multi-pivot quicksort approach has 
        advantages over classical quicksort. 

        We remark that the same analysis can be done for other partitioning
        algorithms. In his PhD thesis \cite{Aumueller15}, Aumüller considers 
        two variants of the super scalar
        sample sort algorithm of \cite{SandersW04} according
        to the same cost measures that are studied here. We give a short
        overview of the results at the end of this section.

        \subsection{The Rearrangement Problem}
        \label{sec:additional:cost:measures:problem}
        With regard to counting key comparisons we defined the classification problem 
        to abstract from the situation that a multi-pivot quicksort algorithm 
        has to move elements around to produce the partition. Here, we assume that 
        for each element its groups is known and we are only interested in moving 
        elements around to produce the partition. This motivates us to 
        consider the \emph{rearrangement problem for $k$ pivots}: 
        Given a sequence of length $n - k$ with entries having labels from the set $\{\text{A}_0,\ldots,\text{A}_k\}$ of group names, the task is to
        rearrange the entries with respect to their labels into ascending order, where $\text{A}_i < \text{A}_{i + 1}$ for 
        $i \in \{0, \ldots, k - 1\}$. Note that any classification strategy can be used to find out element groups. We
        assume that the input resides in an array $A[1..n]$ where the $k$ first cells hold the pivots.\footnote{
        We shall disregard the pivots in
        the description of the problem. In a final step the $k$ pivots have
        to be moved into the correct positions between group segments. This
        is possible by moving not more than $k^2$ elements around using $k$ rotate operations, 
as introduced below.  }   
        For $k=2$, this problem is known under the name \emph{Dutch national
        flag problem}, proposed by Dijkstra \cite{Dijkstra76}. For $k > 2$, the
        problem was considered in the paper of 
        \citeN{McIlroyBM93}, who devised an algorithm called ``American flag
        sort'' to solve the rearrangement problem for $k > 2$. We will discuss
        the applicability of these algorithms at the end of this section.  Our goal is to
        analyze algorithms for this problem with respect to different cost
        measures, e.g., the number of array cells that are inspected during
        rearranging the input, or the number of times the algorithm writes 
        to array cells in the process. We start by introducing an algorithm for
        the rearrangement problem that generalizes the algorithmic ideas
        behind rearranging elements in classical quicksort, the YBB algorithm \cite{NebelWM15},
        and the three-pivot algorithm of \cite{Kushagra14}.

        \subsection{The Algorithm}
        \label{sec:additional:cost:measures:algorithms}
        To capture the cost of rearranging the elements, in the analysis of sorting algorithms 
        one traditionally uses the ``swap''-operation, 
        which exchanges two elements. The cost of rearranging is then the number of 
        swap operations performed during the sorting process. In the case that one uses two or more pivots, we
        we will see that it is beneficial to generalize this operation. We define
        the operation  ${\tt rotate}(i_1,\ldots,i_\ell)$ as follows:
    $$\text{tmp} \gets A[i_1]; A[i_1] \gets A[i_2]; A[i_2] \gets A[i_3];
    \ldots; A[i_{\ell - 1}] \gets A[i_\ell]; A[i_\ell] \gets \text{tmp}.$$
        The operation \texttt{rotate} performs a cyclic shift of the elements by one
    position. A \texttt{swap}($A[i_1], A[i_2]$) is a \texttt{rotate}($i_1,i_2$).
        A \texttt{rotate}($i_1,\ldots,i_\ell$) operation makes exactly $\ell + 1$ assignments
        and inspects and writes into $\ell$ array cells.

        For each $k \geq 1$ we consider an algorithm
        $\textit{Exchange}_k$. Pseudocode of this algorithm is given in Algorithm~\ref{algo:exchange:k}.
        The basic idea is similar to classical
        quicksort: Two pointers\footnote{Note that our pointers are actually variables that hold an array index.} scan the array. One pointer scans the array from left to right; another
        pointer scans the array from right to left, exchanging misplaced
        elements on the way. Formally, the algorithm uses two pointers
        $\texttt{i}$ and $\texttt{j}$. At the beginning, $\texttt{i}$ points to
        the element in $A[k + 1]$ and $\texttt{j}$ points to the element in
        $A[n]$. We set $m = \lceil\frac{k +  1}{2}\rceil$. 
        The algorithm makes sure that all elements
        to the left of pointer $\texttt{i}$ belong to groups
        $\text{A}_0,\ldots,\text{A}_{m - 1}$ (and are  arranged in this order), 
        i.e., $m$ is the number of groups left of pointer $\texttt{i}$. Also,
        all elements to the right of pointer $\texttt{j}$ belong to groups
        $\text{A}_{m},\ldots,\text{A}_k$, arranged in this order.  To do so, Algorithm~\ref{algo:exchange:k} uses $k-1$ additional
        ``border pointers'' $\texttt{b}_1,\ldots,\texttt{b}_{k - 1}$. For $i < m$, the
        algorithm makes sure that at each point in time, pointer $\texttt{b}_i$
        points to the leftmost element to the left of pointer $\texttt{i}$
        which belongs to group $\text{A}_{i'}$, where $i' \geq i$.
        Analogously, for $j \geq m$, the algorithm
        makes sure that pointer $\texttt{b}_j$ points to the rightmost
        element to the right of pointer $\texttt{j}$ which belongs to group
        $\text{A}_{j'}$ with $j' \leq j$, see Figure~\ref{fig:partition:k}. As long as pointers
        $\texttt{i}$ and $\texttt{j}$ have not crossed yet, the algorithm
        increments pointer $\texttt{i}$ until $\texttt{i}$ points to an
        element that belongs to a group $\text{A}_p$ with $p \geq m$. For
        each element $x$ along the way that belongs to a group
        $\text{A}_{p'}$ with $p' < m - 1$, it moves $x$ to the place to which 
        $\texttt{b}_{p' + 1}$ points, using a rotate operation to make
        space to accommodate the element, see
        Figure~\ref{fig:partition:k:rotate1} and Lines~$7$--$11$ in
        Algorithm~\ref{algo:exchange:k}. Pointers $\texttt{b}_{p' +
        1}, \ldots, \texttt{b}_{m - 1}$ are incremented afterwards.
        Similarly, the algorithm decrements pointer $\texttt{j}$ until it
        points to an element that belongs to a group $\text{A}_{q}$ with $q
        < m$, moving elements from $\text{A}_{m + 1}, \ldots, \text{A}_{k}$
        along the way in a similar
        fashion, see Figure~\ref{fig:partition:k:rotate2} and Line~$12$--$16$
        in Algorithm~\ref{algo:exchange:k}. If now $\texttt{i} <
        \texttt{j}$, a single rotate operation suffices to move
        the elements referenced by $\texttt{i}$ and $\texttt{j}$ to a
        (temporarily) correct position, see Figure~\ref{fig:partition:k:rotate3}
        and Line~$17$--$20$ in Algorithm~\ref{algo:exchange:k}. Note that 
        any classification strategy can be used in an ``online fashion'' to find
        out element groups in Algorithm~\ref{algo:exchange:k}.

        Figure~\ref{fig:partition:k} shows the idea
        of the algorithm for $k=6$;
        Figures~\ref{fig:partition:k:rotate1}--\ref{fig:partition:k:rotate3} show
        the different rotations being made by Algorithm~\ref{algo:exchange:k} in
        lines 9, 14, and 18. 

        \begin{algorithm}[t!]
        \caption{Move elements by rotations to produce a partition}
    \samepage
    \label{algo:exchange:k}
    \textbf{procedure} \textit{Exchange}$_k$($A[1..n]$)
    \begin{algorithmic}[1]
        \State $\texttt{i} \gets k+1; \texttt{j} \gets n;$
    \State $m \gets \lceil \frac{k+1}{2} \rceil;$
        \State $\texttt{b}_1,\ldots,\texttt{b}_{m - 1} \gets \texttt{i};$
        \State $\texttt{b}_{m},\ldots,\texttt{b}_{k-1} \gets \texttt{j};$
        \State $\texttt{p}, \texttt{q} \gets -1$; \Comment{$\texttt{p}$ and $\texttt{q}$
        hold the group indices of the elements indexed by $\texttt{i}$ and $\texttt{j}$.}
    \While{$\texttt{i} < \texttt{j}$}
        \While{$A[\texttt{i}]$ belongs to group $A_\texttt{p}$ with $\texttt{p} < m$}
            \If{$p < m - 1$}
            \State \texttt{rotate}($\texttt{i}$,$\texttt{b}_{m - 1}, \ldots, \texttt{b}_{\texttt{p} + 1}$);
            \State $\texttt{b}_{\texttt{p} + 1} \texttt{++};\ldots;\texttt{b}_{m - 1}\texttt{++};$
        \EndIf
                \State $\texttt{i}\texttt{++};$
        \EndWhile
        \While{$A[\texttt{j}]$ belongs to group $A_\texttt{q}$ with $\texttt{q} \geq m$}
            \If{$\texttt{q} \geq m + 1$}
            \State \texttt{rotate}($\texttt{j}$,$\texttt{b}_{m},\ldots,\texttt{b}_{\texttt{q} - 1}$);
            \State $\texttt{b}_{\texttt{q}-1}\texttt{-{}-};\ldots;\texttt{b}_{m}\texttt{-{}-};$
        \EndIf
                \State $\texttt{j}\texttt{-{}-}$;
        \EndWhile
            \If{$\texttt{i} < \texttt{j}$}
            \State \texttt{rotate}($\texttt{i}, \texttt{b}_{m - 1}, \ldots,\texttt{b}_{\texttt{q} + 1},\texttt{j},\texttt{b}_{m},\ldots,\texttt{b}_{\texttt{p}-1}$);
            \State $\texttt{i}\texttt{++}; \texttt{b}_{\texttt{q} + 1}\texttt{++}; \ldots;
            \texttt{b}_{m - 1}\texttt{++};$
            \State $\texttt{j}\texttt{-{}-}; \texttt{b}_{m}\texttt{-{}-}; \ldots; \texttt{b}_{\texttt{p} -
        1}\texttt{-{}-};$
        \EndIf
    \EndWhile
    \end{algorithmic}
\end{algorithm}

\begin{figure}
        \centering
            \begin{tikzpicture}[xscale=0.8,yscale=0.6]
            \draw[draw] (0,0) rectangle (15,1);
            \draw[draw] (2,0) -- (2,1);
            \draw[draw] (4,0) -- (4,1);
            \draw[draw] (5,0) -- (5,1);
            \draw[draw] (7,0) -- (7,1);
            \draw[draw] (7.5,0) -- (7.5,1);
            \draw[draw] (10.5, 0) -- (10.5, 1);
            \draw[draw] (11,0) -- (11,1);
            \draw[draw] (12,0) -- (12,1);
            \draw[draw] (14,0) -- (14,1);
            \node (a0) at (1, 0.5) {$A_0$};
            \node (a1) at (3, 0.5) {$A_1$};
            \node (a2) at (4.5, 0.5) {$A_2$};
            \node (a3) at (6, 0.5) {$A_3$};
            \node (a4) at (11.5, 0.5) {$A_4$};
            \node (a5) at (13, 0.5) {$A_5$};
            \node (a6) at (14.5, 0.5) {$A_6$};

            \node (?1) at (7.25, 0.5) {\small ?};
            \node (?2) at (10.75, 0.5) {\small ?};

            \draw[draw,->]  (2.1, -0.7) -- (2.1, -0.1);
            \draw[draw,->](4.1, -0.7) -- (4.1, -0.1);
            \draw[draw,->](5.1, -0.7) -- (5.1, -0.1);
            \draw[draw,->](7.25, -0.7) -- (7.25, -0.1);
            \draw[draw,->](10.75, -0.7) -- (10.75, -0.1);
            \draw[draw,->](11.9, -0.7) -- (11.9, -0.1);
            \draw[draw,->](13.9, -0.7) -- (13.9, -0.1);

            \node at (2.1, -1) {$\texttt{b}_1$};
            \node at (4.1, -1) {$\texttt{b}_2$};
            \node at (5.1, -1) {$\texttt{b}_3$};
            \node at (7.25, -1) {$\texttt{i}$};
            \node at (10.75, -1) {$\texttt{j}$};
            \node at (11.9, -1) {$\texttt{b}_4$};
            \node at (13.9, -1) {$\texttt{b}_5$};
        \end{tikzpicture}
        \caption{General memory layout of Algorithm~\ref{algo:exchange:k} for $k = 6$. Two pointers $\texttt{i}$ and
            $\texttt{j}$ are used to scan the array from left to right and right to left, respectively. Pointers
        $\texttt{b}_1,\ldots,\texttt{b}_{k-1}$ are used to point to the start (resp. end) of segments.}
        \label{fig:partition:k}

        \vspace{1em}

            \begin{tikzpicture}[xscale=0.8,yscale=0.6]
            \draw[draw] (0,0) rectangle (15,1);
            \draw[draw] (2,0) -- (2,1);
            \draw[draw] (4,0) -- (4,1);
            \draw[draw] (5,0) -- (5,1);
            \draw[draw] (7,0) -- (7,1);
            \draw[draw] (7.5,0) -- (7.5,1);
            \draw[draw] (10.5, 0) -- (10.5, 1);
            \draw[draw] (11,0) -- (11,1);
            \draw[draw] (12,0) -- (12,1);
            \draw[draw] (14,0) -- (14,1);
            \node (a0) at (1, 0.5) {$A_0$};
            \node (a1) at (3, 0.5) {$A_1$};
            \node (a2) at (4.5, 0.5) {$A_2$};
            \node (a3) at (6, 0.5) {$A_3$};
            \node (a4) at (11.5, 0.5) {$A_4$};
            \node (a5) at (13, 0.5) {$A_5$};
            \node (a6) at (14.5, 0.5) {$A_6$};

            \node (?1) at (7.25, 0.5) {\small $A_1$};
            \node (?2) at (10.75, 0.5) {\small ?};

            \draw[draw,->]  (2.1, -0.7) -- (2.1, -0.1);
            \draw[draw,->](4.1, -0.7) -- (4.1, -0.1);
            \draw[draw,->](5.1, -0.7) -- (5.1, -0.1);
            \draw[draw,->](7.25, -0.7) -- (7.25, -0.1);
            \draw[draw,->](10.75, -0.7) -- (10.75, -0.1);
            \draw[draw,->](11.9, -0.7) -- (11.9, -0.1);
            \draw[draw,->](13.9, -0.7) -- (13.9, -0.1);

            \node at (2.1, -1) {$\texttt{b}_1$};
            \node at (4.1, -1) {$\texttt{b}_2$};
            \node at (5.1, -1) {$\texttt{b}_3$};
            \node at (7.25, -1) {$\texttt{i}$};
            \node at (10.75, -1) {$\texttt{j}$};
            \node at (11.9, -1) {$\texttt{b}_4$};
            \node at (13.9, -1) {$\texttt{b}_5$};

            \draw[->, thick] (7.25, 0) to[bend left] (4.1,0);
            \draw[->, thick] (4.1, 1) to[bend left] (5.1, 1);
            \draw[->, thick] (5.1, 1) to[bend left] (7.25, 1);
        \end{tikzpicture}
        \caption{The \texttt{rotate} operation in Line~9 of Algorithm~\ref{algo:exchange:k}. An element that belongs to
        group $A_1$ is moved into its respective segment. Pointers $\texttt{i}, \texttt{b}_2, \texttt{b}_3$ are increased by $1$ afterwards.}
        \label{fig:partition:k:rotate1}

        \vspace{1em}

            \begin{tikzpicture}[xscale=0.8,yscale=0.6]
            \draw[draw] (0,0) rectangle (15,1);
            \draw[draw] (2,0) -- (2,1);
            \draw[draw] (4,0) -- (4,1);
            \draw[draw] (5,0) -- (5,1);
            \draw[draw] (7,0) -- (7,1);
            \draw[draw] (7.5,0) -- (7.5,1);
            \draw[draw] (10.5, 0) -- (10.5, 1);
            \draw[draw] (11,0) -- (11,1);
            \draw[draw] (12,0) -- (12,1);
            \draw[draw] (14,0) -- (14,1);
            \node (a0) at (1, 0.5) {$A_0$};
            \node (a1) at (3, 0.5) {$A_1$};
            \node (a2) at (4.5, 0.5) {$A_2$};
            \node (a3) at (6, 0.5) {$A_3$};
            \node (a4) at (11.5, 0.5) {$A_4$};
            \node (a5) at (13, 0.5) {$A_5$};
            \node (a6) at (14.5, 0.5) {$A_6$};

            \node (?1) at (7.25, 0.5) {\small $A_5$};
            \node (?2) at (10.75, 0.5) {\small $A_6$};

            \draw[draw,->]  (2.1, -0.7) -- (2.1, -0.1);
            \draw[draw,->](4.1, -0.7) -- (4.1, -0.1);
            \draw[draw,->](5.1, -0.7) -- (5.1, -0.1);
            \draw[draw,->](7.25, -0.7) -- (7.25, -0.1);
            \draw[draw,->](10.75, -0.7) -- (10.75, -0.1);
            \draw[draw,->](11.9, -0.7) -- (11.9, -0.1);
            \draw[draw,->](13.9, -0.7) -- (13.9, -0.1);

            \node at (2.1, -1) {$\texttt{b}_1$};
            \node at (4.1, -1) {$\texttt{b}_2$};
            \node at (5.1, -1) {$\texttt{b}_3$};
            \node at (7.25, -1) {$\texttt{i}$};
            \node at (10.75, -1) {$\texttt{j}$};
            \node at (11.9, -1) {$\texttt{b}_4$};
            \node at (13.9, -1) {$\texttt{b}_5$};

            \draw[->, thick] (10.75, 0) to[bend right] (13.9, 0);
            \draw[->, thick] (13.9, 1) to[bend right] (11.9, 1);
            \draw[->, thick] (11.9, 1) to[bend right] (10.75, 1);
        \end{tikzpicture}
        \caption{The \texttt{rotate} operation in Line~14 of Algorithm~\ref{algo:exchange:k}. An element that
            belongs to group $A_6$ is moved into its respective segment. Pointers $\texttt{j}, \texttt{b}_4, \texttt{b}_5$ are decreased by $1$
        afterwards.}
        \label{fig:partition:k:rotate2}

        \vspace{1em}

            \begin{tikzpicture}[xscale=0.8,yscale=0.6]
            \draw[draw] (0,0) rectangle (15,1);
            \draw[draw] (2,0) -- (2,1);
            \draw[draw] (4,0) -- (4,1);
            \draw[draw] (5,0) -- (5,1);
            \draw[draw] (7,0) -- (7,1);
            \draw[draw] (7.5,0) -- (7.5,1);
            \draw[draw] (10.5, 0) -- (10.5, 1);
            \draw[draw] (11,0) -- (11,1);
            \draw[draw] (12,0) -- (12,1);
            \draw[draw] (14,0) -- (14,1);
            \node (a0) at (1, 0.5) {$A_0$};
            \node (a1) at (3, 0.5) {$A_1$};
            \node (a2) at (4.5, 0.5) {$A_2$};
            \node (a3) at (6, 0.5) {$A_3$};
            \node (a4) at (11.5, 0.5) {$A_4$};
            \node (a5) at (13, 0.5) {$A_5$};
            \node (a6) at (14.5, 0.5) {$A_6$};

            \node (?1) at (7.25, 0.5) {\small $A_5$};
            \node (?2) at (10.75, 0.5) {\small $A_1$};

            \draw[draw,->]  (2.1, -0.7) -- (2.1, -0.1);
            \draw[draw,->](4.1, -0.7) -- (4.1, -0.1);
            \draw[draw,->](5.1, -0.7) -- (5.1, -0.1);
            \draw[draw,->](7.25, -0.7) -- (7.25, -0.1);
            \draw[draw,->](10.75, -0.7) -- (10.75, -0.1);
            \draw[draw,->](11.9, -0.7) -- (11.9, -0.1);
            \draw[draw,->](13.9, -0.7) -- (13.9, -0.1);

            \node at (2.1, -1) {$\texttt{b}_1$};
            \node at (4.1, -1) {$\texttt{b}_2$};
            \node at (5.1, -1) {$\texttt{b}_3$};
            \node at (7.25, -1) {$\texttt{i}$};
            \node at (10.75, -1) {$\texttt{j}$};
            \node at (11.9, -1) {$\texttt{b}_4$};
            \node at (13.9, -1) {$\texttt{b}_5$};

            \draw[->, thick] (7.25, 0) to[bend right] (11.9, 0);
            \draw[->, thick] (11.9, 1) to[bend right] (10.75, 1);
            \draw[->, thick] (10.75, 1) to[out=165,in=15] (4.1, 1);
            \draw[->, thick] (4.1, 0) to[bend right] (5.1, 0);
            \draw[->, thick] (5.1, 0) to[bend right] (7.25, 0);
        \end{tikzpicture}
            \caption{Example for the rotate operation in Line~18 of
                Algorithm~\ref{algo:exchange:k}. The element found at $\texttt{i}$ is
            moved into its specific segment. Subsequently, the element found at
        $\texttt{j}$ is moved into its specific segment.}
        \label{fig:partition:k:rotate3}

    \end{figure}

    \subsection{Cost Measures and Assumptions of the Analysis}
    In the following we consider three cost measures as cost for rearranging the input using Algorithm~\ref{algo:exchange:k}.  The first two cost measures aim to 
    describe the memory behavior of Algorithm~\ref{algo:exchange:k}. The first measure counts how often each array cell is accessed during rearranging the input, which in practice
    gives a good approximation on the time the CPU has to wait for memory, even when the data is in cache. We will show later that this theoretical cost measures allows us
    to describe practical cost measures like the average number of cache misses accurately.  The second cost measure counts how often the algorithm 
    writes into an array cell. 
The last cost measure is more classical and counts how many assignments the algorithm makes. It will be interesting to see that while these cost measures appear to be similar, only the first  one will correctly reflect advantages of a multi-pivot quicksort approach in empirical running time.   The first cost measure was also considered for the YBB algorithm in \citeN{NebelWM15}. 

    \paragraph{Scanned Elements} Assume that a pointer $\texttt{l}$ is
    initialized with value $l_s$. Let $l_e$ be the value in $\texttt{l}$
    after the algorithm finished rearranging the input. Then we define
    $\text{cost}(\texttt{l}) = | l_s - l_e|$, i.e., the number of array cells
    inspected by pointer $\texttt{l}$. (Note that a cell is accessed only once per 
    pointer, all pointers move by increments or decrements of $1$, and $A[l_e]$ is not inspected.) Let the variable {\ppv} be the
    \emph{number of scanned elements} of Algorithm~\ref{algo:exchange:k}. It is
    the sum of the costs of pointers $\texttt{i}, \texttt{j}, \texttt{b}_1,
    \ldots, \texttt{b}_{k - 1}$.  From an empirical point of view this cost
    measure gives a lower bound on the number of clock cycles the CPU spends
    waiting for memory.  It can also be used to predict the cache behavior
    of Algorithm~\ref{algo:exchange:k}. We will see that it gives good
    estimates for the cache misses in L$1$ cache which we observed in our
    experiments. This has also been observed for 
    the YBB algorithm in \cite{NebelWM15}.

    \paragraph{Write Accesses} Each rotate operation of $\ell$ elements of Algorithm~\ref{algo:exchange:k} writes into
    exactly $\ell$ array cells. When we assign a value to an array, we call the access to this array cell a \emph{write access}. Let the variable {\pwa} be the \emph{number of write accesses} to array cells (over all rotate operations). 

    \paragraph{Assignments} Each rotate operation of $\ell$ elements of
    Algorithm~\ref{algo:exchange:k} makes exactly $\ell + 1$ assignments. Let
    the variable {\pas} be the \emph{number of assignments} over all rotate
    operations. Since each swap operation consists of three assignments, 
    this is the most classical cost measure with respect to the three cost 
    measures introduced above for the analysis of quicksort. 

    \paragraph{Setup of the Analysis}     In the following we want to obtain the leading term for the average number of scanned elements, write accesses, and assignments, both
    for partitioning and over the whole sorting process. 
    The input is again assumed to be a random permutation
    of the set $\{1,\ldots,n\}$ which resides in an array $A[1..n]$.  Fix an
    integer $k \geq 1$.  The first $k$ elements are chosen as pivots. 
    Then we can think of the input consisting of $n - k$ elements having labels from $\text{A}_0, \ldots, \text{A}_k$,
    and our goal is to rearrange the input. (In terms of multi-pivot quicksort, 
    our goal is to obtain a partition of the input, as depicted in
    Figure~\ref{fig:partition} on Page~\pageref{fig:partition}. However, here
    determining to which of the groups $\text{A}_0,\ldots, \text{A}_k$ element $A[i]$ belongs
    is for free.) We are interested
    in the cost of the rearrangement process and the total sorting cost in the cost measures introduced above.

    \paragraph{From Partitioning Cost to Sorting Cost} Let $P_n$
denote the partitioning cost that the algorithm incurs in the
first partitioning/rearrangement step. Let the random variable $C_n$ count
the sorting cost (over the whole recursion)
of sorting an input of length $n$ in the respective cost measure. As before, we get the recurrence:

\begin{align*}
    \E(C_n) = \E(P_n) + \frac{1}{\binom{n}{k}} \sum_{a_0 + \cdots + a_k = n -k} \left(
    \E(C_{a_0}) + \cdots + \E(C_{a_k})\right).
\end{align*}
Again, this recurrence has the form of \eqref{eq:k:pivot:recurrence}, so we may apply
\eqref{eq:k:pivot:recurrence:solution} for linear partitioning cost. Thus, from now on we focus
on a single partitioning step.

    \subsection{Analysis}

    Our goal in this section is to prove the following theorem. A discussion
    of this result will be given in the next section.

    \begin{theorem}
    Let $k \geq 1 $ be the number of pivots and $m = \lceil \frac{k + 1}{2} \rceil$. Then 
    for Algorithm~\ref{algo:exchange:k}, we have that
    \begin{align}
        \E(\ppv_n) &=
        \begin{dcases}
            \frac{m + 1}{2} \cdot n + O(1), & \text{ \quad\quad for odd $k$},\\
            \frac{m^2}{2m - 1} \cdot n + O(1), & \text{ \quad\quad for even $k$},
        \label{eq:memory:accesses:exchange:k}
    \end{dcases}\\[1em]
        \E(\pwa_n) &=
        \begin{dcases}
            \frac{2m^3 + 3m^2 - m - 2}{2m(2m+1) } \cdot n + O(1), & \text{ \hspace*{0.12em} for odd $k$},\\
            \frac{2m^3-2m - 1}{2m(2m-1)} \cdot n + O(1), & \text{ \hspace*{0.12em} for even $k$},
        \label{eq:write:accesses:exchange:k}
    \end{dcases}\\[1em]
        \E(\pas_n) &=
        \begin{dcases}
            \frac{2m^3+6m^2-m-4}{2m(2m+1)} \cdot n + O(1), & \text{\hspace*{-0.1em} for odd $k$},\\
            \frac{2m^3+3m^2-5m-2}{2m(2m-1)} \cdot n + O(1), & \text{\hspace*{-0.1em} for even $k$},
        \label{eq:assignments:exchange:k}
        \end{dcases}
    \end{align}
    \label{thm:partition:cost}
    \end{theorem}

    From this theorem, we get the leading term of the total sorting cost in the respective 
    cost measure by applying \eqref{eq:k:pivot:recurrence:solution}.

    \paragraph{Scanned Elements}
\label{sec:additional:cost:measures:cache:misses}
We will first study how many elements are scanned by the
pointers used in
Algorithm~\ref{algo:exchange:k}
when sorting an input.

Let the pivots and thus $a_0,\ldots,a_k$ be
fixed. The pointers $\texttt{i}$ and $\texttt{j}$ together scan the whole array, and thus
inspect $n - k$ array cells.  When Algorithm~\ref{algo:exchange:k} terminates,
$\texttt{b}_1$ points to $A[k + a_0 + 1]$, having visited exactly $a_0$ array
cells. An analogous statement can be made for the pointers $\texttt{b}_2,
\ldots, \texttt{b}_{k-1}$.  On average, we have $(n-k)/(k + 1)$ elements of
each group $A_0,\ldots, A_k$, so $\texttt{b}_1$ and $\texttt{b}_{k-1}$ each
visit $(n-k)/(k+1)$ array cells on average, $\texttt{b}_2$ and $\texttt{b}_{k -
2}$ each visit $2 (n-k)/ (k+1)$ array cells, and so on.

For the average number of scanned elements in a partitioning step we
consequently get
\begin{align}
    \E(\ppv_n) =
    \begin{dcases}
        2 \cdot \sum_{i = 1}^{\lceil k/2\rceil} \frac{i \cdot (n-k)}{k+1}, &
        \text{ for odd $k$} ,\\
        2 \cdot \sum_{i = 1}^{k/2} \frac{i \cdot (n-k)}{k+1} +
        \frac{k/2+1}{k+1}\cdot (n-k), & \text{ for even $k$},
    \end{dcases}
    \label{eq:memory:accesses:1}
\end{align}
and a simple calculation shows
\begin{align}
    \E(\ppv_n) =
    \begin{dcases}
            \frac{m + 1}{2} \cdot (n-k), & \text{ for odd $k$},\\
            \frac{m^2}{2m - 1}  \cdot (n-k), & \text{ for even $k$}.
        \end{dcases}
\end{align}

\paragraph{Write Accesses} We now focus on the average number of write accesses. 
First we observe that a rotate operation involving $\ell$ elements in Algorithm~\ref{algo:exchange:k} makes exactly
$\ell$ element scans and $\ell$ write accesses. So, the only difference between element scans and write 
accesses is that whenever pointer $\texttt{i}$ finds an $\text{A}_{m - 1}$-element in Line~7 or pointer $\texttt{j}$ 
finds an $\text{A}_{m}$-element in Line~12, the element is scanned but no write access takes place. Let
$C_{\texttt{i}, m - 1}$ be the random variable that counts the number of $\text{A}_{m - 1}$-elements found in Line~7, and let 
$C_{\texttt{j}, m }$ be the random variable that counts the number of $\text{A}_{m}$-elements found in Line~12 of 
Algorithm~\ref{algo:exchange:k}.

Thus, we know that 
\begin{align}
    \E(\pwa_n) = \E(\ppv_n) - \E(C_{\texttt{i},m - 1}) - \E(C_{\texttt{j}, m}).
    \label{eq:from:scanned:elements:to:write:accesses}
\end{align}

\begin{lemma}
Let $k$ be the number of pivots and let $m = \lceil \frac{k + 1}{2}\rceil$. Then
\begin{align*}
    \E(C_{\textnormal{\texttt{i}}, m - 1}) + \E(C_{\textnormal{\texttt{j}}, m}) = 
        \begin{dcases}
            \frac{m + 1}{m(2m + 1)}\cdot n + O(1), & \text{ for $k$ odd},\\
            \frac{2m + 1}{2m(2m - 1)}\cdot n + O(1), & \text{ for $k$ even}.
    \end{dcases}
\end{align*}
\label{lem:avg:correctly:placed:elements}
\end{lemma}

\begin{proof}
    We start by obtaining bounds on $\E(C_{\texttt{i}, m - 1})$ and $\E(C_{\texttt{j}, m})$ when $k$ is even. The calculations for the case 
when $k$ is odd are simpler because of symmetry. In the calculations, we will consider 
the two events that the groups $\text{A}_0, \ldots, \text{A}_{m-1}$ 
have $L$ elements in total, for $0 \leq L \leq n - k$, and that group $\text{A}_{m - 1}$ has $K$ elements, for $0 \leq K \leq L$. 
If the group sizes are as above, then the expected number of $\text{A}_{m - 1}$ elements scanned by pointer $\texttt{i}$ 
is $L \cdot K / (n - k)$. We first observe that

\begin{align}
    \E(C_{\texttt{i}, m - 1}) &= \sum_{L = 0}^{n} \sum_{K = 0}^L \Pr(a_0 + \cdots + a_{m - 1} = L \wedge a_{m - 1} = K) \cdot L \cdot \frac{K}{n - k}\notag\\
                     &\stackrel{(\ast)}{=} \frac{1}{n \cdot m} \sum_{L = 0}^n \Pr(a_0 + \cdots + a_{m - 1} = L) \cdot L^2 + O(1)\notag\\
                     &= \frac{1}{n\cdot m} \sum_{L = 1}^n \frac{\binom{L - 1}{m - 1}\binom{n - L}{m - 2}}{\binom{n}{2(m - 1)}} \cdot L^2 + O(1),\label{eq:proof:write:accesses:1}
\end{align}
where $(\ast)$ follows by noticing that $\sum_{K = 0}^L \Pr(a_{m - 1} = K \mid a_0 + \cdots + a_{m - 1} = L) \cdot K$ is the expected size of the group $A_{m - 1}$ given that the first $m$ groups have exactly $L$ elements, which is $L/m$. 
We calculate the sum of binomial coefficients as in \eqref{eq:proof:write:accesses:1} for a more general situation: 
\begin{claim} Let $\ell_1$ and $\ell_2$ be arbitrary integers. Then we have 
 \begin{align*}
     &\textnormal{(i)} \sum_{L = 1}^n \frac{\binom{L-1}{\ell_1}\binom{n - L}{\ell_2}}{\binom{n}{\ell_1 + \ell_2 + 1}} \cdot L^2 
     = \frac{(\ell_1 + 1)(\ell_1 + 2) \cdot (n + 2) (n + 1)}{(\ell_1 + \ell_2 + 2) (\ell_1 + \ell_2 + 3)} - \frac{(\ell_1 + 1) \cdot( n + 1)}{(\ell_1 + \ell_2 + 2)}.\\
          &\textnormal{(ii)} \sum_{L = 1}^n \frac{\binom{L-1}{\ell_1}\binom{n - L}{\ell_2}}{\binom{n}{\ell_1 + \ell_2 + 1}} \cdot (n - L)^2 
     = \frac{(\ell_2 + 1) (\ell_2 + 2) \cdot (n + 2)(n+1)}{(\ell_1 + \ell_2 + 2)(\ell_1 + \ell_2 + 3)} - \frac{3(\ell_2 + 1) \cdot (n+1)}{\ell_1 + \ell_2 + 2} + 1.
\end{align*}    
\label{claim:binomial:coefficients}
\end{claim}
\begin{proof}
    We denote by $n^{\underline{k}}$ the $k$-th falling factorial of $n$, i.e., $n (n - 1) \cdots (n - k +1)$. 
    Using known identities for sums of binomial coefficients, we may calculate
    \begin{align*}
        &\sum_{L = 1}^n \frac{\binom{L - 1}{\ell_1}\binom{n - L}{\ell_2}}{\binom{n}{\ell_1 + \ell_2 + 1}} \cdot L^2\\ 
        & = \frac{1}{\binom{n}{\ell_1 + \ell_2 + 1}} \sum_{L = 1}^n \left((\ell_1 + 1) (\ell_1 + 2) \binom{L + 1}{\ell_1 + 2} \binom{n - L}{\ell_2} - L \binom{L - 1}{\ell_1}\binom{n - L}{\ell_2}\right)\\
        & = \frac{1}{\binom{n}{\ell_1 + \ell_2 + 1}} \sum_{L = 1}^n \left((\ell_1 + 1) (\ell_1 + 2) \binom{L + 1}{\ell_1 + 2} \binom{n - L}{\ell_2} - (\ell_1 + 1) \binom{L}{\ell_1 + 1}\binom{n - L}{\ell_2}\right)\\
        & \stackrel{(\ast)}{=} \frac{1}{\binom{n}{\ell_1 + \ell_2 + 1}} \left((\ell_1 + 1) (\ell_1 + 2) \binom{n + 2}{\ell_1 + \ell_2 + 3} - (\ell_1 + 1) \binom{n  + 1}{\ell_1 + \ell_2 + 2}\right)\\
        & = \frac{(\ell_1 + 1) (\ell_1 + 2) (n + 2)^{\underline{\ell_1 + \ell_2 + 3}} (\ell_1 + \ell_2 + 1)!}{(\ell_1 + \ell_2 + 3)! \cdot n^{\underline{\ell_1 + \ell_2 + 1}}} - \frac{(\ell_1 + 1) (n + 1)^{\underline{\ell_1 + \ell_2 + 2}} (\ell_1 + \ell_2 + 1)!}{(\ell_1 + \ell_2 + 2)! \cdot n^{\underline{\ell_1 + \ell_2 + 1}}}\\
        & = \frac{(\ell_1 + 1)(\ell_1 + 2) (n + 2) (n + 1)}{(\ell_1 + \ell_2 + 2)(\ell_1 + \ell_2 + 3)} - \frac{(\ell_1 + 1) ( n + 1)}{(\ell_1 + \ell_2 + 2)},
    \end{align*}
    where $(\ast)$ follows by using the identity \cite[(5.26)]{GrahamKP}. 
    The calculations for (ii) are analogous by using an index transformation $K = n - L$. 
\end{proof} 
 Using the claim, we continue from \eqref{eq:proof:write:accesses:1} as follows:
\begin{align}
    \E(C_{\texttt{i}, m - 1}) &= \frac{m(m+1)\cdot(n+1)(n+2)}{nm\cdot (2m-1)2m}+ O(1) = \frac{m + 1}{2m(2m - 1)}n + O(1).
\label{eq:proof:write:accesses:cim}
\end{align}
By similar arguments, we obtain
\begin{align}
    \E(C_{\texttt{j}, m }) &= \frac{1}{n \cdot (m - 1) } \sum_{L = 0}^n \frac{\binom{L}{m - 1}\binom{n - L}{m - 2}}{\binom{n}{2(m - 1)}} \cdot (n - L)^2 + O(1)\notag\\
                  &= \frac{(m - 1) m}{2nm (m - 1) (2m - 1)}n^2 + O(1) = \frac{1}{2(2m - 1)}n + O(1). 
    \label{eq:proof:write:accesses:cjm}
\end{align}
Thus, in the asymmetric case it holds that $\E(C_{\texttt{i}, m - 1}) + \E(C_{\texttt{j}, m}) = \frac{2m + 1}{2m(2m - 1)}n + O(1)$.
\end{proof}
Applying Lemma~\ref{lem:avg:correctly:placed:elements} to \eqref{eq:memory:accesses:exchange:k} and \eqref{eq:from:scanned:elements:to:write:accesses}
and simplifying gives us the value from Theorem~\ref{thm:partition:cost}.

\paragraph{Assignments} To count the total number of assignments, we first observe that each rotate operation that involves
$\ell$ elements makes $\ell$ write accesses and $\ell + 1$ assignments. Thus, the total number of assignment is the
sum of the number of write accesses and the number of rotate operations. So we observe 
\begin{align}
    \E(\pas_n) = \E(\pwa_n) + \E(\text{\#rotate operations}).
    \label{eq:from:write:accesses:to:assignments}
\end{align}

\begin{lemma}
    Let $k$ be the number of pivots and $m = \lceil \frac{k+1}{2}\rceil$. Then it holds that
    \begin{align*}
        \E(\textnormal{\#rotate operations}) = 
            \begin{cases}
                \frac{3m^2- 2}{2m(2m+1)},& \text{for odd $k$},\\
                \frac{3m^2-3m - 1}{2m(2m-1)},& \text{for even $k$}.
            \end{cases}
    \end{align*}
    \label{lem:avg:rotate:operations}
\end{lemma}

\begin{proof}
    The number of rotate operations is counted as follows. For each non-$\text{A}_{m - 1}$ element that is scanned by pointer $\texttt{i}$, 
a rotate operation is invoked (Line~9 and Line~18 in Algorithm~\ref{algo:exchange:k}). In addition, each $A_{m'}$ element with $m' > m$ scanned by pointer $\texttt{j}$ invokes a rotate operation (Line~14 in Algorithm~\ref{algo:exchange:k}). So, the number of rotate operations is the sum of these two quantities. Again, we focus on 
the case that $k$ is even. Let $C_{\texttt{i}, <m - 1}$ be the number of $A_{m'}$ elements with $m' < m - 1$ scanned by pointer $\texttt{i}$. Define $C_{\texttt{i}, > m - 1}$ and $C_{\texttt{j}, > m}$ analogously. By symmetry (cf. \eqref{eq:proof:write:accesses:1}) we have that  
$$\E(C_{\texttt{i}, < m - 1}) = (m - 1) \cdot \E(C_{\texttt{i}, m - 1}) = \frac{(m - 1) ( m + 1)}{2 m (2m - 1)}n + O(1),$$ 
see \eqref{eq:proof:write:accesses:cim}. Furthermore, 
since we expect that pointer $\texttt{i}$ scans $\frac{m}{2m - 1} (n - k)$ elements, 
we know that $$\E(C_{\texttt{i}, > m - 1}) = \frac{m}{2m - 1} (n - k) - m \cdot \E(C_{\texttt{i}, m - 1}) = \left(\frac{m}{2m - 1} - \frac{m + 1}{2(2m - 1)}\right) n + O(1).$$ 
Finally, again by symmetry we obtain
$$\E(C_{\texttt{j}, > m}) = (m - 2) \cdot \E(C_{\texttt{j}, m}) = \frac{m - 2}{2(2m - 1)}n + O(1).$$ For even $k$ the result now follows by adding
these three values. For odd $k$, we only have to adjust that we expect that pointer $\texttt{i}$ scans $(n - k)/2$ elements, and that there are $m - 1$ groups $\text{A}_{m + 1}, \ldots, \text{A}_k$ when calculating $\E(C_{\texttt{j}, > m})$. 
 \end{proof}
Applying Lemma~\ref{lem:avg:rotate:operations} to \eqref{eq:from:write:accesses:to:assignments} and simplifying gives the value from 
Theorem~\ref{thm:partition:cost}.

\subsection{Discussion and Empirical Validation} Using the formulae developed in the previous subsection
we calculated the average number of scanned elements, write accesses, and assignments in partitioning and in sorting for $k \in
\{1,\ldots,9\}$ using Theorem~\ref{thm:partition:cost} and \eqref{eq:k:pivot:recurrence:solution}.
Figure~\ref{fig:sorting:cost} gives a graphical overview 
of these calculations. Next, we will discuss our findings. 

\begin{figure}[t]
    \centering
\begin{tikzpicture}
  \begin{axis}[
    xlabel={\# pivots},
    ylabel={avg. cost $/n \ln n$},
    height=5cm,
    width=11cm,
    legend style = { at = {(-.0,1.12)}, anchor=west, draw=none},
    cycle list name = black white,
    legend columns = 3
    ]
    \addplot coordinates { (1, 2) (2,
        1.6) (3, 1.385) (4, 1.402) (5, 1.379) (6, 1.435) (7, 1.455) (8, 1.519)
    (9, 1.5555)};
    \addlegendentry{scanned elements}
    \addplot coordinates { (1, 0.667) (2,
        1.1) (3, 1.108) (4, 1.221) (5, 1.248) (6, 1.334) (7, 1.374) (8, 1.452)
    (9, 1.499)};
    \addlegendentry{write accesses}
        \addplot coordinates { (1, 1) (2,
        1.6) (3, 1.569) (4, 1.662) (5, 1.658) (6, 1.726) (7, 1.746) (8, 1.810)
    (9, 1.843)};
    \addlegendentry{assignments}
\end{axis}
\end{tikzpicture}
\caption{Average number of scanned elements, write accesses, and assignments for 
    sorting an input of length $n$ using 
Algorithm~\ref{algo:exchange:k}.}
\label{fig:sorting:cost}
\end{figure}
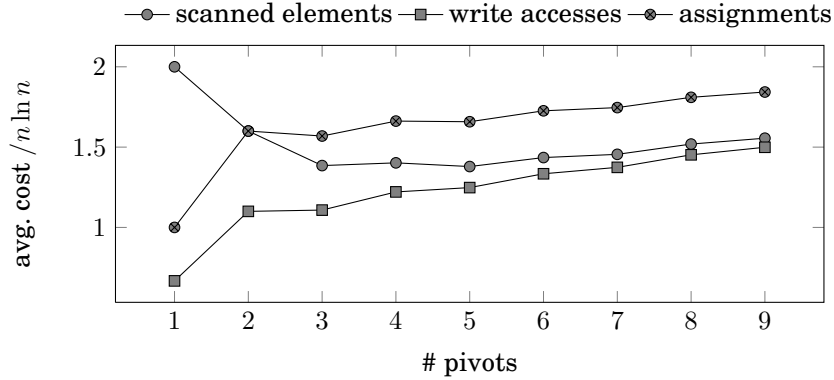

Interestingly, Algorithm~\ref{algo:exchange:k} improves over classical
quicksort when using more than one pivot with regard to scanned elements. A $3$-pivot
quicksort algorithm, using this partitioning algorithm, has lower cost ($
1.385 n \ln n$) than
classical ($2 n \ln n$) and dual-pivot quicksort ($1.6 n \ln n$).  Moreover, the average number of
scanned elements is minimized by the $5$-pivot partitioning algorithm ($1.379 n
\ln n$ scanned elements on average).  However, the
difference to the $3$-pivot algorithm is small. Using more
than five pivots increases the average number of scanned elements. With respect 
to write accesses and assignments, we see a different picture. In both cost measures, 
the average sorting cost increases considerably when moving from classical
quicksort to quicksort variants with at least two pivots.  For a growing number of pivots, it
slowly increases further. In conclusion,
Algorithm~\ref{algo:exchange:k} benefits from using more than one pivot only with 
respect to scanned elements, but not with respect to the average number of write accesses and assignments. 

In \cite[Chapter 7]{Aumueller15} two very different partitioning
algorithms are considered. They are based on the super scalar sample sort algorithm of \cite{SandersW04},
which is an implementation of samplesort \cite{FrazerK70}. These algorithms
classify the elements in a first pass and use these classifications to move
the elements efficiently around in a second pass. While these approaches lead to
higher cost for a small number of pivots, they theoretically outperform the 
generalized partitioning algorithm considered here for, e.g., 127 pivots. 

\paragraph{Running Time Implications}
We now ask what the considerations made so far mean for empirical running time.
Since each memory access,
even if it can be served from L$1$ cache, is much more expensive than other
operations like simple subtraction, addition, or assignments on or between
registers, the results for the cost measure ``scanned elements'' show that 
for different multi-pivot quicksort algorithms we must expect
big differences in the time the CPU has to
wait for memory.\footnote{As an example,
    the Intel i7 used in our experiments needs at least 4 clock cycles to read
    from memory if the data is in L1 cache and its physical address is known. 
    If the data is in L2 cache but not in L1 cache, there is an additional penalty of $6$ clock cycles. 
    On the other hand, three (data-independent) \texttt{MOV} operations between registers
    on the same core take only 1 clock cycle. See \cite{Fog14} for more details.}
If in addition writing an element back into cache/memory is more expensive than reading from the cache 
(as it could happen with the ``write-through'' cache strategy), then the calculations show that we should
not expect advantages of multi-pivot quicksort algorithms over classical quicksort in terms of memory behavior. 
However, cache architectures in modern CPUs apply the ``write-back'' strategy which does not add a penalty 
to running time for writing into memory.

As is well known from classical quicksort and dual-pivot quicksort, the influence of
lower order terms cannot be neglected for real-world values of $n$. Next, we will 
validate our findings for practical values of $n$. 

\paragraph{Empirical Validation}
We implemented Algorithm~\ref{algo:exchange:k} and ran it for different input 
lengths and pivot numbers. In the experiments, we sorted inputs of size $2^i$ with 
$9 \leq i \leq 27$. Each data point is the average over $600$ trials. For 
measuring cache misses we used the
``performance application programming interface'' (PAPI), which is available
at {\tt http://icl.cs.utk.edu/papi/}.

Intuitively, fewer scanned elements should yield better cache
behavior when memory accesses are done ``scan-like'' as in the algorithms considered
here. The argument used in \cite{LaMarcaL99} and
\cite{Kushagra14} is as follows: When each of the $m$ cache memory blocks holds
exactly $B$ keys, then a scan of $n'$ array cells (that have never been
accessed before) incurs $\lceil n'/B\rceil$
cache misses. 
Now we check whether
the assertion that partitioning an input of $n$ elements using
Algorithm~\ref{algo:exchange:k}  incurs
$\lceil \E\left(\ppv_n\right) / B \rceil$ cache misses is justifiable. (Recall that
$\E\left(\ppv_n\right)$ 
is the average number of scanned elements during
partitioning.) In
the experiment, we partitioned $600$ inputs consisting of $n=2^{27}$ items using
Algorithm~\ref{algo:exchange:k}, for $1, 2,  5,$ and $9$ pivots. The measurements
with respect to L$1$
cache misses
are shown in Table~\ref{tab:cache:misses:partitioning}. In our setup, each 
L$1$ cache line contains $8$ elements. So, by~\eqref{eq:memory:accesses:exchange:k}
Algorithm~\ref{algo:exchange:k} should theoretically incur
$0.125n$, $0.166n$, $0.25n$, and $0.375n$ L$1$ cache misses for $k \in
\{1,2,5,9\}$, respectively. The results
from Table~\ref{tab:cache:misses:partitioning} show that the empirical measurements are very close
to these values.

\begin{table}
    \tbl{Cache misses incurred by
    Algorithm~\ref{algo:exchange:k} (``Exchange$_k$'') in a single partitioning step.
    All values are
    averaged over $600$ trials.
\label{tab:cache:misses:partitioning}}
    {
    \begin{tabular}{r  r  r r  r }
        \toprule
        \textbf{Algorithm} & Exchange$_1$ & Exchange$_2$ & Exchange$_5$ & Exchange$_9$ \\ \hline
        \textbf{avg. L1 misses / $n$} & $0.125$ & $0.163$ & $0.25$ & $0.378$ \\
        \bottomrule
        \\
    \end{tabular}
}
\end{table}

\begin{table}
    \tbl{Average number of L$1$/L$2$ cache misses compared to the average number of scanned elements for
    sorting inputs of size $n=2^{27}$. Cache misses are scaled
by $n \ln n$ and are averaged over $600$ trials. In parentheses, we show the ratio
    to the best algorithmic variant of Algorithm~\ref{algo:exchange:k} w.r.t. memory/cache behavior ($k = 5$),
calculated from the non-truncated experimental data.
\label{tab:cache:misses}}
    {
    \begin{tabular}{lrrr }
        \toprule
        \textbf{Algorithm} & $\mathbf{\E(\pv_n)}$ & \textbf{L1 Cache Misses}
        & \textbf{L2 Cache Misses}
        \\
        \hline
        Exchange$_1$ & $2.000 n \ln n$ ($+\phantom{0}45.0\%$)
        & $0.140 n \ln n$ ($+\phantom{0}48.9\%$)
        & $0.0241 n \ln n$ ($+263.1\%$)
        \\
        Exchange$_2$ & $1.600 n \ln n$ ($+\phantom{0}16.0\%$)
        & $0.110 n \ln n$ ($+\phantom{0}16.9\%$)
        & $0.0124 n \ln n$ ($+\phantom{0}86.8\%$)
        \\
        Exchange$_3$ & $1.385 n \ln n$ ($+\phantom{00}0.4\%$)                 & $0.096 n \ln n$ ($+\phantom{00}1.3\%$)
        & $0.0080 n \ln n$ ($+\phantom{0}19.8\%$)
        \\
        Exchange$_5$ & $1.379 n \ln n$ ($\phantom{006}$---$\phantom{0\%}$)
        & $0.095 n \ln n$ ($\phantom{006}$---$\phantom{0\%}$)
        & $0.0067 n \ln n$ ($\phantom{006}$---$\phantom{0\%})$
        \\
        Exchange$_7$ & $1.455 n \ln n$ ($+\phantom{00}5.5\%$)                 & $0.100 n \ln n$ ($+\phantom{00}5.3\%$)
        & $0.0067 n \ln n$ ($+\phantom{00}0.7\%$)
        \\
        Exchange$_9$ & $1.555 n \ln n$ ($+\phantom{0}12.8\%$)
        & $0.106 n \ln n$ ($+\phantom{0}12.2\%$)
        & $0.0075n \ln n$ ($+\phantom{0}12.9\%$)\\
        \bottomrule
        \\
    \end{tabular}
}
\end{table}

Table~\ref{tab:cache:misses} shows the exact measurements regarding L$1$ and L$2$
cache misses for sorting $600$ random inputs  consisting of $n=2^{27}$ elements using Algorithm~\ref{algo:exchange:k}
and relates them to each other.\footnote{We omit our measurements for L3 cache
    misses. In contrast to our measurements for L1 and L2 cache misses, 
    the measurements on this level of the hierarchy were very different in each
run. We believe this is due to L3 caches being shared among cores.} The
figures indicate that the
relation with respect to the measured number of L$1$ cache misses of the different
algorithms reflect their relation with respect to the average number of
scanned elements \emph{very well}. However, while the average number of cache misses correctly
reflects the relative relations, the measured values (scaled by $n \ln n$)
are lower than we would expect by simply dividing $\E(\pv_n)$ by the
block size $B$. We give the following heuristic argument that adjusts 
the theoretical formulas to this
effect. First, it is no surprise that the values are too high since the
recursion to compute the sorting cost changes if true cache misses are counted. Once the whole  
input is contained in the L1 cache, no cache misses are involved. 
So, the actual cost in terms of cache misses for sorting these inputs is $0$.
Let the average number of scanned elements during partitioning be $a n + O(1)$. 
Then we adjust our estimate on the average number of L1 cache misses by
subtracting  $a/(B(\HH_{k + 1} - 1))n
\ln M$ from our estimate $\E(\pv_n)/B$, for a constant $M$ that
depends on the actual machine.\footnote{It is well known for classical quicksort
and dual-pivot quicksort \cite{WildNN15} that changing the recurrence to charge cost 
$P_{n'} = 0$ for sorting inputs of length $n'$, for $n' \leq M$, influences 
only the linear term in the average sorting cost. For example, \citeN{WildNN15}
showed that for the dual-pivot
quicksort recurrence the solution for average partitioning cost $an + O(n^{1-
\varepsilon})$ is $(6/5)a n \ln (n/(M+2)) + O(n)$.
We did not check the exact influence of stopping the 
recursion for the multipivot recursion at these input sizes, but assume that it behaves similarly.}
Numerical
computations for the measured number of L1 cache misses showed that setting $M =
3751$, which is very close to the number of cache lines in the L1 cache of the
machine used measuring cache misses, gives results very close to the measured number 
of L1 cache misses. We remark that this approach does not give good estimates 
for L2 cache misses.

In
summary, scanned elements appear to be a suitable cost measure to predict the L$1$
cache behavior of Algorithm~\ref{algo:exchange:k}.
However, this is not true with regard to L$2$ cache
behavior of these algorithms, as shown in Table~\ref{tab:cache:misses}.

Figure~\ref{fig:assignments} shows the measurements we got with regard to
the average number of assignments. We see that the measurements agree with our theoretical study (\emph{cf.}
Fig.~\ref{fig:sorting:cost}). In particular, lower order terms seem to have low
influence on the sorting cost. 

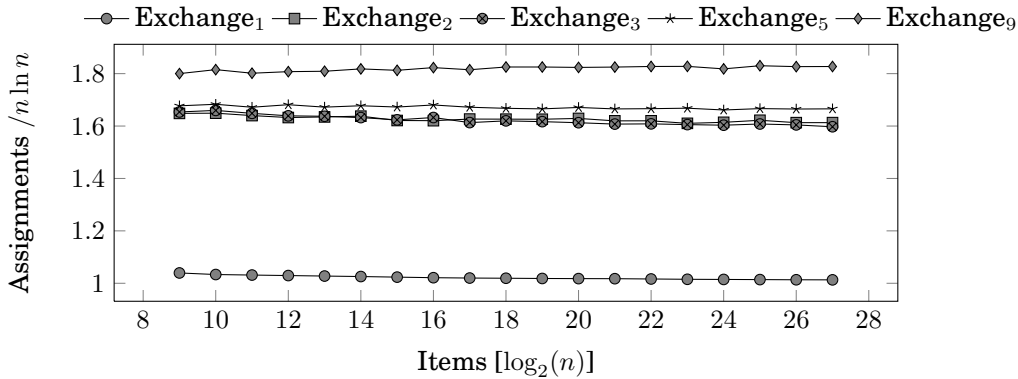
\begin{figure}[t!]
    \centering
    \begin{tikzpicture}
        \begin{axis}[
                xlabel={Items [$\log_2(n)$]},
                ylabel={Assignments $/ n \ln n$ },
                height=5cm,
                width=12cm,
                legend style = { at = {(-0.05,1.1)}, anchor=west, draw=none},
                cycle list name = black white,
                legend columns = 5
            ]
            \addplot coordinates { (9.0,1.03939) (10.0,1.03337) (11.0,1.03138) (12.0,1.02945) (13.0,1.02755) (14.0,1.02571) (15.0,1.02344) (16.0,1.02129) (17.0,1.01996) (18.0,1.01921) (19.0,1.01829) (20.0,1.01779) (21.0,1.0175) (22.0,1.01637) (23.0,1.0152) (24.0,1.01462) (25.0,1.01403) (26.0,1.01336) (27.0,1.01303) };
            \addlegendentry{Exchange$_1$};
            \addplot coordinates { (9.0,1.64833) (10.0,1.64915) (11.0,1.63993) (12.0,1.63249) (13.0,1.63474) (14.0,1.63774) (15.0,1.62115) (16.0,1.6205) (17.0,1.62664) (18.0,1.62632) (19.0,1.62594) (20.0,1.62909) (21.0,1.61977) (22.0,1.62) (23.0,1.61032) (24.0,1.61422) (25.0,1.62204) (26.0,1.61292) (27.0,1.613) };
            \addlegendentry{Exchange$_2$};
            \addplot coordinates { (9.0,1.65377) (10.0,1.65964) (11.0,1.64846) (12.0,1.639) (13.0,1.63728) (14.0,1.63222) (15.0,1.62363) (16.0,1.63256) (17.0,1.61312) (18.0,1.61978) (19.0,1.61676) (20.0,1.61268) (21.0,1.60746) (22.0,1.6083) (23.0,1.6055) (24.0,1.60325) (25.0,1.60779) (26.0,1.60435) (27.0,1.59726) };
            \addlegendentry{Exchange$_3$};
            \addplot coordinates { (9.0,1.67718) (10.0,1.68306) (11.0,1.67178) (12.0,1.68143) (13.0,1.67217) (14.0,1.67747) (15.0,1.67254) (16.0,1.68041) (17.0,1.6718) (18.0,1.66764) (19.0,1.66557) (20.0,1.67092) (21.0,1.66535) (22.0,1.66631) (23.0,1.66812) (24.0,1.6613) (25.0,1.66655) (26.0,1.66493) (27.0,1.66567) };
            \addlegendentry{Exchange$_5$};
            \addplot coordinates { (9.0,1.80008) (10.0,1.81528) (11.0,1.80161) (12.0,1.80755) (13.0,1.80908) (14.0,1.81808) (15.0,1.81264) (16.0,1.82299) (17.0,1.81489) (18.0,1.82533) (19.0,1.82531) (20.0,1.82375) (21.0,1.82505) (22.0,1.82734) (23.0,1.82775) (24.0,1.81817) (25.0,1.83017) (26.0,1.82675) (27.0,1.82719) };
            \addlegendentry{Exchange$_9$};
        \end{axis}
    \end{tikzpicture}
    \caption{The average number of assignments for sorting a random input
        consisting of $n$ elements using Algorithm~\ref{algo:exchange:k} (``Exchange$_k$'')
        for certain values of $k$.  Each data point is the average over $600$
    trials.}
    \label{fig:assignments}
\end{figure}

\paragraph{Conclusion} In this section we analyzed an algorithm (Algorithm~\ref{algo:exchange:k}) for the rearrangement problem 
    with respect to three different cost measures. We found that the cost 
    measure ``scanned elements'' is very useful for estimating the number of cache misses. 
    With respect to the number of scanned elements,
    Algorithm~\ref{algo:exchange:k} is particularly good with three to five
    pivots. For the cost measures ``write accesses'' and ``assignments'' we
    found that the cost increases with an increasing number of pivots. 
    

\section{Pivot Sampling in Multi-Pivot Quicksort}\label{sec:pivot:sampling}
In this section we study possible benefits of sampling pivots with respect to comparisons and to scanned elements. 
By ``pivot sampling'' we mean that we take a 
sample of elements from the input, sort these elements, and then pick certain elements of 
this sorted sequence as pivots. One particularly popular strategy for classical
quicksort, known as \emph{median-of-three}, is to choose as pivot the median of
a sample of three elements \cite{vanEmden}. 
From a theoretical point of view it is well known that
choosing the median in a sample of $\Theta(\sqrt{n})$ elements in classical quicksort is optimal
with respect to minimizing the average comparison count \cite{MartinezR01}.
Using this sample size, quicksort achieves the (asymptotically) best possible average
comparison count of $(1/\ln 2) n \ln n = 1.4426..n \ln n + O(n)$ comparisons on average.
For the YBB algorithm, Nebel, Wild, and Mart\'inez \citeyear{NebelWM15}
studied the influence of pivot sampling in detail. In particular, they considered the setting 
where one can choose pivots of a given rank for free. (Of course this only gives a lower 
bound on the sorting cost.) In this model they proved that no matter how well the pivots are chosen,
the YBB algorithm makes at least $1.49.. n \ln n + O(n)$
comparisons on average. Aumüller and Dietzfelbinger \citeyear{AumullerD15} demonstrated
that this is not an inherent limitation of the dual-pivot quicksort approach.

Here we study two different sampling scenarios for multi-pivot quicksort.
First we develop formulae to calculate
the average number of comparisons and the average number of scanned elements 
for Algorithm~\ref{algo:exchange:k}
when pivots are chosen from a small sample. Example calculations demonstrate
that the cost in both measures can be decreased by choosing
pivots from a small (fixed-sized) sample. Interestingly, with respect to scanned elements
 the best pivot choices do not
balance subproblem sizes but tend to make the middle groups,
i.e., groups $\text{A}_{p}$ with $p$ close to $\lceil \frac{k + 1}{2}\rceil$,
larger. Then we
consider the different setting in which we can choose pivots of a given rank for
free. For this setting we wish to find out which pivot choices minimize the respective
cost.
Our first result shows that if we choose an arbitrary comparison tree and use it in 
every classification, it is possible to choose pivots in such a way that on average we need at most
$1.4426.. n\ln n + O(n)$ comparisons to sort the input, which is optimal. In a second result we identify 
a particular pivot choice that minimizes the average number of scanned elements. In contrast to the results of
the previous section we show that with these pivot choices, 
the average number of scanned elements decreases with a growing number of
pivots. (Recall that when choosing pivots directly from the input,
the five-pivot quicksort algorithm based on Algorithm~\ref{algo:exchange:k}
has minimal cost.) From these calculations we also learn which
comparison tree (among the exponentially many available) has lowest cost when
considering as cost measure the sum of the number of comparisons and the number
of scanned elements. In contrast to intuition, the balanced comparison tree,
in which all leaves are as even in depth as possible,
has non-optimal cost. The best choice under this cost measure is to use the
comparison tree which uses as root pivot $\text{p}_{m} $ with $m = {\lceil \frac{k + 1}{2}\rceil}$. In its left subtree, the node labeled with pivot $\text{p}_i$ is the left child of the node 
labeled with pivot $\text{p}_{i + 1}$ 
for $1 \leq i \leq m - 1$. (So, the inner nodes in its left subtree are a path $(\text{p}_{m - 1}, \ldots,
\text{p}_1$).) Analogously, in its right subtree, the node labeled with pivot $\text{p}_{i + 1}$ is the right child
of the node labeled with pivot $\text{p}_i$ for $m \leq i \leq k - 1$.  For given $k \geq 1$,
we call this tree the \emph{extremal tree for $k$ pivots}. In
Figure~\ref{fig:extremal:comp:tree} we see an example for the extremal tree for seven pivots.

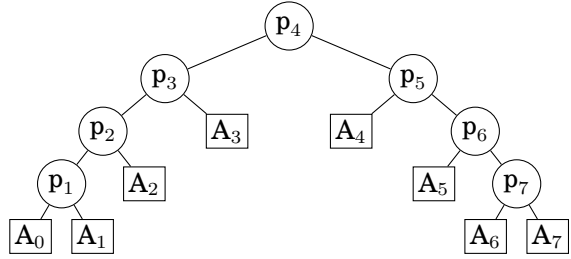
\begin{figure}[t]
    \centering
    \begin{tikzpicture}[level/.style = {sibling distance = 3.3cm/#1, level distance=0.7cm}]
        \tikzset {
            treenode/.style = {align = center, inner sep = 2pt, text centered},
            n/.style = {treenode, circle, white, draw = black, text = black},
            leaf/.style = {treenode, rectangle, draw = black}
        }

        \node [n] {$\text{p}_4$}
            child { node [n] {$\text{p}_3$}
                child { node [n] {$\text{p}_2$}
                    child {node [n] {$\text{p}_1$}
                        child { node [leaf] {$\text{A}_0$}}
                        child { node [leaf] {$\text{A}_1$}}
                    }
                    child { node [leaf] {$\text{A}_2$}}
                }
                child { node [leaf] {$\text{A}_3$}}
            }
            child { node [n] {$\text{p}_5$}
                child { node [leaf] {$\text{A}_4$}}
                child { node [n] {$\text{p}_6$}
                    child { node [leaf] {$\text{A}_5$}}
                    child { node [n] {$\text{p}_7$}
                        child { node [leaf] {$\text{A}_6$}}
                        child { node [leaf] {$\text{A}_7$}}
                    }
                }
            };

    \end{tikzpicture}
    \caption{The extremal tree for seven pivots.}
    \label{fig:extremal:comp:tree}
\end{figure}

\paragraph{General Structure of a Multi-Pivot Quicksort Algorithm Using Sampling}
We generalize a multi-pivot quicksort algorithm in the following way. (This description
is analogous to those in \cite{hennequin} for multi-pivot quicksort and \cite{NebelWM15} for dual-pivot quicksort.) For
a given number $k \geq 1$ of pivots, we fix a vector $\mathbf{t} = (t_0, \ldots,
t_k) \in \mathbb{N}^{k+1}$. Let $\kappa := \kappa(\mathbf{t}) = k + \sum_{0 \leq i \leq k} t_i$ be
the number of samples.\footnote{
The notation differs from \cite{hennequin} and \cite{NebelWM15} in the following way: 
This paper focuses on a ``pivot number''-centric approach, where the main parameter is 
$k$, the number of pivots. The other papers focus on the parameter $s$, 
the number of element groups. In particular, it holds $s = k + 1$. The sample size $\kappa$ is denoted $k$ in these papers.}) Assume that an input of $n$ elements residing in an array
$A[1..n]$ is to be sorted. If $n \leq \kappa$, sort $A$ directly. Otherwise, sort 
the first $\kappa$ elements and then set $p_i = A[i + \sum_{j < i} t_j]$, for $1
\leq i \leq k$. Next, partition the input $A[\kappa + 1..n]$ with respect to the
pivots $p_1,\ldots,p_k$. Subsequently, by a constant number of
rotations, move the elements residing in $A[1..\kappa]$
to correct final locations.
Finally, sort the $k + 1$ subproblems recursively.%

The sampling technique described above does not preserve randomness in subproblems,
because some elements have already been sorted during the pivot sampling step.
For the analysis, we ignore that the unused samples have been seen and get
only an estimate on the sorting cost. A detailed analysis of this situation
for dual-pivot quicksort is given in \cite{NebelWM15}; the same methods that were 
proposed in their paper for a tighter analysis can be used in the multi-pivot quicksort case, too.

\paragraph{The Generalized Multi-Pivot Quicksort Recurrence}
For a given sequence $\mathbf{t} = (t_0, \ldots, t_k) \in \mathbb{N}^{k + 1}$ we define
$H(\mathbf{t})$ by
\begin{align}
    H(\mathbf{t}) = \sum_{i = 0}^{k} \frac{t_i + 1}{\kappa + 1} \left(H_{\kappa + 1} -
    H_{t_i + 1}\right).
    \label{eq:entropy}
\end{align}
Let $\text{P}_n$ denote the random variable which counts the cost of a single partitioning step,
and let $\text{C}_n$ denote the cost over the whole sorting procedure. In general, we get the recurrence
	\begin{align}
            \E(\text{C}_n) = \E(\text{P}_n) + \!\!\!\!\sum_{a_0 + \cdots + a_k = n - k} \left(
            \E(\text{C}_{a_0}) + \cdots + \E(\text{C}_{a_k})\right) \cdot \Pr(\langle a_0, \ldots, a_k\rangle),
            \label{eq:sampling:recurrence}
        \end{align}
    where $\langle a_0,\ldots,a_k\rangle$ is the event that the group sizes are exactly $a_0, \ldots, a_k$.
        The probability of this event for a given vector $\mathbf{t}$ is
        \begin{align*}
            \frac{\binom{a_0}{t_0} \cdots \binom{a_k}{t_k}}{\binom{n}{\kappa}}.
        \end{align*}
        For the following discussion, we re-use the result of Hennequin~\citeyear[Proposition III.9]{hennequin} which says that
        for fixed $k$ and $\mathbf{t}$ and average partitioning cost $\E(P_n) = a \cdot n + O(1)$ recurrence \eqref{eq:sampling:recurrence}
has the solution
\begin{align}
    \E(C_n) = \frac{a}{H(\mathbf{t})} n \ln n + O(n).
    \label{eq:quicksort:recurrence:sampling}
\end{align}

\paragraph{The Average Comparison Count Using a Fixed Comparison Tree} Fix a
vector $\mathbf{t} \in \mathbb{N}^{k + 1}$.
First, observe that for each $i \in \{0, \ldots, k\}$ the expected number of
elements belonging to group $\text{A}_i$ is $\frac{t_i + 1}{\kappa + 1} (n -
\kappa)$. If the $n - \kappa$ remaining input elements are classified using a
fixed comparison tree $\lambda$, the average comparison count for partitioning (cf. \eqref{eq:30014}) is
\begin{align}
    \left(n - \kappa\right) \cdot \sum_{i = 0}^{k} \text{depth}_{\lambda}(\text{A}_i) \cdot \frac{t_i + 1}{\kappa + 1}.
    \label{eq:comp:count:sampling}
\end{align}

\paragraph{The Average Number of Scanned Elements of Algorithm~\ref{algo:exchange:k}}
Fix a vector $\mathbf{t} \in \mathbb{N}^{k + 1}$ and let $m = \lceil \frac{k +
1}{2}\rceil$. Arguments analogous to the ones presented in the previous section, see \eqref{eq:memory:accesses:1},
show that the average number of scanned elements of
Algorithm~\ref{algo:exchange:k} is
\begin{align}
		\begin{cases}
			n \cdot \sum_{i = 1}^{m} i \cdot \left(\frac{t_{m - i} +
            t_{k - m + i} + 2}{\kappa + 1}\right) + O(1), &\text{for $k$ odd},\\
			n \cdot (m + 1) \cdot \frac{t_0 + 1}{\kappa + 1} + n \cdot \sum_{i = 1}^{m} i \cdot
            \left(\frac{t_{m + 1 -i} + t_{k - m + i} + 2}{\kappa + 1}\right) + O(1),& \text{for $k$ even}.
\end{cases}
\label{eq:access:count:sampling}
\end{align}
Next, we will use these formulae to give some example calculations for small sample sizes.

\paragraph{Optimal Pivot Choices for Small Sample Sizes}
Table~\ref{tab:sampling:cost} contains the lowest possible cost for a given
number of pivots and a given number of sample elements. We consider three
different cost measures: the average number of comparisons, the average number
of scanned elements, and the sum of these two costs.
Additionally, Table~\ref{tab:sampling:cost} 
contains the $\mathbf{t}$-vector and the comparison tree that achieves this
value. (Of course, each comparison tree yields the same number of scanned elements. Consequently, 
no comparison tree is given for the best $\mathbf{t}$-vector w.r.t. scanned elements.)

Looking at Table~\ref{tab:sampling:cost}, we make the following observations. Increasing the sample size for a fixed
number of pivots decreases the average cost significantly, at least
asymptotically. Interestingly, to minimize the
average number of comparisons, the best comparison tree is not always 
one that minimizes the depth of the tree, e.g., see the best comparison
tree for 5 pivots and 6 additional samples. However, for most situations the tree
with minimal cost has minimal depth. 
To minimize the number of scanned elements, 
the groups in the middle should be made larger. The extremal 
tree provides the best possible \emph{total cost},
summing up comparisons and scanned elements. Compared to the sampling
choices that minimize the number of scanned elements, the sampled elements are
slightly less concentrated around the middle element groups. This is
first evidence that the extremal tree might be the best possible comparison tree for
a given number of pivots with respect to the total cost.

With regard to the question of the best choice of $k$, Table~\ref{tab:sampling:cost}
shows that it is not possible to give a definite answer.
All of the considered
pivot numbers and additional sampling elements make it possible to decrease the
total average sorting cost to around $2.6 n \ln n$ using only a small sample.

Next, we will study the behavior of these cost
measures when choosing the pivots is for free.

\begin{table}[t!]
    \tbl{Best sampling and comparison tree choices for a given number $k$ of
        pivots and a given sample size (in addition to pivots).
        A comparison tree is presented as the output of the preorder traversal of
        its inner nodes, e.g., the extremal tree in Fig.~\ref{fig:extremal:comp:tree} is described $[4,3,2,1,5,6,7]$. 
        Since the number of scanned elements is independent of the used comparison 
        tree, no comparison tree is given for this cost measure.
\label{tab:sampling:cost}}
    {
    \begin{tabular}{c  c l  l}
        \toprule
        $k$ & \textbf{Add. Samples} & \textbf{Cost measure} & \textbf{Best candidate} (cost, $\mathbf{t}$, tree) \\ \hline
         &  & comparisons & $1.846 n \ln n$, $(0, 0, 0, 0)$, $[2, 1, 3]$ \\
        3 & $0$ & scanned elements & $1.385 n \ln n$, $(0, 0, 0, 0)$, --- \\
          &  & cmp + scanned elements & $3.231 n \ln n$, $(0, 0, 0, 0)$, $[2, 1, 3]$ \\[0.1cm]

          &  & comparisons & $1.642 n \ln n$, $(1, 1, 1, 1)$, $[2, 1, 3]$ \\
         & $4$ & scanned elements & $1.144 n \ln n$, $(0, 2, 2, 0)$, --- \\
          &  & cmp + scanned elements & $2.874 n \ln n$, $(1, 1, 1, 1)$, $[2, 1, 3]$ \\[0.1cm]

         &  & comparisons & $1.575 n \ln n$, $(2, 2, 2, 2)$, $[2, 1, 3]$ \\
         & $8$ & scanned elements & $1.098 n \ln n$, $(1, 3, 3, 1)$, --- \\
          &  & cmp + scanned elements & $2.745 n \ln n$, $(1, 3, 3, 1)$, $[2, 1, 3]$ \\[0.1cm]

         &  & comparisons & $1.522 n \ln n$, $(4, 4, 4, 4)$, $[2, 1, 3]$ \\
         & $16$ & scanned elements & $1.055 n \ln n$, $(2, 6, 6, 2)$, --- \\
          &  & cmp + scanned elements & $2.627 n \ln n$, $(3, 5, 5, 3)$, $[2, 1, 3]$ \\[0.1cm]
        \hline \\[-.4cm]
        &  & comparisons & $1.839 n \ln n$, $(0, 0, 0, 0, 0, 0)$, $[3, 2, 1, 4, 5]$ \\
        5 & $0$ & scanned elements & $1.379 n \ln n$, $(0, 0, 0, 0, 0, 0)$, --- \\
          &  & cmp + scanned elements & $3.218 n \ln n$, $(0, 0, 0, 0, 0, 0)$, $[3, 2, 1, 4, 5]$ \\[0.1cm]

        &  & comparisons & $1.635 n \ln n$, $(0, 0, 0, 2, 2, 2)$, $[4, 3, 1, 2, 5]$ \\
         & $6$ & scanned elements & $1.097 n \ln n$, $(0, 1, 2, 2, 1, 0)$, --- \\
          &  & cmp + scanned elements & $2.741 n \ln n$, $(0, 1, 2, 2, 1, 0)$, $[3, 2, 1, 4, 5]$ \\[0.1cm]
         &  & comparisons & $ 1.567  n \ln n$, $(1, 1, 4, 4, 1, 1)$, $[3, 2, 1, 4, 5]$ \\
         & $12$ & scanned elements & $1.019  n \ln n$, $(0, 1, 5, 5, 1, 0)$, --- \\
          &  & cmp + scanned elements & $ 2.635 n \ln n$, $(1, 1, 4, 4, 1, 1)$, $[3, 2, 1, 4, 5]$ \\[0.1cm]
        \hline \\[-.4cm]
        &  & comparisons & $1.746 n \ln n$, $(0,0,0,0,0,0,0,0)$, $[4, 2, 1, 3, 6, 5, 7]$ \\
        7 & $0$ & scanned elements & $1.455 n \ln n$, $ (0,0,0,0,0,0,0,0)$, --- \\
          &  & cmp + scanned elements & $3.201 n \ln n$, $(0, 0, 0, 0, 0, 0, 0, 0)$, $[4, 2, 1, 3, 6, 5, 7]$ \\[0.1cm]
        &  & comparisons & $1.595 n \ln n$, $(1,1,1,1,1,1,1,1)$, $[4, 2, 1, 3, 6, 5, 7]$ \\
         & $8$ & scanned elements & $1.094 n \ln n$, $(0, 0, 1, 3, 3, 1, 0, 0)$, --- \\
          &  & cmp + scanned elements & $2.698 n \ln n$, $(0, 0, 1, 3, 3, 1, 0, 0)$, $[4, 3, 2, 1, 5, 6, 7]$ \\[0.1cm]
         &  & comparisons & $1.544 n \ln n$, $(2,2,2,2,2,2,2,2)$, $[4, 2, 1, 3, 6, 5, 7]$ \\
         & $16$ & scanned elements & $1.017 n \ln n$, $(0, 0, 2, 6, 6, 2, 0, 0)$, --- \\
          &  & cmp + scanned elements & $2.594 n \ln n$, $(0, 0, 2, 6, 6, 2, 0, 0)$, $[4, 3, 2, 1, 5, 6, 7]$ \\[0.1cm]
        \hline \\[-.4cm]
        &  & comparisons & $1.763 n \ln n$, $(0, 0, 0, 0, 0, 0, 0, 0, 0, 0)$, $[5, 3, 2, 1, 4, 7, 6, 8, 9]$ \\
        9 & $0$ & scanned elements & $1.555 n \ln n$, $ (0, 0, 0, 0, 0, 0, 0, 0, 0, 0)$, --- \\
          &  & cmp + scanned elements & $3.318  n \ln n$, $(0, 0, 0, 0, 0, 0, 0, 0, 0, 0)$, $[5, 3, 2, 1, 4, 7, 6, 8, 9]$ \\[0.1cm]
        &  & comparisons & $1.602 n \ln n$, $(0, 0, 1, 2, 2, 2, 2, 1, 0, 0)$, $ [5, 3, 2, 1, 4, 7, 6, 8, 9]$ \\
         & $10$ & scanned elements & $1.131 n \ln n$, $ (0, 0, 0, 1, 4, 4, 1, 0, 0, 0)$, --- \\
          &  & cmp + scanned elements & $2.748  n \ln n$, $(0, 0, 0, 1, 4, 4, 1, 0, 0, 0)$, $[5, 4, 3, 2, 1, 6, 7, 8, 9]$ \\[0.1cm]
          &  & comparisons & $1.543 n \ln n$, $(1, 1, 2, 3, 3, 3, 3, 2, 1, 1)$, $ [5, 3, 2, 1, 4, 7, 6, 8, 9]$ \\
         & $20$ & scanned elements & $1.040  n \ln n$, $(0, 0, 0, 2, 8, 8, 2, 0, 0, 0)$, --- \\
         &  & cmp + scanned elements & $2.601 n \ln n$, $(0, 0, 1, 2, 7, 7, 2, 1, 0, 0)$, $[5, 4, 3, 2, 1, 6, 7, 8, 9]$ \\[0.1cm]
        \bottomrule
    \end{tabular}
}
    \end{table}
\paragraph{Optimal Pivot Choices} We now consider the following setting. We assume that for a
random input of $n$ elements\footnote{We disregard the $k$ pivots in the
following discussion.} we can choose (for free) $k$ pivots w.r.t. a vector
$\mathbf{\tau} = (\tau_0, \ldots, \tau_k)$ such that the input contains exactly
$\tau_i n$ elements from group $\text{A}_i$, for $i \in \{0,\ldots,k\}$. By definition,
we have $\sum_{0 \leq i \leq k} \tau_i = 1$.
This setting was studied in \cite{MartinezR01,NebelWM15} as well.

We make the following preliminary observations. In our setting, the average
number of comparisons \emph{per element} (see \eqref{eq:comp:count:sampling})
becomes
\begin{align}
    c_\tau := \sum_{i = 0}^{k} \text{depth}_t(\text{A}_i) \cdot \tau_i,
    \label{eq:sampling:comparisons}
\end{align}
and the average number of scanned elements per element (see \eqref{eq:access:count:sampling}) is
\begin{align}
    a_\tau :=
		\begin{cases}
			\sum_{i = 1}^{m - 1} i \cdot \left(\tau_{m - i -1 } +
            \tau_{k - m + i + 1}\right), &\text{for $k$ odd},\\
			m \cdot \tau_0 + \sum_{i = 1}^{m - 1} i \cdot
            \left(\tau_{m -i} + \tau_{k - m + i + 1}\right),& \text{for $k$ even}.
    \end{cases}
    \label{eq:sampling:mem:accesses:tau}
\end{align}
Furthermore, using that the $\kappa$th harmonic number $H_\kappa$ is approximately $\ln \kappa$, we get that $H(\tau)$ from
\eqref{eq:entropy}  converges to the entropy of $\tau$, which is
\begin{align*}
    \mathcal{H}(\mathbf{\tau}) := - \sum_{i = 0}^{k} \tau_i \ln \tau_i.
\end{align*}
In the following we want to obtain optimal choices for the vector $\mathbf\tau$ to minimize the three values
\begin{align}
    \frac{c_\tau}{\mathcal{H}(\mathbf{\tau})}, \quad \quad\frac{a_\tau}{\mathcal{H}(\mathbf{\tau})}, \quad\quad \frac{c_\tau + a_\tau}{\mathcal{H}(\mathbf{\tau})},
    \label{eq:sample:asymp}
\end{align}
i.e., that minimize the factor in the $n \ln n$ term of the sorting cost in the respective cost measure.

We start by giving a general solution to the problem of minimizing functions $f$ 
over the simplex $S_{k+1}=\{(\tau_0,\dots,\tau_k)\in\RR^{k+1} \mid \tau_0,\dots,\tau_k\ge0, \sum_i \tau_i = 1\}$, 
where $f$ is the quotient of some linear function in $\tau_0, \ldots, \tau_k$ 
and the entropy of $\tau_0, \ldots, \tau_k$, as in \eqref{eq:sample:asymp}.

\begin{lemma}
    Let $\alpha_0, \ldots, \alpha_k > 0$ be arbitrary constants.
    For $\mathbf\tau = (\tau_0, \ldots, \tau_k) \in S_{k+1}$ define
    \begin{align*}
        f(\mathbf\tau) = \frac{\alpha_0 \tau_0 + \cdots + \alpha_k \tau_k}{\mathcal{H}(\mathbf\tau)}.
    \end{align*}
    Let $x$ be the unique solution in $(0,1)$ of the equation
    \begin{align*}
        1 = x^{\alpha_0} + x^{\alpha_1} + \cdots + x^{\alpha_k}.
    \end{align*}
    Then $\mathbf\tau = (x^{\alpha_0}, x^{\alpha_1}, \ldots, x^{\alpha_k})$ minimizes
    $f(\mathbf\tau)$ over $S_{k+1}$, i\,.e., $f(\mathbf\tau) \leq f(\mathbf{\sigma})$ for all 
    $\mathbf{\sigma} \in S_{k+1}$. The minimum value is $-1/\ln x$.
    \label{lem:sol:sampling:minima}
\end{lemma}

\begin{proof}
    Gibb's inequality says that for arbitrary nonnegative $q_0, \ldots, q_k$ with $\sum_{i} q_i \leq 1$ and arbitrary nonnegative
    $p_0, \ldots, p_k$ with $\sum_{i} p_i  = 1$ the following holds:
    \begin{align*}
        \mathcal{H}(p_0, \ldots, p_k) \leq \sum_{i = 0}^k p_i
        \ln\left(\frac{1}{q_i}\right).
    \end{align*}
    Now consider $x \in (0, 1)$ such that $\sum_i x^{\alpha_i} = 1$ and set $\tau_i
    = x^{\alpha_i}$.
    By Gibb's inequality, for arbitrary $(\sigma_0, \ldots, \sigma_k)\in S_{k+1}$ we have
    \begin{align*}
        \mathcal{H}(\sigma_0, \ldots, \sigma_k) \leq \sum_{i = 0}^k \sigma_i \ln\left(\frac{1}{x^{\alpha_i}}\right)= \sum_{i = 0}^{k} \sigma_i \alpha_i \ln\left( \frac{1}{x}\right),
    \end{align*}
    hence
    \begin{align*}
        f(\sigma_0, \ldots, \sigma_k) = \frac{\sum_{i = 0}^k \sigma_i \alpha_i}{\mathcal{H}(\sigma_0, \ldots, \sigma_k)} \geq \frac{1}{\ln \left(\frac{1}{x}\right)}.
    \end{align*}
    Finally, observe that
    \begin{align*}
        f(\tau_0, \ldots, \tau_k) = \frac{\sum_{i = 0}^k \tau_i \alpha_i}{\mathcal{H}(\tau_0, \ldots, \tau_k)} = \frac{\sum_{i = 0}^k \alpha_i x^{\alpha_i}}{- \sum_{i = 0}^k x^{\alpha_i} \ln\left(x^{\alpha_i}\right)} = \frac{1}{\ln\left(\frac{1}{x}\right)}.\tag*\qed
    \end{align*}
\end{proof}

We first consider optimal choices for the sampling vector $\mathbf\tau$ in order to minimize
the average number of comparisons. The following theorem says
that each comparison tree makes it possible to achieve the minimum possible
sorting cost for comparison-based sorting algorithms in the considered setting.

\begin{theorem}
    Let $k \geq 1$ be fixed. Let $\lambda \in \Lambda^{k}$ be an arbitrary comparison tree. Then there exists
    $\mathbf{\tau}$ such that the average comparison count using comparison tree $\lambda$ in
    each classification
    is $(1/\ln 2) n \ln n + O(n) = 1.4426..n \ln n + O(n)$.
    \label{thm:opt:class}
\end{theorem}

\begin{proof}
    For each $i \in \{0, \ldots, k\}$, set $\tau_i =
    2^{-\text{depth}_\lambda({\text{A}_i})}$. First, observe that $\sum \tau_i = 1$.
    (This is true since every inner node in $\lambda$ has exactly two children.)
    Starting from \eqref{eq:comp:count:sampling}, we may calculate:
    \begin{align*}
        \sum_{i = 0}^{k} \text{depth}_{\lambda}(\text{A}_i) \cdot \tau_i = \sum_{i = 0}^{k} -\log(\tau_i) \cdot \tau_i
        = -\frac{1}{\ln 2} \sum_{i = 0}^{k} \tau_i \cdot \ln(\tau_i).
    \end{align*}
    This shows the theorem.\qed
\end{proof}
We now consider the average number of scanned elements of
Algorithm~\ref{algo:exchange:k}. The following theorem says that pivots should
be chosen in such a way that (in the limit for $k \to \infty$) $2/3$ of the
input are only scanned once (by the pointers $\texttt{i}$ and $\texttt{j}$), $2/9$ should be
scanned twice, and so on.
\begin{theorem}
    Let $k \geq 1$ be fixed. Let $m = \lceil \frac{k + 1}{2}\rceil$. Let
    $\mathbf{\tau}$ be chosen according to the following two cases:
    \begin{enumerate}
        \item If $k$ is odd, let $x$ be the unique value in $(0, 1)$ such that
            \begin{align*}
                1 = 2(x + x^2 + \cdots + x^{m}).
            \end{align*}
            Let $\mathbf\tau = (x^{m}, x^{m - 1}, \ldots, x, x, \ldots, x^{m})$.
        \item If $k$ is even, let $x$ be the unique value in $(0, 1)$ such that
            \begin{align*}
                1 = 2(x + x^2 + \cdots + x^{m - 1}) + x^{m}.
            \end{align*}
            Let $\mathbf\tau = (x^{m}, x^{m - 1}, \ldots, x, x, \ldots, x^{m - 1})$.
    \end{enumerate}
    Then the average number of scanned elements using
    Algorithm~\ref{algo:exchange:k} with $\mathbf\tau$ is minimal over all choices of
    vectors $\mathbf\tau'$. For $k \rightarrow \infty$, this minimum is $(1/\ln 3) n \ln n \approx 0.91 n \ln n$ scanned elements.
    \label{thm:mem:acc}
\end{theorem}
\begin{proof}
    Setting the values $\alpha_0, \ldots, \alpha_k$ in Lemma~\ref{lem:sol:sampling:minima}
    according to Equation~\eqref{eq:sampling:mem:accesses:tau} shows that the choices
    for $\mathbf\tau$ are optimal with respect to minimizing the average number of scanned
    elements. One easily checks that in the limit for $k \to \infty$ the value $x$
    in the statement is $1/3$.\qed
\end{proof}
For example, for the YBB algorithm
\citeN{NebelWM15} noticed that $\mathbf\tau=(q^2, q, q)$ with $q = \sqrt{2} - 1$
is the optimal pivot choice to minimize element scans. In this case,
around $1.13 n \ln n$ elements are scanned on average. The minimal average number of scanned elements using Algorithm~\ref{algo:exchange:k} for $k \in \{3, 5, 7, 9\}$ are
around $0.995n \ln n$, $0.933 n \ln n$, $0.917 n \ln n$, and $0.912 n \ln n$, respectively.
Hence, already for small values of $k$ the average number of scanned elements is
close to $0.91 n \ln n$. However, from Table~\ref{tab:sampling:cost}
we see that for sample sizes suitable in practice, both the average comparison
count and average pointer visit count are around $0.1 n \ln n$ higher
than these asymptotic values.

To summarize, we learned that (i) every fixed comparison
tree yields an optimal classification strategy and (ii) one specific pivot
choice has the best possible average number of scanned elements in
Algorithm~\ref{algo:exchange:k}.  Next, we consider as cost measure the sum of
the average number of comparisons and the average number of scanned elements.

\begin{theorem}
    Let $k \geq 1$ be fixed. Let $m = \lceil \frac{k + 1}{2}\rceil$. Let
    $\mathbf{\tau}$ be chosen according to the following two cases:
    \begin{enumerate}
        \item If $k$ is odd, let $x$ be the unique value in $(0, 1)$ such that
            \begin{align*}
                1 = 2(x^3 + x^5 + \cdots + x^{2m - 3} + x^{2m - 1} + x^{2m}).
            \end{align*}
            Let $\mathbf\tau = (x^{2m}, x^{2m - 1}, x^{2m - 3}, \ldots, x^3, x^3, \ldots, x^{2m - 3}, x^{2m - 1}, x^{2m})$.
        \item If $k$ is even, let $x$ be the unique value in $(0, 1)$ such that
            \begin{align*}
                1 = 2(x^3 + x^5 + \cdots + x^{2m - 3}) + x^{2m - 2} + x^{2m - 1} + x^{2m}.
            \end{align*}
            Let $\mathbf\tau = (x^{2m}, x^{2m - 1}, x^{2m - 3}, \ldots, x^3, x^3, \ldots, x^{2m - 3}, x^{2m - 2})$.
    \end{enumerate}
    Then the average cost using
    Algorithm~\ref{algo:exchange:k} and classifiying all elements with the
    extremal comparison tree for $k$ pivots using $\mathbf\tau$ is minimal among all choices of
    vectors $\mathbf\tau'$. For $k \rightarrow \infty$
    this minimum cost is about $2.38 n \ln n$.
    \label{thm:memcmp:sampling}
\end{theorem}

\begin{proof}
    First, observe that for the extremal comparison tree for $k$ pivots,
    \eqref{eq:sampling:comparisons} becomes
    \begin{align}
        c_\tau =\begin{cases}
            m (\tau_0 + \tau_k) +
                \sum_{i = 1}^{m - 1} (i + 1) (\tau_{m - i} + \tau_{m + i - 1}), & \text{for $k$ odd},\\
            m (\tau_0 + \tau_1) + (m - 1)  \tau_k +
            \sum_{i = 1}^{m - 2} (i + 1) (\tau_{m - i} + \tau_{m + i - 1}), & \text{for $k$ even}.
        \end{cases}
        \label{eq:proof:memcmp:1}
    \end{align}
    The optimality of the choice for $\mathbf\tau$ in the theorem now follows
    from adding \eqref{eq:sampling:mem:accesses:tau} to \eqref{eq:proof:memcmp:1}, and
    using Lemma~\ref{lem:sol:sampling:minima}. For $k \to \infty$, the optimal $x$ value
    is the unique solution in $(0, 1)$ of the equation $2x^3 + x^2 = 1$, which is about $0.6573$.
    Thus, the total cost is about $2.38 n \ln n$. \qed
\end{proof}
Interestingly, the minimal total cost of $2.38n \ln n$ is only about $0.03 n
\ln n$ higher than adding $(1/\ln 2) n \ln n$ (the minimal average comparison
count) and $0.91 n \ln n$ (the minimal average number of scanned elements).
Using the extremal tree is much better than, e.g., using the
balanced comparison tree for $k = 2^\kappa - 1$ pivots.
The minimum sorting cost using this tree is
$2.489 n \ln n$, which is achieved for three pivots\footnote{For three pivots, the extremal tree
and the balanced tree are identical.}. We conjecture that the
extremal tree has minimal total sorting cost, i.e., minimizes the sum of
scanned elements and comparisons, but must leave this as an open problem. 
We checked this conjecture via exhaustive search for $k \leq 9$ pivots.

Again, including scanned elements as cost measure yields unexpected results and
design recommendations for engineering a sorting algorithm. Looking
at comparisons, the balanced tree for $2^\kappa - 1$ pivots is the
obvious comparison tree to use in a multi-pivot quicksort algorithm. Only
when including scanned elements, the extremal tree shows its potential in leading
to fast 
multi-pivot quicksort algorithms.\footnote{It is interesting to note that
    papers dealing with multi-pivot quicksort such as \cite{hennequin,%
    Kushagra14,Iliopoulos14} do not consider the full design space for
    partitioning strategies, but rather always assume the ``most-balanced''
tree is best.}

\section{Experiments}\label{sec:experiments}

In this section, we address the question how the design recommendations of the
theoretical findings in this paper work out in practice. We stress that we tested
our algorithms strictly in the theoretical setting of random permutations as studied throughout
this paper. So, our algorithms are far away from ``library-ready'' implementations
that work nicely on different input distributions, especially in the presence of equal keys.

\paragraph{Objectives of the Experiments}
The key theoretical findings of the paper suggest the following questions about 
running times of multi-pivot quicksort algorithms.
\begin{enumerate}
  \item[Q1] \emph{Can comparison-optimal multi-pivot quicksort compete in running time?} {While the (asymptotically) comparison-optimal multi-pivot quicksort algorithms from Section~\ref{sec:optimal:strategies} minimize 
        the average comparison count, they add a time overhead to obtain the 
    optimal comparison tree in each classification step. 
		In this light it seems that they cannot compete with simpler algorithms that use a 
  fixed tree.}
  \item[Q2] \emph{How do running times of multi-pivot quicksort algorithms compare?} {We have noticed in 
      Section~\ref{sec:assignments} that using 3 to 5 pivots in 
            Algorithm~\ref{algo:exchange:k} yields the best memory behavior. We study how this is reflected in running times.}
          \item[Q3] \emph{Is the extremal comparison tree faster than the balanced comparison tree?} {In Section~\ref{sec:pivot:sampling}  we showed that multi-pivot quicksort algorithms using extremal classification trees have lower 
              cost in a certain cost model than balanced comparison trees. We investigate whether this corresponds 
							to a difference in running times or not.}
      \item[Q4] \emph{Is pivot sampling beneficial for running time?} {For a fixed sample vector $\tau$, choosing the pivots in a skewed way as 
        proposed in Theorem~\ref{thm:mem:acc} has lower cost than using equidistant pivots (Section~\ref{sec:pivot:sampling}), and greatly improves over the case where pivots are chosen directly from the input. We examine whether such sampling choices yield noticable differences with regard to running time.}
\end{enumerate}

\paragraph{Experimental Setup}
We sort random permutations
of the set $\{1,\ldots,n\}$. We implemented all algorithms in \Cpp{} and
used \emph{clang} in version 3.5 for compiling with the optimization flag \emph{-O3}.
Our experiments were carried out on an Intel i7-2600 at 3.4 GHz
with 16 GB Ram running Ubuntu 14.10 with kernel version 3.16.0. 

We restricted our experiments to
sorting random permutations of the integers $\{1,\ldots,n\}$.
We tested inputs of size $2^i$, $21 \leq i \leq 27$. For each
input size, we ran each algorithm on the same $600$ random permutations.
All figures in the following are the average over these $600$ trials. 
We say that a running 
time improvement is \emph{significant} if it was observed in at least 
95\% of the trials.

To deal with Questions 2 and 4, we wrote a script to generate 
multi-pivot quicksort algorithms with a fixed pivot choice 
based on Algorithm~\ref{algo:exchange:k}. 
(We remark that manually writing the algorithm code is error-prone and tedious.
For example, the source code for the $9$-pivot algorithm uses $376$ lines of
code and needs nine different rotate operations.)
In these algorithms, subarrays of size at most $500$ were sorted using
our implementation of the fast three-pivot quicksort algorithm of \cite{Kushagra14} 
as described in \cite{AumullerD15}. The source code of all algorithms and the code generator 
is available at \webpage. 

At the end of this section, we compare the results to  
the standard introsort sorting method from \Cpp{}'s standard library and 
the super scalar sample sort (``SSSS'') algorithm of \citeN{SandersW04}.

\subsection{Experimental Evaluation Regarding Questions 1--4} 
\emph{Question 1: Can comparison-optimal multi-pivot quicksort compete in running time?}
We implemented the comparison-optimal algorithm $\mathcal{O}_k$ and 
the (asymptotically) optimal sample strategy $\mathcal{SP}_k$ for two and three pivots such that 
the optimal comparison tree was chosen directly by comparing group sizes.  For example, in the three-pivot algorithm the optimal comparison tree 
out of the five possible comparison trees is chosen by selecting the tree that has the smallest cost term in 
\eqref{eq:cost:all:trees}. In addition, we implemented $k$-pivot 
variants that choose the optimal comparison tree algorithmically. We computed the tree using 
Knuth's dynamic programming approach \cite{Knuth}. (We did not implement the algorithm of 
Garsia and Wachs.)
The time measurements of our experiments with comparison-optimal multi-pivot quicksort algorithms are shown in Figure~\ref{fig:running:times:opt:algorithms}. 
We used the implementation of the three-pivot algorithm $\mathcal{K}$ from \cite{Kushagra14} as described in \cite{AumullerD15} as the baseline to which to compare the algorithms. This three-pivot algorithm
was the fastest algorithm, being around $5\%$ faster on average than its closest competitor 
$\mathcal{SP}_2$. 
We see that the sampling variant $\mathcal{SP}_k$ is faster than the comparison-optimal variant $\mathcal{O}_k$ (by about $7\%$ both for two and three pivots). The two-pivot sampling
algorithm is about $9.2\%$ faster than its three-pivot variant. 
The generic sampling variants that choose the tree algorithmically are far slower than these algorithms. 
The generic three-pivot sampling implementation $\mathcal{GSP}_3$ is about 2.2 times slower than its direct counterpart $\mathcal{SP}_3$. 
On average, the generic seven-pivot sampling implementation $\mathcal{GSP}_7$  is only about $3\%$ slower than  $\mathcal{GSP}_3$.
We omit the running time for the generic comparison-optimal algorithms. For three pivots, 
it was about 5 times as slow on average than its sampling variant $\mathcal{GSP}_3$.
We conclude that comparison-optimal multi-pivot quicksort cannot compete in running time
when sorting random permutations of integers, especially when the comparison tree is chosen algorithmically.

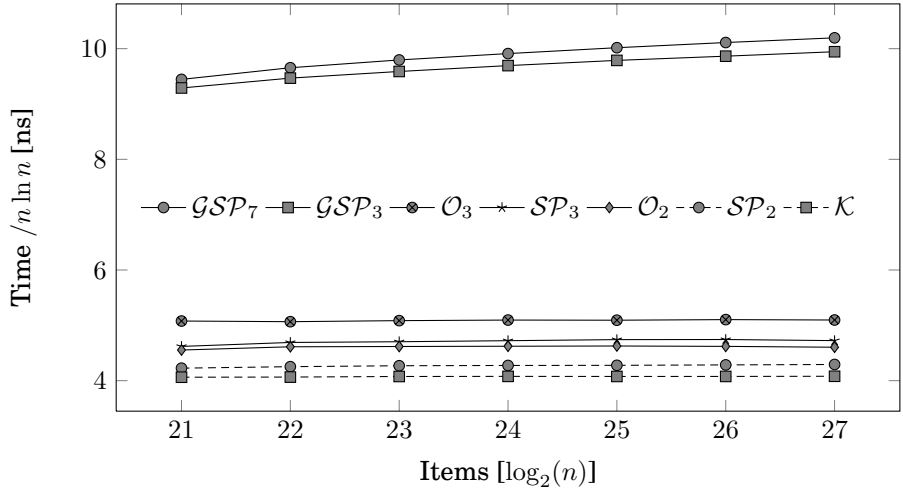
\begin{figure}
    \centering
\begin{tikzpicture}
  \begin{axis}[
    xlabel={Items [$\log_2(n)$]},
    ylabel={Time $/ n \ln n$ [ns]},
    height=7cm,
    width=12cm,
    legend style = { at = {(0.02,0.5)}, anchor=west, draw=none},
    cycle list name = black white,
    legend columns = 7
    ]
    \addplot coordinates { (21.0,9.44483) (22.0,9.65542) (23.0,9.79621) (24.0,9.91109) (25.0,10.0173) (26.0,10.1109) (27.0,10.1957) };
    \addlegendentry{$\mathcal{GSP}_7$};
    \addplot coordinates { (21.0,9.28936) (22.0,9.46806) (23.0,9.58788) (24.0,9.69387) (25.0,9.78817) (26.0,9.86428) (27.0,9.94452) };
    \addlegendentry{$\mathcal{GSP}_3$};
    \addplot coordinates { (21.0,5.07758) (22.0,5.06567) (23.0,5.0847) (24.0,5.09499) (25.0,5.09241) (26.0,5.10325) (27.0,5.09598) };
    \addlegendentry{$\mathcal{O}_3$};
    \addplot coordinates { (21.0,4.61896) (22.0,4.69044) (23.0,4.70335) (24.0,4.72236) (25.0,4.74156) (26.0,4.74358) (27.0,4.72355) };
    \addlegendentry{$\mathcal{SP}_3$};
    \addplot coordinates { (21.0,4.55345) (22.0,4.61226) (23.0,4.61736) (24.0,4.62203) (25.0,4.62633) (26.0,4.61955) (27.0,4.60392) };
    \addlegendentry{$\mathcal{O}_2$};
    \addplot coordinates { (21.0,4.22586) (22.0,4.25266) (23.0,4.26966) (24.0,4.27449) (25.0,4.27721) (26.0,4.28468) (27.0,4.2924) };
    \addlegendentry{$\mathcal{SP}_2$};
    \addplot coordinates { (21.0,4.06207) (22.0,4.06504) (23.0,4.07524) (24.0,4.07742) (25.0,4.07599) (26.0,4.07632) (27.0,4.07941) };
    \addlegendentry{$\mathcal{K}$};
\end{axis}
\end{tikzpicture}
\caption{Running time experiments for sorting inputs of size
    $2^i$ for $21 \leq i \leq 27$. The algorithms shown are three- ($\mathcal{GSP}_3$) and seven-pivot ($\mathcal{GSP}_7$) algorithms that compute the optimal comparison tree and use the sampling approach,
two- and three-pivot variants of the optimal algorithm $\mathcal{O}_k$ and its sampling variant $\mathcal{SP}_k$ that choose the optimal comparison tree directly, and 
the three-pivot algorithm $\mathcal{K}$ of {\protect\cite{Kushagra14}} as described in {\protect\cite{AumullerD15}}. Each data point is the average over 600 trials.
Times are scaled by $n \ln n$.}
\label{fig:running:times:opt:algorithms}
\end{figure}

\medskip

\emph{Question 2: How do running times of multi-pivot quicksort algorithms compare?} 
The time measurements of our experiments based on $k$-pivot quicksort algorithms
generated automatically from Algorithm~\ref{algo:exchange:k} are shown in
Figure~\ref{fig:running:times:k:pivot:quicksort}.  With respect to average
running time, we see that the variants using four and five pivots are
fastest. On average the $5$-pivot algorithm is about $1\%$ slower
than the $4$-pivot algorithm. However, there is no significant difference in
running time between these two variants.  On average, the $3$-pivot and
$2$-pivot algorithm are $3.5\%$ and $3.7\%$ slower than the $4$-pivot quicksort
algorithm. The $6$- and $7$-pivot algorithms are about $5.0\%$ slower. It
follows the $8$-pivot algorithm ($6.5\%$ slower), the $9$-pivot algorithm
($8.5\%$ slower), classical quicksort ($11.0\%$ slower), and the $15$-pivot
algorithm ($15.5\%$ slower). If we only consider significant differences in running
time, these figures must be decreased by about $1$--$2\%$. These results are
in line with our study of scanned elements in Section~\ref{sec:assignments}, see
in particular the second column of Table~\ref{tab:cache:misses}. This shows that
the cost measure ``scanned elements'' describes observed running time
differences very well. For a larger
number of pivots, additional work, e.g., finding and sorting the pivots, seems to have a
noticeable influence on running time. For example, according to the average number of
scanned elements the $15$-pivot quicksort should not be slower than classical quicksort.

\medskip

\emph{Question 3: Is the extremal comparison tree faster than the balanced comparison tree?} We
report on an experiment regarding that compared a $7$-pivot algorithm using the extremal
comparison tree and a $7$-pivot algorithm using the balanced comparison
tree. For inputs of size $n = 2^{27}$, the $7$-pivot quicksort algorithm with the extremal
tree was at least $2\%$ faster than the $7$-pivot algorithm
with the balanced tree in $95\%$ of the runs. 
The difference in running time is thus not large, but statistically significant.

\begin{figure}
    \centering
\begin{tikzpicture}
  \begin{axis}[
    xlabel={Items [$\log_2(n)$]},
    ylabel={Time $/ n \ln n$ [ns]},
    height=7cm,
    width=12cm,
    legend style = { at = {(0.175,1.11)}, anchor=west, draw=none},
    cycle list name = black white,
    legend columns = 5
    ]
    \addplot coordinates { (21.0,4.36582) (22.0,4.36891) (23.0,4.39619) (24.0,4.40324) (25.0,4.42603) (26.0,4.44308) (27.0,4.46176) };
    \addlegendentry{Ex$_{15}$};
    \addplot coordinates { (21.0,4.25107) (22.0,4.2569) (23.0,4.26868) (24.0,4.27228) (25.0,4.28326) (26.0,4.29492) (27.0,4.30213) };
    \addlegendentry{Ex$_{1}$};
    \addplot coordinates { (21.0,4.14344) (22.0,4.14692) (23.0,4.15878) (24.0,4.16289) (25.0,4.17173) (26.0,4.17857) (27.0,4.18746) };
    \addlegendentry{Ex$_9$};
    \addplot coordinates { (21.0,4.07977) (22.0,4.08477) (23.0,4.09241) (24.0,4.0975) (25.0,4.10497) (26.0,4.10356) (27.0,4.11944) };
    \addlegendentry{Ex$_8$};
    \addplot coordinates { (21.0,4.03391) (22.0,4.03699) (23.0,4.04718) (24.0,4.04064) (25.0,4.0545) (26.0,4.05212) (27.0,4.06393) };
    \addlegendentry{Ex$_7$};
    \addplot coordinates { (21.0,4.03527) (22.0,4.02999) (23.0,4.03859) (24.0,4.04152) (25.0,4.04685) (26.0,4.05131) (27.0,4.06116) };
    \addlegendentry{Ex$_6$};
    \addplot coordinates { (21.0,3.99835) (22.0,3.9964) (23.0,4.00579) (24.0,4.00831) (25.0,4.00903) (26.0,4.01997) (27.0,4.02109) };
    \addlegendentry{Ex$_{2}$};
    \addplot coordinates { (21.0,3.9863) (22.0,3.98223) (23.0,3.98758) (24.0,3.991) (25.0,3.99537) (26.0,4.00174) (27.0,4.00186) };
    \addlegendentry{Ex$_{3}$};
    \addplot coordinates { (21.0,3.89865) (22.0,3.88803) (23.0,3.89307) (24.0,3.89127) (25.0,3.89517) (26.0,3.89596) (27.0,3.90495) };
    \addlegendentry{Ex$_5$};
    \addplot coordinates { (21.0,3.85936) (22.0,3.85102) (23.0,3.85671) (24.0,3.8599) (25.0,3.85681) (26.0,3.86185) (27.0,3.86507) };
    \addlegendentry{Ex$_{4}$};
\end{axis}
\end{tikzpicture}
\begin{tikzpicture}
  \begin{axis}[
    xlabel={\# pivots},
    ylabel={Time $/n \ln n$ [ns]},
    height=3cm,
    width=11cm,
    legend style = { at = {(-.0,1.12)}, anchor=west, draw=none},
    cycle list name = black white,
    legend columns = 3
    ]
    \addplot coordinates { (1, 4.302) (2,
        4.021) (3, 4.002) (4, 3.865) (5, 3.905) (6, 4.0611) (7, 4.064) (8, 4.119)
    (9, 4.187) (15, 4.462)};
\end{axis}
\end{tikzpicture}
\caption{Running time experiments for $k$-pivot quicksort algorithms based on
the ``Exchange$_k$'' partitioning strategy. Each data point is the average over 600 trials.
Times are scaled by $n \ln n$. Top: Running times for sorting inputs of size
$2^i$ for $21 \leq i \leq 27$. Bottom: Running times from top with respect to the number
of pivots used for sorting inputs of size $2^{27}$.}
\label{fig:running:times:k:pivot:quicksort}
\end{figure}
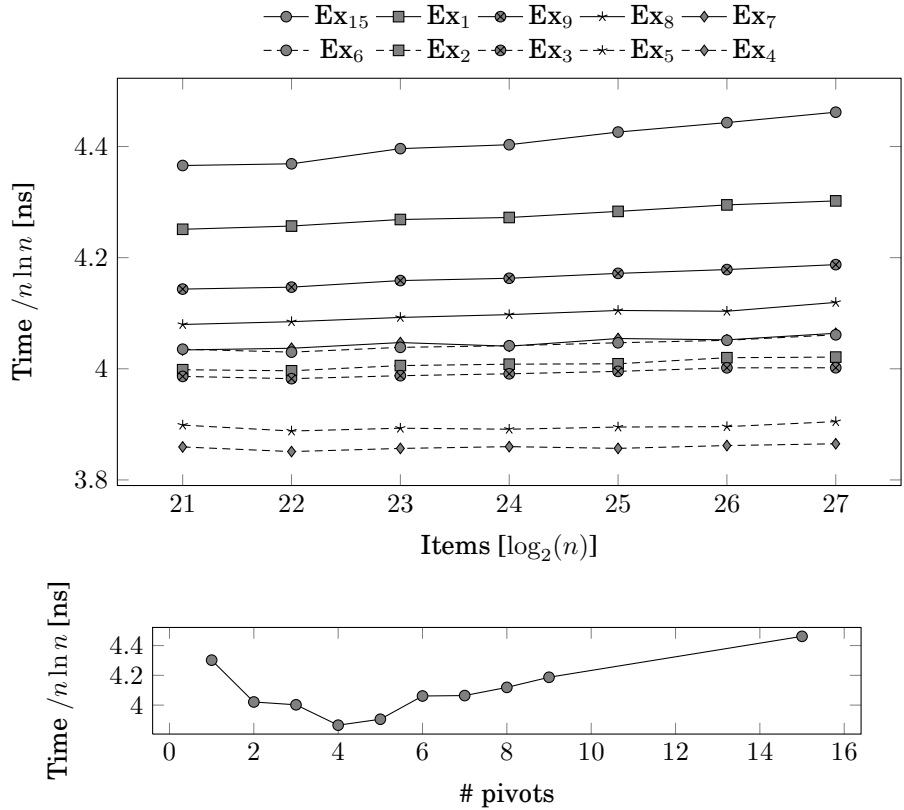

\medskip

\emph{Question 4: Is pivot sampling beneficial for running time?} Before 
reporting on the results of our experiments, we want to stress that
comparing the overhead incurred by sorting a larger sample to the benefits
of having better pivots is very difficult from a theoretical point of view because
of the influence of lower-order terms to the sorting cost. 
In these experiments we observed that sampling strategies that make the groups
closer to the center larger improved the running time more significantly than
balanced samples, which validates our findings from Section~\ref{sec:pivot:sampling}.
However, the benefits regarding empirical running time observed 
when pivots are chosen from a small sample are minimal.  Compared to
choosing pivots directly from the input, the largest improvements in running
time were observed for the $7$- and $9$-pivot algorithms. For these algorithms,
the running time could be improved by about $3\%$ by choosing pivots from a
sample of size 13 or 15. For example, Fig.~\ref{fig:running:times:7:sampling} depicts the running
times we obtained for variants using seven pivots and different sampling vectors. In
some cases, e.g., for the $4$-pivot algorithm, we could not observe any
improvements by sampling pivots.%
\footnote{Note that sampling improves the cache behavior. For example, the $4$-pivot algorithm
without sampling incurred $10\%$ more L1 cache misses, $10\%$ more L2 cache
misses, and needed $7\%$ more instructions than the $4$-pivot algorithm with
sampling vector $(0,0,1,1,0)$. Still, it was $2\%$ faster in experiments.
The reason might be that it made $7\%$ less branch mispredictions.}

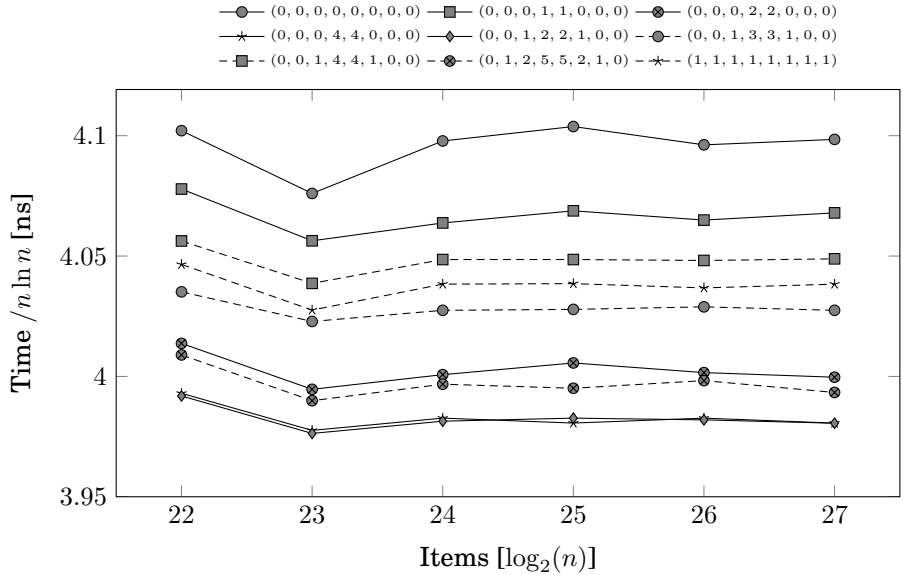
\begin{figure}
    \centering
\begin{tikzpicture}
  \begin{axis}[
    xlabel={Items [$\log_2(n)$]},
    ylabel={Time $/ n \ln n$ [ns]},
    height=7cm,
    width=12cm,
    ymin=3.95,
    legend style = { at = {(0.12,1.13)}, anchor=west, draw=none, font=\tiny},
    cycle list name = black white,
    legend columns = 3
    ]
    \addplot coordinates {  (22.0,4.10209) (23.0,4.07601) (24.0,4.09781) (25.0,4.10381) (26.0,4.0962) (27.0,4.09846) };
    \addlegendentry{$(0,0,0,0,0,0,0,0)$};
    \addplot coordinates {  (22.0,4.07785) (23.0,4.05633) (24.0,4.06375) (25.0,4.06879) (26.0,4.06495) (27.0,4.06793) };
    \addlegendentry{$(0,0,0,1,1,0,0,0)$};
    \addplot coordinates {  (22.0,4.01372) (23.0,3.99464) (24.0,4.00073) (25.0,4.00559) (26.0,4.00158) (27.0,3.99967) };
    \addlegendentry{$(0,0,0,2,2,0,0,0)$};
    \addplot coordinates {  (22.0,3.99298) (23.0,3.97758) (24.0,3.98267) (25.0,3.98062) (26.0,3.98266) (27.0,3.9806) };
    \addlegendentry{$(0,0,0,4,4,0,0,0)$};
    \addplot coordinates {  (22.0,3.99188) (23.0,3.97632) (24.0,3.9814) (25.0,3.98267) (26.0,3.98195) (27.0,3.98056) };
    \addlegendentry{$(0,0,1,2,2,1,0,0)$};
    \addplot coordinates {  (22.0,4.03513) (23.0,4.02283) (24.0,4.02746) (25.0,4.02784) (26.0,4.0289) (27.0,4.02746) };
    \addlegendentry{$(0,0,1,3,3,1,0,0)$};
    \addplot coordinates {  (22.0,4.05626) (23.0,4.03866) (24.0,4.04856) (25.0,4.04856) (26.0,4.04817) (27.0,4.04884) };
    \addlegendentry{$(0,0,1,4,4,1,0,0)$};
    \addplot coordinates {  (22.0,4.00893) (23.0,3.98991) (24.0,3.99677) (25.0,3.99508) (26.0,3.99825) (27.0,3.99334) };
    \addlegendentry{$(0,1,2,5,5,2,1,0)$};
    \addplot coordinates {  (22.0,4.04653) (23.0,4.02756) (24.0,4.03833) (25.0,4.03854) (26.0,4.03677) (27.0,4.03833) };
    \addlegendentry{$(1,1,1,1,1,1,1,1)$};
\end{axis}
\end{tikzpicture}
\caption{Running time experiments for different sampling strategies for the
$7$-pivot quicksort algorithm based on the ``Exchange$_7$'' partitioning
strategy. Each line represents the running time measurements obtained with the given
$\mathbf\tau$ sampling vector. Each data point is the average over 600 trials.  Times are scaled by
$n \ln n$.}
\label{fig:running:times:7:sampling}
\end{figure}

\subsection{Comparison with Other Methods} 
Here, we report on
experiments in which the algorithms from before were compared with other quicksort-based
algorithms known from the literature. For the comparison, we used the
\texttt{std::sort} implementation found in {\Cpp}'s standard STL (from gcc),
the YBB algorithm from \cite[Figure 4]{NebelWM15}, the two-pivot
algorithm from \cite[Algorithm 3]{AumullerD15}, the three-pivot algorithm of
\cite{Kushagra14}, see \cite[Algorithm 8]{AumullerD15}, and an 
implementation of the super scalar sample sort algorithm of 
\citeN{SandersW04} with basic source code for classifications provided by Timo Bingmann.
Algorithm \texttt{std::sort} is an introsort implementation, which combines
quicksort with a heapsort fallback when subproblem sizes decrease too slowly.
As explained in \cite{SandersW04}, using a large number of pivots in a sample
sort approach can be implemented to exploit \emph{data independence}
(classifications are decoupled from each other) and \emph{predicated move instructions}, which
reduce branch mispredictions in the classification step. See the source code at {\webpage} for details.

\begin{figure}[t!]
    \centering
\begin{tikzpicture}
  \begin{axis}[
    xlabel={Items [$\log_2(n)$]},
    ylabel={Time $/ n \ln n$ [ns]},
    height=7cm,
    width=12cm,
    legend style = { at = {(0.15,0.43)}, anchor=west, draw=none},
    cycle list name = black white,
    legend columns = 4 
    ]
    \addplot coordinates { (21.0,4.46247) (22.0,4.47105) (23.0,4.46797) (24.0,4.46997) (25.0,4.48147) (26.0,4.48226) (27.0,4.46671) };
    \addlegendentry{stdsort};
    \addplot coordinates { (21.0,4.18381) (22.0,4.17051) (23.0,4.06113) (24.0,4.07187) (25.0,4.07515) (26.0,4.09054) (27.0,4.09321) };
    \addlegendentry{YBB};
    \addplot coordinates { (21.0,4.04261) (22.0,4.05024) (23.0,4.0483) (24.0,4.06348) (25.0,4.06516) (26.0,4.06899) (27.0,4.05622) };
    \addlegendentry{$\mathcal{L}$};
    \addplot coordinates { (21.0,3.9863) (22.0,3.98223) (23.0,3.98758) (24.0,3.991) (25.0,3.99537) (26.0,4.00174) (27.0,4.00186) };
    \addlegendentry{Ex$_3$};
    \addplot coordinates { (21.0,3.93608) (22.0,3.93725) (23.0,3.93055) (24.0,3.93994) (25.0,3.94979) (26.0,3.94842) (27.0,3.93299) };
    \addlegendentry{$\mathcal{K}$};
    \addplot coordinates { (21.0,3.85936) (22.0,3.85102) (23.0,3.85671) (24.0,3.8599) (25.0,3.85681) (26.0,3.86185) (27.0,3.86507) };
    \addlegendentry{Ex$_4$};
    \addplot coordinates { (21.0,3.05932) (22.0,3.02076) (23.0,2.95387) (24.0,2.87619) (25.0,2.92197) (26.0,2.95301) (27.0,2.94284) };
    \addlegendentry{SSSS$_{127}$};
\end{axis}
\end{tikzpicture}
\caption{Running time experiments for different quicksort-based algorithms compared to
    the ``Exchange$_k$'' partitioning strategies $\text{Ex}_3$ and $\text{Ex}_4$.
    The plot includes the average running time of {\protect\Cpp}'s \protect\texttt{std::sort}
    implementation (``stdsort''), the YBB algorithm (``YBB''),
    the dual-pivot quicksort algorithm from \protect\cite[Algorithm 3]{AumullerD15} (``$\mathcal{L}$''),
    the three-pivot quicksort algorithm from Kushagra et al.
    (``$\mathcal{K}$''),  and the super-scalar sample sort
    variant of \protect\cite{SandersW04} using 127 pivots (``SSSS$_{127}$'').
    Each data point is the average over 600 trials.
    Times are scaled by $n \ln n$.}
\label{fig:running:times:other}
\end{figure}
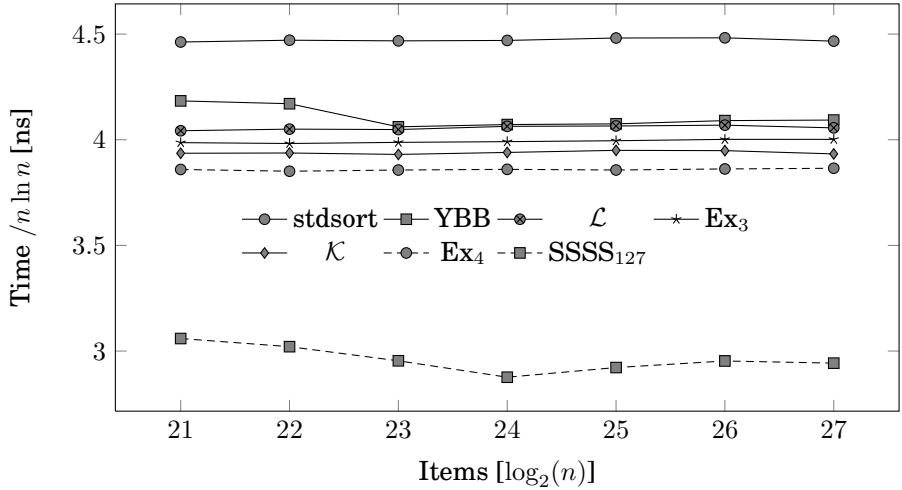

Figure~\ref{fig:running:times:other} shows the measurements we got for these
algorithms.\footnote{We used \emph{gcc} version 4.9 to compile the YBB 
    algorithm and the super scalar sample sort algorithm. For both algorithms, the
    executable obtained by compiling with \emph{clang} was much slower. For
the YBB algorithm, compiling with the compiler flag \emph{-O3} was
fastest, while for SSSS flags \emph{-O3
-funroll-loops} was fastest.} We see that \texttt{std::sort} is by far the
slowest algorithm. Of course, it is slowed down by special precautions for
inputs that are not random or have equal entries. (Such precautions have not
been taken in the other implementations.) Next come the two dual-pivot
quicksort algorithms.  The three-pivot algorithm of \cite{Kushagra14} is a
little bit faster than the automatically generated three-pivot quicksort
algorithm, but a little slower than the automatically generated four-pivot
quicksort algorithm. 
(This shows that the automatically generated quicksort algorithms
do not suffer in performance compared to manually tuned implementations such as
the three-pivot algorithm from \cite{Kushagra14}.)
The implementation of the super scalar sample sort algorithm using 127 pivots of Sanders and
Winkel is by far the fastest algorithms. However, it is only faster if 
classifications are stored in the first pass and no re-classifications take place 
in the second pass. Using 255 pivots is
slightly slower, using 511 is much slower. 

In conclusion, the experiments validated our theoretical results from before.
It is interesting to see that the memory-related cost measure ``scanned
elements'' comes closest to explain running times. It seems unlikely that algorithms based 
on the partitioning method from Algorithm~\ref{algo:exchange:k} provide drastic improvements 
in running time when
used with more than five pivots. For sorting random permutations,
the benefits of pivot sampling seem only minimal. When a lot of space and
good compiler/architecture support is available, drastic improvements in running
time are possible using two-pass quicksort variants such as the super scalar
sample sort algorithm of \citeN{SandersW04}. During the process of reviewing
this paper, a variant of a single-pivot quicksort algorithm with running time comparable 
to super scalar sample sort but without additional memory has been proposed by
\citeN{EdelkampW16}. 
As in super scalar sample sort, the running time improvement comes from a trick
to avoid branches.

\section{Conclusion}
In this paper we studied the design space of multi-pivot quicksort algorithms
and demonstrated how to analyze their sorting cost. In the first part we showed
how to calculate the average comparison count of an arbitrary multi-pivot
quicksort algorithm.  We described optimal multi-pivot quicksort algorithms
with regard to key comparisons. It turned out that calculating their average
comparison count seems difficult already for four pivots (we resorted to
experiments) and that they offer nor drastic improvement  
with respect to comparisons with respect to other methods. Similar improvements in key comparisons can be achieved by much
simpler strategies such as the median-of-$k$ strategy for classical quicksort. From
a running time point of view, comparison-optimal $k$-pivot
quicksort is not competitive 
in running time to simpler quicksort variants that, e.g., use some fixed
comparison tree, even when we bypass the computation of an optimal comparison
tree.
In the second part we switched our viewpoint and studied the problem of
rearranging entries to obtain a partition of the input.  We analyzed a one-pass
algorithm (Algorithm~\ref{algo:exchange:k}) with respect to the cost measures
``scanned elements'', ``write accesses'', and ``assignments''. For the second
and third cost measure, we found that the cost increases with an increasing
number of pivots.  The first cost measure turned out to be more interesting.
Using Algorithm~\ref{algo:exchange:k} with three or five pivots was
particularly good, and we asserted that ``scanned elements'' correspond to L1
cache misses in practice. Experiments revealed that there is a high correlation
to observed running times.  In the last part of this paper, we discussed the influence
of pivot sampling to sorting cost.  With respect to comparisons we showed that
for every comparison tree we can describe a pivot choice such that the cost is
optimal. With regard to the number of scanned elements, we noticed that pivots
should be chosen such that they make groups closer to the center larger.  To
determine element groups, the extremal comparison tree should be used. We
conjectured that this choice of tree is also best when we consider as cost the
sum of comparisons and scanned elements.

For future work, it would be very interesting to see how the optimal average
comparison count of $k$-pivot quicksort can be calculated analytically, cf.
\eqref{eq:optimal:cost:formula}. The situation seems to be much more complicated 
than for dual-pivot quicksort, for which \citeN{AumullerDHKP16} were able to obtain the 
\emph{exact} average comparison count in very recent work. With respect to the rearrangement problem, 
it would be interesting to identify an optimal algorithm that shifts 
as few elements as possible to rearrange an input. With respect to pivot
sampling, we conjecture that
the extremal tree minimizes the sorting cost when counting comparisons and 
scanned elements. This is true for up to nine pivots, but it remains open
to prove it for general values of $k$.
From an empirical point
of view, it would be interesting to test multi-pivot quicksort algorithms 
on input distributions that are 
not random permutations, but have some kind of ``presortedness'' or contain
equal keys. Very recent work \cite{EdelkampW16} showed that our variants
perform worse on inputs containing equal elements than other standard approaches. 
How to deal with these elements efficiently should be investigated in a separate 
study on the topic. Initial experiments for two and three pivots showed that 
running times greatly improve when we make sure that whenever two pivots are the
same, we put all elements equal to these pivots between them and avoid the 
recursive call. However, there any many ways of doing this and it would be
very interesting to see what turns out to be most efficient both in theory and in
practice. Again with respect to running times, a key technique in algorithm
engineering is to avoid branch mispredictions. Techniques that do this 
have been successfully applied in super scalar sample sort \cite{SandersW04}, 
tuned quicksort \cite{ElmasryKS12}, and the recent ``Block Quicksort'' algorithm
\cite{EdelkampW16}. It would be interesting to see how one can apply 
these techniques for multi-pivot quicksort for a small number of two to five
pivots.

\section*{Acknowledgements}
We thank Sebastian Wild for many interesting discussions and his valuable
comments on an earlier version of this manuscript. In particular,
the general partitioning algorithm from Section~\ref{sec:assignments} evolved from
discussions with him at Dagstuhl Seminar 14091. 
We also thank Conrado Martínez and  Markus Nebel
for interesting discussions and comments on multi-pivot quicksort. 
In addition, we thank Timo Bingmann for providing source 
code.  We thank the referees of this submission for their
insightful comments, which helped a lot in improving the presentation and 
focusing on the main insights.

\bibliographystyle{ACM-Reference-Format-Journals}
\bibliography{lit}


\begin{thebibliography}{00}


\ifx \showCODEN    \undefined \def \showCODEN     #1{\unskip}     \fi
\ifx \showDOI      \undefined \def \showDOI       #1{{\tt DOI:}\penalty0{#1}\ }
  \fi
\ifx \showISBNx    \undefined \def \showISBNx     #1{\unskip}     \fi
\ifx \showISBNxiii \undefined \def \showISBNxiii  #1{\unskip}     \fi
\ifx \showISSN     \undefined \def \showISSN      #1{\unskip}     \fi
\ifx \showLCCN     \undefined \def \showLCCN      #1{\unskip}     \fi
\ifx \shownote     \undefined \def \shownote      #1{#1}          \fi
\ifx \showarticletitle \undefined \def \showarticletitle #1{#1}   \fi
\ifx \showURL      \undefined \def \showURL       #1{#1}          \fi

\bibitem[\protect\citeauthoryear{Aum\"{u}ller and Dietzfelbinger}{Aum\"{u}ller
  and Dietzfelbinger}{2013}]%
        {AumullerD13}
{Martin Aum\"{u}ller} {and} {Martin Dietzfelbinger}. 2013.
\newblock \showarticletitle{Optimal Partitioning for Dual Pivot Quicksort}. In
  {\em Proc. of the 40th International Colloquium on Automata, Languages and
  Programming ({ICALP}'13)}. Springer, 33--44.
\newblock


\bibitem[\protect\citeauthoryear{Aum{\"u}ller and Dietzfelbinger}{Aum{\"u}ller
  and Dietzfelbinger}{2015}]%
        {AumullerD15}
{Martin Aum{\"u}ller} {and} {Martin Dietzfelbinger}. 2015.
\newblock \showarticletitle{Optimal Partitioning for Dual-Pivot Quicksort}.
\newblock  (2015).
\newblock
\newblock
\shownote{To appear in ACM Transactions on Algorithms, preprint available at
  \url{http://arxiv.org/abs/1303.5217}.}


\bibitem[\protect\citeauthoryear{Aum{\"{u}}ller, Dietzfelbinger, Heuberger,
  Krenn, and Prodinger}{Aum{\"{u}}ller et~al\mbox{.}}{2016}]%
        {AumullerDHKP16}
{Martin Aum{\"{u}}ller}, {Martin Dietzfelbinger}, {Clemens Heuberger}, {Daniel
  Krenn}, {and} {Helmut Prodinger}. 2016.
\newblock \showarticletitle{Counting Zeros in Random Walks on the Integers and
  Analysis of Optimal Dual-Pivot Quicksort}.
\newblock {\em CoRR\/}  {abs/1602.04031} (2016).
\newblock
\showURL{%
\url{http://arxiv.org/abs/1602.04031}}
\newblock
\shownote{To appear in AofA'16.}


\bibitem[\protect\citeauthoryear{Aumüller}{Aumüller}{2015}]%
        {Aumueller15}
{Martin Aumüller}. 2015.
\newblock {\em On the Analysis of Two Fundamental Randomized Algorithms:
  Multi-Pivot Quicksort and Efficient Hash Functions}.
\newblock Ph.D. Dissertation. Technische Universität Ilmenau.
\newblock


\bibitem[\protect\citeauthoryear{Bentley and McIlroy}{Bentley and
  McIlroy}{1993}]%
        {BentleyM93}
{Jon~L. Bentley} {and} {M.~Douglas McIlroy}. 1993.
\newblock \showarticletitle{Engineering a sort function}.
\newblock {\em Software: Practice and Experience\/} {23}, 11 (1993),
  1249--1265.
\newblock


\bibitem[\protect\citeauthoryear{Bentley and Sedgewick}{Bentley and
  Sedgewick}{1997}]%
        {BentleyS97}
{Jon~Louis Bentley} {and} {Robert Sedgewick}. 1997.
\newblock \showarticletitle{Fast Algorithms for Sorting and Searching Strings}.
  In {\em Proc. of the Eighth Annual {ACM-SIAM} Symposium on Discrete
  Algorithms ({SODA}'97).} ACM, 360--369.
\newblock


\bibitem[\protect\citeauthoryear{Bloch}{Bloch}{2015}]%
        {BlochPersonalCom}
{Joshua Bloch}. 2015.
\newblock  (2015).
\newblock
\newblock
\shownote{Personal Communication.}


\bibitem[\protect\citeauthoryear{Dijkstra}{Dijkstra}{1976}]%
        {Dijkstra76}
{Edsger~W. Dijkstra}. 1976.
\newblock {\em A Discipline of Programming}.
\newblock Prentice-Hall.
\newblock


\bibitem[\protect\citeauthoryear{Dubhashi and Panconesi}{Dubhashi and
  Panconesi}{2009}]%
        {dp09}
{Devdatt~P. Dubhashi} {and} {Alessandro Panconesi}. 2009.
\newblock {\em Concentration of Measure for the Analysis of Randomized
  Algorithms}.
\newblock Cambridge University Press.
\newblock
\showISBNx{978-0-521-88427-3}


\bibitem[\protect\citeauthoryear{Edelkamp and Wei{\ss}}{Edelkamp and
  Wei{\ss}}{2016}]%
        {EdelkampW16}
{Stefan Edelkamp} {and} {Armin Wei{\ss}}. 2016.
\newblock \showarticletitle{BlockQuicksort: How Branch Mispredictions don't
  affect Quicksort}.
\newblock {\em CoRR\/}  {abs/1604.06697} (2016).
\newblock
\showURL{%
\url{http://arxiv.org/abs/1604.06697}}


\bibitem[\protect\citeauthoryear{Elmasry, Katajainen, and Stenmark}{Elmasry
  et~al\mbox{.}}{2012}]%
        {ElmasryKS12}
{Amr Elmasry}, {Jyrki Katajainen}, {and} {Max Stenmark}. 2012.
\newblock \showarticletitle{Branch Mispredictions Don't Affect Mergesort}. In
  {\em Experimental Algorithms - 11th International Symposium, {SEA} 2012,
  Bordeaux, France, June 7-9, 2012. Proceedings}. 160--171.
\newblock
\showDOI{%
\url{http://dx.doi.org/10.1007/978-3-642-30850-5_15}}


\bibitem[\protect\citeauthoryear{Fog}{Fog}{2014}]%
        {Fog14}
{Agner Fog}. 2014.
\newblock 4. Instruction tables.
\newblock \url{http://www.agner.org/optimize/instruction_tables.pdf}.   (2014).
\newblock


\bibitem[\protect\citeauthoryear{Frazer and McKellar}{Frazer and
  McKellar}{1970}]%
        {FrazerK70}
{W.~D. Frazer} {and} {A.~C. McKellar}. 1970.
\newblock \showarticletitle{Samplesort: A Sampling Approach to Minimal Storage
  Tree Sorting}.
\newblock {\em J. ACM\/} {17}, 3 (July 1970), 496--507.
\newblock
\showISSN{0004-5411}
\showDOI{%
\url{http://dx.doi.org/10.1145/321592.321600}}


\bibitem[\protect\citeauthoryear{Garsia and Wachs}{Garsia and Wachs}{1977}]%
        {GarsiaW77}
{Adriano~M. Garsia} {and} {Michelle~L. Wachs}. 1977.
\newblock \showarticletitle{A New Algorithm for Minimum Cost Binary Trees}.
\newblock {\em {SIAM} J. Comput.\/} {6}, 4 (1977), 622--642.
\newblock
\showDOI{%
\url{http://dx.doi.org/10.1137/0206045}}


\bibitem[\protect\citeauthoryear{Graham, Knuth, and Patashnik}{Graham
  et~al\mbox{.}}{1994}]%
        {GrahamKP}
{Ronald~L. Graham}, {Donald~E. Knuth}, {and} {Oren Patashnik}. 1994.
\newblock {\em Concrete mathematics - a foundation for computer science {(2.
  ed.)}}.
\newblock Addison-Wesley.
\newblock
\showISBNx{978-0-201-55802-9}


\bibitem[\protect\citeauthoryear{Hennequin}{Hennequin}{1991}]%
        {hennequin}
{Pascal Hennequin}. 1991.
\newblock {\em Analyse en moyenne d'algorithmes: tri rapide et arbres de
  recherche}.
\newblock Ph.D. Dissertation. Ecole Politechnique, Palaiseau.
\newblock


\bibitem[\protect\citeauthoryear{Hoare}{Hoare}{1962}]%
        {Hoare62}
{C.~A.~R. Hoare}. 1962.
\newblock \showarticletitle{Quicksort}.
\newblock {\em Comput. J.\/} {5}, 1 (1962), 10--15.
\newblock


\bibitem[\protect\citeauthoryear{Hu and Tucker}{Hu and Tucker}{1971}]%
        {HuTu71}
{T.~C. Hu} {and} {A.~C. Tucker}. 1971.
\newblock \showarticletitle{Optimal Computer Search Trees and Variable-Length
  Alphabetical Codes}.
\newblock {\it SIAM J. Appl. Math.} {21}, 4 (1971), 514--532.
\newblock


\bibitem[\protect\citeauthoryear{Iliopoulos}{Iliopoulos}{2014}]%
        {Iliopoulos14}
{Vasileios Iliopoulos}. 2014.
\newblock \showarticletitle{A note on multipivot Quicksort}.
\newblock {\em CoRR\/}  {abs/1407.7459} (2014).
\newblock


\bibitem[\protect\citeauthoryear{Kaligosi and Sanders}{Kaligosi and
  Sanders}{2006}]%
        {KaligosiS06}
{Kanela Kaligosi} {and} {Peter Sanders}. 2006.
\newblock \showarticletitle{How Branch Mispredictions Affect Quicksort}. In
  {\em Proc. of the 14th Annual European Symposium on Algorithms ({ESA}'06)}.
  Springer, 780--791.
\newblock


\bibitem[\protect\citeauthoryear{Knuth}{Knuth}{1973}]%
        {Knuth}
{Donald~E. Knuth}. 1973.
\newblock {\em The Art of Computer Programming, Volume III: Sorting and
  Searching}.
\newblock Addison-Wesley.
\newblock
\showISBNx{0-201-03803-X}


\bibitem[\protect\citeauthoryear{Kushagra, L{\'o}pez-Ortiz, Qiao, and
  Munro}{Kushagra et~al\mbox{.}}{2014}]%
        {Kushagra14}
{Shrinu Kushagra}, {Alejandro L{\'o}pez-Ortiz}, {Aurick Qiao}, {and} {J.~Ian
  Munro}. 2014.
\newblock \showarticletitle{Multi-Pivot Quicksort: Theory and Experiments}. In
  {\em Proc. of the 16th Meeting on Algorithms Engineering and Experiments
  ({ALENEX}'14)}. SIAM, 47--60.
\newblock


\bibitem[\protect\citeauthoryear{LaMarca and Ladner}{LaMarca and
  Ladner}{1999}]%
        {LaMarcaL99}
{Anthony LaMarca} {and} {Richard~E. Ladner}. 1999.
\newblock \showarticletitle{The Influence of Caches on the Performance of
  Sorting}.
\newblock {\em J. Algorithms\/} {31}, 1 (1999), 66--104.
\newblock


\bibitem[\protect\citeauthoryear{Leischner, Osipov, and Sanders}{Leischner
  et~al\mbox{.}}{2010}]%
        {LeischnerOS10}
{Nikolaj Leischner}, {Vitaly Osipov}, {and} {Peter Sanders}. 2010.
\newblock \showarticletitle{{GPU} sample sort}. In {\em 24th {IEEE}
  International Symposium on Parallel and Distributed Processing, {IPDPS} 2010,
  Atlanta, Georgia, USA, 19-23 April 2010 - Conference Proceedings}. 1--10.
\newblock
\showDOI{%
\url{http://dx.doi.org/10.1109/IPDPS.2010.5470444}}


\bibitem[\protect\citeauthoryear{Levinthal}{Levinthal}{2009}]%
        {Levinthal09}
{David Levinthal}. 2009.
\newblock Performance Analysis Guide for Intel Core i7 Processor and Intel Xeon
  5500 processors.
\newblock
  \url{https://software.intel.com/sites/products/collateral/hpc/vtune/performance_analysis_guide.pdf}.
    (2009).
\newblock


\bibitem[\protect\citeauthoryear{Mart\'{\i}nez, Nebel, and Wild}{Mart\'{\i}nez
  et~al\mbox{.}}{2015}]%
        {MartinezNW15}
{Conrado Mart\'{\i}nez}, {Markus~E. Nebel}, {and} {Sebastian Wild}. 2015.
\newblock \showarticletitle{Analysis of Branch Misses in Quicksort}. In {\em
  Proc. of the 12th Meeting on Analytic Algorithmics and Combinatorics
  ({ANALCO}' 15)}.
\newblock


\bibitem[\protect\citeauthoryear{Mart\'{\i}nez and Roura}{Mart\'{\i}nez and
  Roura}{2001}]%
        {MartinezR01}
{Conrado Mart\'{\i}nez} {and} {Salvador Roura}. 2001.
\newblock \showarticletitle{Optimal Sampling Strategies in Quicksort and
  Quickselect}.
\newblock {\em SIAM J. Comput.\/} {31}, 3 (2001), 683--705.
\newblock


\bibitem[\protect\citeauthoryear{McIlroy, Bostic, and McIlroy}{McIlroy
  et~al\mbox{.}}{1993}]%
        {McIlroyBM93}
{Peter~M. McIlroy}, {Keith Bostic}, {and} {M.~Douglas McIlroy}. 1993.
\newblock \showarticletitle{Engineering Radix Sort}.
\newblock {\em Computing Systems\/} {6}, 1 (1993), 5--27.
\newblock


\bibitem[\protect\citeauthoryear{Nebel, Wild, and Martínez}{Nebel
  et~al\mbox{.}}{2015}]%
        {NebelWM15}
{Markus~E. Nebel}, {Sebastian Wild}, {and} {Conrado Martínez}. 2015.
\newblock \showarticletitle{Analysis of Pivot Sampling in Dual-Pivot Quicksort:
  A Holistic Analysis of Yaroslavskiy’s Partitioning Scheme}.
\newblock {\em Algorithmica\/} (2015), 1--52.
\newblock
\showISSN{0178-4617}


\bibitem[\protect\citeauthoryear{Rahman}{Rahman}{2002}]%
        {Rahman02}
{Naila Rahman}. 2002.
\newblock \showarticletitle{Algorithms for Hardware Caches and {TLB}}. In {\em
  Algorithms for Memory Hierarchies, Advanced Lectures [Dagstuhl Research
  Seminar, March 10-14, 2002]}. 171--192.
\newblock


\bibitem[\protect\citeauthoryear{Roura}{Roura}{2001}]%
        {Roura01}
{Salvador Roura}. 2001.
\newblock \showarticletitle{Improved master theorems for divide-and-conquer
  recurrences}.
\newblock {\em J. ACM\/} {48}, 2 (2001), 170--205.
\newblock


\bibitem[\protect\citeauthoryear{Sanders and Winkel}{Sanders and
  Winkel}{2004}]%
        {SandersW04}
{Peter Sanders} {and} {Sebastian Winkel}. 2004.
\newblock \showarticletitle{Super Scalar Sample Sort}. In {\em Proc. of the
  12th Annual European Symposium on Algorithms ({ESA}'04)}. Springer, 784--796.
\newblock


\bibitem[\protect\citeauthoryear{Sedgewick}{Sedgewick}{1975}]%
        {sedgewick}
{Robert Sedgewick}. 1975.
\newblock {\em Quicksort}.
\newblock Ph.D. Dissertation. Standford University.
\newblock


\bibitem[\protect\citeauthoryear{Sedgewick}{Sedgewick}{1977}]%
        {SedgewickEqual}
{Robert Sedgewick}. 1977.
\newblock \showarticletitle{Quicksort with Equal Keys}.
\newblock {\em SIAM J. Comput.\/} {6}, 2 (1977), 240--268.
\newblock


\bibitem[\protect\citeauthoryear{Tan}{Tan}{1993}]%
        {Tan93}
{Kok-Hooi Tan}. 1993.
\newblock {\em An asymptotic analysis of the number of comparisons in
  multipartition quicksort}.
\newblock Ph.D. Dissertation. Carnegie Mellon University.
\newblock


\bibitem[\protect\citeauthoryear{van Emden}{van Emden}{1970}]%
        {vanEmden}
{M.~H. van Emden}. 1970.
\newblock \showarticletitle{Increasing the efficiency of quicksort}.
\newblock {\em Commun. ACM\/} {13}, 9 (Sept. 1970), 563--567.
\newblock
\showISSN{0001-0782}


\bibitem[\protect\citeauthoryear{Wild and Nebel}{Wild and Nebel}{2012}]%
        {nebel12}
{Sebastian Wild} {and} {Markus~E. Nebel}. 2012.
\newblock \showarticletitle{Average case analysis of Java 7's dual pivot
  quicksort}. In {\em Proc. of the 20th European Symposium on Algorithms
  ({ESA}'12)}. 825--836.
\newblock


\bibitem[\protect\citeauthoryear{Wild, Nebel, and Neininger}{Wild
  et~al\mbox{.}}{2015}]%
        {WildNN15}
{Sebastian Wild}, {Markus~E. Nebel}, {and} {Ralph Neininger}. 2015.
\newblock \showarticletitle{Average Case and Distributional Analysis of
  Dual-Pivot Quicksort}.
\newblock {\em {ACM} Transactions on Algorithms\/} {11}, 3 (2015), 22.
\newblock


\bibitem[\protect\citeauthoryear{Yaroslavskiy}{Yaroslavskiy}{2009}]%
        {Mailingliste}
{Vladimir Yaroslavskiy}. 2009.
\newblock   (2009).
\newblock
\showURL{%
\url{http://permalink.gmane.org/gmane.comp.java.openjdk.core-libs.devel/2628}}


\end{thebibliography}
\appendix

\section{Solving the General Quicksort
Recurrence}\label{app:sec:recurrence:solution}
In Section~\ref{sec:setup} we have shown that the sorting
cost of $k$-pivot quicksort follows the recurrence:
\begin{align*}
    \E(C_n) = \E(P_n) + \frac{1}{\binom{n}{k}} \sum_{i = 0}^{n-k} (k+1)
    \binom{n-i-1}{k-1} \E(C_i).
\end{align*}
We will use the continuous Master theorem of \citeN{Roura01} to solve this
recurrence. For completeness, we give the CMT below:

\begin{theorem}[{\cite[Theorem 18]{MartinezR01}}]
    Let $F_n$ be recursively defined by
    \begin{align*}
        F_n = \begin{cases}
            b_n, & \text{for } 0 \leq n < N,\\
            t_n + \sum_{j=0}^{n-1}w_{n,j}F_j, & \text{for } n \geq N,
        \end{cases}
    \end{align*}
    where the toll function $t_n$ satisfies $t_n \sim K n^\alpha \log^\beta(n)$ as $n \rightarrow \infty$
    for constants $K \neq 0, \alpha \geq 0, \beta > -1$. Assume there exists a function $w:[0,1] \rightarrow \mathbb{R}$
    such that
    \begin{align}\label{eq:cmt:shape:function}
        \sum_{j = 0}^{n-1} \left\vert w_{n,j} - \int_{j/n}^{(j+1)/n} w(z) \text{ d}z \right\vert = O(n^{-d}),
    \end{align}
    for a constant $d > 0$. Let $H:=1 - \int_0^1 z^\alpha w(z) \text{ d}z$. Then we have the
    following cases:\footnote{Let $f(n)$ and $g(n)$ be two functions. We write $f(n) \sim g(n)$ if
    $f(n) = g(n) + o(g(n))$. If $f(n) \sim g(n)$, we say that \emph{``$f$ and $g$ are asymptotically equivalent.''}}
    \begin{enumerate}
        \item If $H > 0$, then $F_n \sim t_n/H.$
        \item If $H = 0$, then $F_n \sim (t_n \ln n)/\hat{H}$, where
            \begin{align*}
                \hat{H} := - (\beta + 1) \int_0^1 z^\alpha \ln (z) w(z) \text{ d}z.
            \end{align*}
        \item If $H < 0$, then $F_n \sim \Theta(n^c)$ for the unique $c \in \mathbb{R}$ with
            \begin{align*}
                \int_0^1 z^c w(z) \text{ d}z = 1.
            \end{align*}
    \end{enumerate}
    \label{thm:CMT}
\end{theorem}

\begin{theorem}
    Let $\mathcal{A}$ be a $k$-pivot quicksort algorithm which has for each
    subarray of length $n$ partitioning
    cost $\E(P_n) = a \cdot n + O(n^{1-\varepsilon})$. Then
    \begin{align*}
        \E(C_n) = \frac{1}{H_{k+1} - 1} a n \ln n + O(n),
    \end{align*}
    where $H_{k+1} = \sum_{i = 1}^{k+1} (1/i)$ is the $(k + 1)$st harmonic
    number.
\end{theorem}
\begin{proof}
    By linearity of expectation we may obtain a solution for the recurrence for toll function $t_{1,n} = a \cdot n$
    and toll function $t_{2, n} = K \cdot n^{1-\varepsilon}$ separately and add the solutions.

    For toll function $t_{1,n}$, we use the result of
    Hennequin~\citeyear[Proposition III.9]{hennequin} that says that for
    partitioning cost
    $a\cdot n + O(1)$ we get sorting cost $\E(C_{1, n}) = \frac{a}{H_{k + 1} - 1} n \ln n + O(n)$.

    For $t_{2,n}$, we apply the CMT as follows.
    First, observe that Recurrence \eqref{eq:k:pivot:recurrence} has weight
    \begin{align*}
        w_{n,j} = \frac{(k+1)\cdot k \cdot (n - j - 1) \cdot \ldots \cdot (n - j - k + 1)}{n
        \cdot (n-1) \cdot \ldots \cdot (n - k + 1)}.
    \end{align*}
    We define the shape function $w(z)$ as suggested in \cite{Roura01} by
    \begin{align*}
        w(z) &= \lim_{n \rightarrow \infty} n \cdot w_{n, zn} = (k + 1)k (1 -
        z)^{k - 1}.
    \end{align*}
    Using the Binomial theorem for asymptotic bounds, we note that for all $z \in [0,1]$:
    \begin{align*}
        \left\vert n \cdot w_{n, zn} - w(z)\right\vert&\leq k \cdot (k+1) \cdot
        \left \vert \frac{(n - zn - 1)^{k-1}}{(n-k)^{k-1}} - (1-z)^{k-1} \right
        \vert\\
        &= k \cdot (k+1) \cdot
        \left \vert \left(\frac{n(1-z)}{n-k} - \frac{1}{n-k}
        \right)^{k-1} - (1-z)^{k-1} \right
        \vert\\
        &\leq k \cdot (k+1) \cdot
        \left \vert \left(\frac{n(1-z)}{n-k}
        \right)^{k-1} - (1-z)^{k-1} + O\left(n^{-1}\right) \right
        \vert\\
        &\leq k \cdot (k+1) \cdot
        \left \vert (1-z)^{k-1} \cdot \left(\frac{1}{\left(1-\frac{k}{n}\right)^{k-1}}
        - 1 \right)  + O\left(n^{-1}\right) \right
        \vert\\
        &\leq k \cdot (k+1) \cdot
        \left \vert (1-z)^{k-1} \cdot \left(\frac{1}{1 - O\left(n^{-1}\right)}
        - 1 \right)  + O\left(n^{-1}\right) \right
        \vert\\
        &\leq k \cdot (k+1) \cdot
        \left \vert (1-z)^{k-1} \cdot O(n^{-1})
        + O\left(n^{-1}\right) \right
        \vert
        =O(n^{-1}).
    \end{align*}
    Now we have to check \eqref{eq:cmt:shape:function} to see whether the shape
    function is suitable. We calculate:
    \begin{align*}
        &\phantom{=} \sum_{j = 0}^{n-1} \left\vert w_{n,j} - \int_{j/n}^{(j+1)/n} w(z) \text{ d}z \right\vert\\
           & =  \sum_{j = 0}^{n-1} \left\vert \int_{j/n}^{(j + 1)/n} n \cdot w_{n,j}
        - w(z) \text{ d}z \right\vert\\
        &\leq  \sum_{j = 0}^{n-1} \frac{1}{n}\max_{z \in [j/n,
(j+1)/n]}\left\vert   n \cdot w_{n,j}
- w(z) \right\vert\\
&\leq  \sum_{j = 0}^{n-1} \frac{1}{n}\left(\max_{z \in [j/n,
(j+1)/n]}\left\vert  w(j/n)
- w(z) \right\vert + O\left(n^{-1}\right) \right)\\
&\leq  \sum_{j = 0}^{n-1} \frac{1}{n}\left(\max_{|z - z'| \leq
    1/n}\left\vert  w(z)
- w(z') \right\vert + O\left(n^{-1}\right) \right)\\
&\leq  \sum_{j = 0}^{n-1} \frac{k(k+1)}{n}\left(\max_{|z - z'| \leq
    1/n}\left\vert  (1-z)^{k-1}
- (1-z - 1/n)^{k-1} \right\vert + O\left(n^{-1}\right) \right)
\leq \sum_{j = 0}^{n-1} O\left(n^{-2}\right) = O\left(n^{-1}\right),
    \end{align*}
    where we again used the Binomial theorem in the last two lines.

    Thus, $w$ is a suitable shape function. Using partial integration, we see that
    \begin{align*}
        H := 1 - k(k+1) \int_0^1 z^{1-\varepsilon} (1 - z)^{k-1} \text{ d}z < 0.
    \end{align*}
    Thus, the third case of the CMT applies. Again using partial integration, we check that
    \begin{align*}
        k(k+1) \int_0^1 z (1 - z)^{k-1} \text{ d}z = 1,
    \end{align*}
    so we conclude that for toll function $t_{2,n}$ the recurrence has solution $\E(C_{2,n}) = O(n)$.
    The theorem follows by adding $\E(C_{1,n})$ and $\E(C_{2,n})$.
\end{proof}

\end{document}